\documentclass[%
 reprint,
superscriptaddress,
 amsmath,amssymb,
 aps,
]{revtex4-2}

\usepackage{amsthm}
\usepackage{graphicx}% Include figure files
\usepackage{dcolumn}% Align table columns on decimal point
\usepackage{bm}% bold math
\usepackage[colorlinks=true, linkcolor=blue, citecolor=blue, urlcolor=blue]{hyperref}
\usepackage{dsfont}
\usepackage[caption=false]{subfig}
\usepackage{multirow}
\usepackage{tikz}
\usepackage{kotex}
\usepackage{makecell}

\newcommand{\ket}[1]{|#1\rangle}

\newtheorem{theorem}{Theorem}

\newtheorem{lemma}{Lemma}

\newtheorem{definition}{Definition}[section]

\DeclareRobustCommand{\orcidicon}{%
	\begin{tikzpicture}
	\draw[lime, fill=lime] (0,0) 
	circle [radius=0.16] 
	node[white] {{\fontfamily{qag}\selectfont \tiny ID}};
	\draw[white, fill=white] (-0.0625,0.095) 
	circle [radius=0.007];
	\end{tikzpicture}
	\hspace{-2mm}
}
\foreach \x in {A, ..., Z}{%
	\expandafter\xdef\csname orcid\x\endcsname{\noexpand\href{https://orcid.org/\csname orcidauthor\x\endcsname}{\noexpand\orcidicon}}
}
\begin{document}

\preprint{APS/123-QED}

\title{Reducing Circuit Depth in Quantum State Preparation for Quantum Simulation\\Using Measurements and Feedforward}
% Constant-Depth Quantum Circuits for Single Antisymmetric State \\and Bethe Ansatz Preparation via Measurements and Feedforward

\author{Hyeonjun Yeo\orcidA{}}
\email{duguswns11@snu.ac.kr}
\affiliation{Department of Physics and Astronomy, Seoul National University, Seoul 08826, Korea}%Lines break automatically or can be forced with \\
\author{Ha Eum Kim\orcidB{}}
%\email{hamkim0114@gmail.com}
\affiliation{Department of Mathematics, Kyung Hee University, Seoul 02447, Korea}%
\affiliation{Department of Physics, University of Illinois at Urbana-Champaign, Urbana, Illinois 61801, USA}
\author{IlKwon Sohn\orcidC{}}
%\email{d2estiny@kisti.re.kr}
\affiliation{Quantum Network Research Center, Korea Institute of Science and Technology Information, Daejeon 34141, Korea}
\author{Kabgyun Jeong\orcidD{}}
\email{kgjeong6@snu.ac.kr}
\affiliation{Research Institute of Mathematics, Seoul National University, Seoul 08826, Korea}
\affiliation{School of Computational Sciences, Korea Institute for Advanced Study, Seoul 02455, Korea}

\date{\today}% It is always \today, today,
% d-sparse quantum state 후 symmetrization은, logd depth, d\eta\log N width\dots? 아 이럼 d~N이면 못함
% 대신에 내껀 d depth, log d width로 width 희생 안보고 full state quantum prep 후 가는거임.
\begin{abstract}
Reducing circuit depth and identifying an optimal trade-oﬀ between circuit depth and width is crucial
for successful quantum computation. In this context, midcircuit measurement and feedforward have been
shown to significantly reduce the depth of quantum circuits, particularly in implementing logical gates.
By leveraging these techniques, we propose several parallelization strategies that reduce quantum circuit
depth at the expense of increasing width in preparing various quantum states relevant to quantum simulation. With measurements and feedforward, we demonstrate that utilizing unary encoding as a bridge
between two quantum states substantially reduces the circuit depth required for preparing quantum states,
such as sparse quantum states and sums of Slater determinants within the first quantization framework,
while maintaining an efficient circuit width. Additionally, we show that a Bethe wave function, characterized by its high degree of freedom in its phase, can be probabilistically prepared in a constant-depth quantum circuit using measurements and feedforward. We anticipate that our study will contribute to the reduction of circuit depth in initial state preparation, particularly for quantum simulation, which is a critical step toward achieving quantum advantage. 
\end{abstract}

\maketitle
%\tableofcontents

\section{\label{sec:introduction}Introduction}
Many quantum algorithms are initially designed without the inclusion of auxiliary qubits. However, due to the limited coherence time of current quantum hardware, finding an adjustable trade-off between circuit depth and width (i.e., the total number of required qubits) is essential for optimizing performance on specific quantum devices. The most prominent way for reducing circuit depth involves the introduction of mid-circuit measurements and feedforward operations (quantum operations conditioned on measurement outcomes) \cite{piroli2021quantum, lu2022measurement, smith2024constant}. This approach is often referred as an adaptive circuit \cite{foss2023experimental} or dynamic circuit \cite{baumer2024efficient}. There have been several studies successfully employing this approach, including the preparation of matrix product states \cite{malz2024preparation, smith2024constant}, exploring topological order \cite{bravyi2022adaptive, tantivasadakarn2023hierarchy}, Hamiltonian simulation \cite{boyd2023low}, and implementing non-Clifford gates \cite{gidney2021cccz,kim2024resource}. 

An inevitable challenge associated with this approach is that incorporating measurements into a quantum circuit introduces probabilistic outcomes due to the principle of quantum mechanics. Several strategies have been proposed to address this issue. One approach involves leveraging the symmetry of the target state \cite{smith2023deterministic, smith2024constant}, while another focuses on designing circuits where any measurement outcomes can be corrected \cite{piroli2024approximating}. An alternative solution builds on the fact that quantum gate teleportation generates Pauli errors and that Clifford operations commute with Pauli gates. This concept, initially proposed in the context of measurement-based quantum computing \cite{gottesman1999demonstrating,jozsa2006introduction}, allows the construction of a quantum fan-out gate with a constant-depth quantum circuit with measurements and feedforward \cite{low2024trading, buhrman2024state}. Because it was well-established that various quantum logical operations, such as quantum AND and OR gates, can be derived from the quantum fan-out gate \cite{hoyer2005quantum, takahashi2016collapse}, Buhrman and collaborators \cite{buhrman2024state} proposed a constant-depth quantum circuit which prepares a Dicke state by exploiting those logical operations along with measurements, feedforward, and a number of auxiliary qubits that scales quadratically with the system size.

The most effective way to leverage this advantage in operating quantum logic gates may include preparing specific quantum states with known analytic forms. Unlike other quantum algorithms like Hamiltonian evolution, which involves complex operations, having an exact analytic form of the target quantum state provides a straightforward answer: ``What logical operations are needed to prepare this state?" We emphasize that even when the desired quantum state is known, the preparation of a well-characterized quantum state is a critical factor in the success of quantum computation for quantum simulation. For example, in finding the ground state via quantum phase estimation and its variants \cite{kitaev1995quantum, abrams1999quantum, dong2022ground, wang2023quantum}, the overlap between the input state and the ground state must be maximized to make the number of repetitions of quantum phase estimation manageable. In that sense, it is essential to prepare well-established, physically motivated ansatz such as the unitary coupled cluster ansatz \cite{romero2018strategies}, sums of Slater determinants \cite{fomichev2023initial}, and eigenstates of a Hamiltonian \cite{piroli2024approximating}. Similarly, in quantum adiabatic evolution \cite{albash2018adiabatic, ebadi2021quantum} where the initial Hamiltonian significantly differs from the final Hamiltonian, the ability to prepare a desired initial state is crucial, as the choice of the Hamiltonian path strongly influences the performance of adiabatic evolution \cite{vcepaite2023counterdiabatic, van2023adiabatic, mc2024towards}. Moreover, quantum state preparation has broad applications in quantum information, often called quantum access models \cite{zhang2022quantum, zhang2024circuit}.
% 이렇게 줄이는 방법이 있는데, 그 중에 logical이 대두되었다. 이는 무엇인지 알고 있어야 좋으니까, initial state prep에 대응된다.

In this study, we propose novel parallelization techniques that leverage constant-depth quantum logic gates inspired by Ref. \cite{buhrman2024state}. These techniques are applied to quantum state preparation to minimize quantum circuit depth. Furthermore, we emphasize that our parallelization techniques can be selectively applied to optimize the balance between depth and width. First, we demonstrate that the transformation between two quantum states, where each state is a superposition of an equal number of non-negative integers, can be performed using a constant-depth quantum circuit. This transformation allows us to prepare $d$-sums of $n$-qubit sparse quantum state with a quantum circuit depth of $O(\log d)$ and a width of $O(dn \log n)$, utilizing measurements and feedforward. Additionally, by combining this constant-depth transformation with other logical operations, we construct a constant-depth quantum circuit capable of preparing single (anti)symmetric states. These states are of great importance for the initial stages of quantum simulation in the first quantization framework. Moreover, we propose the preparation of higher-quality states like $d$-sums of Slater determinants, using quantum circuits of depth $O(\log d)$ and width $\tilde{O}(\eta^2 d^2\log N)$ or of depth $O(d)$ and width $\tilde{O}(\eta^2(\eta\log\eta+\log N)+\log d)$.

In addition to the transformation between two quantum states with different non-negative integer sets, we utilize a series of controlled pairwise commuting gates within a constant-depth quantum circuit  \cite{hoyer2005quantum}. This technique serves as a powerful tool for encoding phases into a quantum state. Using this method, we demonstrate that a Bethe wavefunction \cite{sahu2024fractal}, characterized by its extensive degrees of freedom in phase, can be prepared in a constant-depth quantum circuit with measurements and feedforward. This Bethe wavefunction defined in \cite{sahu2024fractal} is the generalization of conventional wavefunction under the coordinate Bethe ansatz, so it is not restricted to eigenstates of integrable Hamiltonians. The Bethe ansatz, first introduced by Hans Bethe \cite{bethe1931theorie}, is an analytic solution for many interacting one-dimensional Hamiltonians. The coordinate Bethe ansatz refers to the original form among various Bethe ansatz \cite{giamarchi2003quantum, van2016introduction,slavnov2018algebraic}. Due to the computational difficulty in calculating higher-order correlations or other complex observables on classical computers, preparation of Bethe wavefunction (or fixed Hamming weight encoder) on a quantum computer, which include the coordinate Bethe ansatz eigenstates have attracted a lot of attention \cite{gard2020efficient, van2021preparing, li2022bethe, raveh2024deterministic, sahu2024fractal, farias2025quantum}. There has also been extensive research in the realization of algebraic Bethe ansatz on a quantum computer \cite{sopena2022algebraic,ruiz2024bethe1, ruiz2024bethe2, ruiz2024efficient}. Our preparation of the Bethe wavefunction employs a series of controlled pairwise commuting gates for its phase encoding and a symmetric quantum state, where both are achievable using a constant-depth quantum circuit. For the Bethe wavefunction with $N$ sites and $\eta$ particles, our constant-depth preparation of Bethe wavefunction requires a circuit width of $\tilde{O}(N^2\log N)$, with a success probability of $O(1/M!)$. A summary of our results can be found in Table. \ref{tab:summary}.

%In comparison, previous algorithms for \textcolor{blue}{Bethe wavefunction} preparation required an $O(N)$ quantum circuit depth with $O(\eta^2)$ auxiliary qubits, where the success probability $O(1/\eta !)$ is equivalent to ours \cite{van2021preparing}. Additionally, we also mention that a deterministic algorithm for a coordinate Bethe ansatz preparation requires a circuit depth of $O(\binom{N}{\eta})$ \cite{raveh2024deterministic}.

The remainder of this paper is organized as follows. First, we introduce various notations that will be used throughout this paper in Sec. \ref{sec:preliminaries}. In Sec. \ref{sec:routines_for_parallelization}, we introduce several parallelization techniques aimed at reducing a quantum circuit's depth and apply them to prepare a sparse quantum state. In Sec. \ref{sec:antisymmetric_quantum_states}, we propose algorithms for preparing single (anti)symmetric state and sums of Slater determinants, leveraging the techniques from Sec. \ref{sec:routines_for_parallelization}. Furthermore, we propose a constant-depth quantum circuit in Sec. \ref{sec:bethe_ansatz} to prepare the Bethe wavefunction. Finally, in Sec. \ref{sec:conclusion}, we summarize our findings and discuss the broader implications of our work.

\begin{table*}[t]
\begin{ruledtabular}
\begin{tabular}{cccccc}
Target state & Reference & Circuit size & Circuit depth & Circuit width & Remarks \\ \hline 

 & \cite{luo2024circuit} & $O(dn/\log d)$ & $O(\log dn)$ & $O(dn/\log d)$ & / \\ [1ex]
Sparse quantum state & Thm. \ref{thm:sparse} & / & $O(\log d)$ & $O(dn \log n)$ & Adaptive circuit \\ [1ex]
  & \cite{zi2025constant} & / & $O(1)$ & $O(d^2\log n)$ & Adaptive circuit \\ \hline
 & \cite{berry2018improved,fomichev2023initial}$^{*}$ & $O(d\log d)^{**}$ & / & $O(\eta\log N+\log d)$ & $^{**}$\#Toffoli qubits \\ [1ex]
\multirow{2}{*}{\makecell{Sums of Slater\\determinants}} & Thm. \ref{thm:sos} & / & $O(d)$ & $\tilde{O}(\eta^2(\eta\log \eta+\log N)+\log d)$ & Adaptive circuit \\ [1ex]
  & \cite{berry2018improved}, Thm. \ref{thm:sparse}$^{***}$ & / & $\tilde{O}(\log^2\eta)$ & $\tilde{O}(\eta d \log N \log \eta)$ & / \\ [1ex]
    & Thm. \ref{thm:sos} & / & $O(\log d)$ & $\tilde{O}(\eta^2 d^2\log N)$ & Adaptive circuit \\  \hline
    
 & \cite{van2021preparing} & $O(L)$ & $O(L)$ & $O(L)$ & Success Prob. $O(1/M!)$ \\ [1.5ex]
 \multirow{3}{*}{\makecell{Bethe\\wavefunction}}& \cite{farias2025quantum} & $O(\binom{L}{M})^{****}$ & $O(\binom{L}{M})$ & $O(L)$ & $^{****}$\#CNOT gates \\ [1.5ex]
  & \cite{sahu2024fractal} & / & $O(\log L)$ & $O(L)$ & \makecell{Adaptive circuit,\\Classical cost $O(2^{4M})$}\\ [2.5ex]
    & Thm. \ref{thm:bethe} & / & $O(1)$ & $O(L^2\log L)$ & \makecell{Success Prob. $O(1/M!)$,\\Infidelity $O(1/2^{L^L})$ }
  
\end{tabular}
\end{ruledtabular}
\caption{Summary of several quantum state preparation algorithms, including our results. For sparse quantum state preparation, we list only those algorithms which aim to minimize quantum circuit depth, even at the expense of circuit width, which is our main purpose. “Adaptive circuit” in the remarks column denotes the fact that midcircuit measurement and feedforward operations are used in the given algorithm. / indicates that the item is either not specified in the original reference or is unclear. $^{*}$ denotes the algorithm that are newly proposed in this work by applying algorithms from Refs. \cite{berry2018improved,fomichev2023initial}. $^{**}$ indicates the circuit size was measured by the number of Toﬀoli gates. $^{***}$ denotes circuit complexity newly derived in this paper, based on techniques in Ref. \cite{berry2018improved}. $^{****}$ refers to circuit size measured by the number of CNOT gates.}
\label{tab:summary}
\end{table*}

\section{Preliminaries}\label{sec:preliminaries}
In this work, we demonstrate the preparation of three distinct quantum states: a sparse quantum state, linear combinations of Slater determinants, and a Bethe wavefunction. Due to the differing contexts in which these states arise, we employ a variety of notational conventions throughout the paper. For clarity, we provide a summary of these notations below, even though they are reintroduced in the sections where they appear.

\begin{itemize}
    \item $|x\rangle$: a quantum state representing an integer $x$ under the binary basis.
    \item $|e_x\rangle$: a quantum state representing an integer $x$ under the unary basis.
    \item $n$: the number of qubits for general quantum states, especially for sparse quantum states.
    \item $d$: the sparsity of sparse quantum state, or the number of considered energy level combinations in sums of Slater determinants.
    \item $N$: the number of considered energy levels in a (anti)symmetric state, and sums of Slater determinants.
    \item $\eta$: the number of particles in a (anti)symmetric state, and sums of Slater determinants.
    \item $S_\eta$: a permutation group of size $\eta$.
    \item $\pi(\sigma)$: a parity of a permutation $\sigma$.
    \item $M$: the number of particles of a Bethe wavefunction.
    \item $L$: the number of sites of a Bethe wavefunction.
    \item $\bm{\theta}$: a $M \times M$ real scattering matrix.
    \item $\vec{k}$: a real quasimomentum vector.
\end{itemize}

\section{Routines for parallelization}\label{sec:routines_for_parallelization}
\begin{table*}[tbh!]
\begin{ruledtabular}
\begin{tabular}{ccc}
 Operation&Operation on $n$-qubit quantum state&Width\\ \hline 
 $Fan$-$out$ & $\ket{x}_1 \ket{y_0}_1\dots\ket{y_{n-1}}_1\longrightarrow \ket{x}_1\ket{y_0\oplus x}_1\dots\ket{y_{n-1}\oplus x}_1$ & $O(n)$      \\ [2ex]
 $OR$ & $\ket{y_0}_1\dots\ket{y_{n-1}}_1\ket{x}_1\longrightarrow \ket{y_0}_1\dots\ket{y_{n-1}}_1\ket{\rm{OR}_n(y)\oplus x}_1$ & $O(n\log n)$\\ [2ex]
    $Equal_i$ & $\ket{x}_{n}\ket{y_1}_1\longrightarrow \ket{x}_n\ket{y_1\oplus \mathds{1}_{x=i}}_1=
  \begin{cases}
      \ket{x}_n\ket{y_1\oplus 1}_1 & \quad \text{if } x=i\\ 
      \ket{x}_n\ket{y_1}_1 & \quad \text{if } x\neq i\\
  \end{cases}$ & $O(n\log n)$ \\ [4ex]
 $AND$ & $\ket{y_0}_1\dots\ket{y_{n-1}}_1\ket{x}_1\longrightarrow \ket{y_0}_1\dots\ket{y_{n-1}}_1\ket{\rm{AND}_n(y)\oplus x}_1$ & $O(n\log n)$ \\ [2ex]
  $Hammingweight$ & $\ket{x}_n\ket{0}_{\log n}\longrightarrow \ket{x}_n\ket{|x|}_{\log n}$ & $O(n\log n)$   \\ [2ex]
  $Greatherthan$ & $\ket{x}_n\ket{y}_n\ket{0}_1\longrightarrow \ket{x}_n\ket{y}_n\ket{\mathds{1}_{x>y}}_1=
  \begin{cases}
      \ket{x}_n\ket{y}_n\ket{1}_1 & \quad \text{if } x>y\\
      \ket{x}_n\ket{y}_n\ket{0}_1 & \quad \text{if } x\leq y\\
  \end{cases}$ & $O(n^2)$\\ [4ex]
  $Permutation(\sigma)$ & $\ket{y_0}_1\dots\ket{y_{n-1}}_1\longrightarrow \ket{y_{\sigma(0)}}_1\dots\ket{y_{\sigma(n-1)}}_1$ & $O(n^2)$

\end{tabular}
\end{ruledtabular}
\caption{Lists of quantum logical operations that can be implemented in a constant-depth quantum circuit with measurements and feedforward. The subscript $N$ in the quantum state $\ket{y}_n$ indicates that this state consists of $n$ qubits. For the $OR$($AND$) gate, $\mathrm{OR}_n(y)$($\mathrm{AND}_n(y)$) refers to the result bit OR(AND) operation for N-bit $y_0 y_1 \dots y_{n-1}$. In addition, we use the notation $|x|$ to denote the Hamming weight of $n$-bit $x$.} Additional operations are detailed in Ref. \cite{hoyer2005quantum, takahashi2016collapse,buhrman2024state}.\label{tab:operations}
\end{table*}
Various measures can be used to quantify the complexity of quantum circuits, such as the number of auxiliary qubits, circuit depth, and the number of $T$ gates. Given the diverse characteristics and capabilities of different quantum hardware, exploring trade-offs among these metrics is an active research area. In this work, we focus on reducing circuit depth, which is defined as the number of layers of quantum gates that cannot be executed simultaneously. To achieve this, we leverage well-established trade-offs between circuit depth and width via measurements and feedforward. Notably, this approach enables the circuit depth of certain computations to become independent of the system size \cite{hoyer2005quantum, takahashi2016collapse}. It is worth mentioning that mid-circuit measurements and feedforward are essential features of fault-tolerant quantum hardware, thus they are reasonable, non-restrictive requirements. In this section, we introduce specific parallelizable computations that utilize these techniques: (1) the transformation between two quantum states, where each state is a superposition of an equal number of non-negative integers, and (2) multi-controlled pairwise commuting gates. Additionally, all logarithmic expressions in this paper are taken to base 2, and the circuit depth is calculated under the assumption that each layer of the quantum circuit consists of single-qubit gates and CNOT gates.

% APS 논문이니까 approximating ~~ 처럼 complexity 논의 깊게는 안하는게 좋을듯
\subsection{The hybrid circuit model in consideration}
Since this study involves both quantum circuits and classical circuits that govern feedforward operations, we examine the structure of hybrid circuits in this study. Specifically, we focus on circuits composed of alternating layers of quantum and classical operations \cite{buhrman2024state}. After each quantum layer, mid-circuit measurements are performed, and feedforward operations are determined by subsequent classical computation that takes these measurement results as inputs. In this paper, we will refer to this hybrid circuit model as a ``quantum circuit with measurements and feedforward." If the depth of the quantum circuit is independent of the system size, we will refer to this circuit as a ``constant-depth quantum circuit with measurements and feedforward"-or simply, a ``constant-depth quantum circuit". All classical computation layers in this hybrid circuit, as considered in our study, have a depth of $O(\log n)$, and a circuit size of $O(\mathrm{poly} (n))$ for $n$-bit calculation \cite{buhrman2024state}. % 굳이 저얘길 여기서 꺼내야하는가
% Fig1을 아예 그냥 수학 회로 빼버리고 전체적인 변환을 카툰식으로. 간단하게

\subsection{Clifford ladder circuit}
One approach to reducing quantum circuit depth is to recursively apply quantum gate teleportation \cite{jozsa2006introduction, gottesman1999demonstrating}. However, since quantum gate teleportation requires quantum correction at each step based on the measurement outcomes, the probability of success decreases exponentially with the number of teleportations. This issue can be solved when we only teleport Clifford gates \cite{buhrman2024state}. Because Clifford gates commute with Pauli operators, the correction process can be deferred and parallelized within a ladder circuit structure. The definition of the Clifford ladder circuit and its parallelization is provided below.
\begin{definition}{\textnormal{(Clifford-ladder circuit \cite{buhrman2024state})}}
Let $n$ be the number of qubits under consideration. A Clifford-ladder circuit is defined as
\begin{equation}
    C_{\rm{ladder}}:=\prod_{i=0}^{n-2} U^{(i)}_{i,i+1},
\end{equation}
where the circuit has a depth of $n-1$. The set of gates $\{U^i\}^{n-2}_{i=0}$ consists of $n-1$ two-qubit Clifford gates, and $U^{(i)}_{i,i+1}$ indicates that the gate $U^{(i)}$ acts on the $i$-th and $(i+1)$-th qubits.
\end{definition}
\begin{lemma}\label{lem:clifford}
Any $n$-qubit Clifford-ladder circuit can be implemented in a constant-depth quantum circuit of width $O(n)$ using measurements and feedforward.
\end{lemma}
\begin{proof}
See Lemma 3.8 in Ref. \cite{buhrman2024state}.    
\end{proof}
We note that all auxiliary qubits used in our study, particularly those employed to create Bell pairs in $C_{\rm{ladder}}$, can be reused in subsequent steps of the algorithm after measurements. 

From Lemma \ref{lem:clifford}, we can implement the quantum fan-out gate in a constant-depth quantum circuit, where the quantum fan-out gate is defined as
\begin{equation}
    \ket{x}\ket{y_1}\dots\ket{y_n}\longrightarrow \ket{x}\ket{y_1\oplus x}\dots\ket{y_n\oplus x},
\end{equation}
where $x,y_i\in \{0,1\}$. In Ref. \cite{baumer2024measurement}, this measurement-based quantum fan-out gate was suggested to be more scalable than its unitary implementation. Moreover, since the $\rm{SWAP}$ gate is a Clifford gate, any $\rm{SWAP}$ operation between far-reaching qubits can also be implemented in a constant-depth quantum circuit with measurements and feedforward. In fact, numerous quantum logic operations have been shown to be implementable in constant-depth quantum circuits if a constant-depth quantum circuit for the quantum fan-out gate is available. These operations belong to the complexity class $\rm{QNC}_{\rm{f}}^0$ \cite{hoyer2005quantum}. We summarize logical operations that can be applied by constant-depth quantum circuits in Table. \ref{tab:operations}. 

\subsection{Unary encoding}
\begin{figure}[t!]  % 여기 그림에 T도 넣고, n, N 반영하기
    \centering
    \subfloat[\label{subfig:uncompress}]{
    \includegraphics[width=1.0\linewidth]{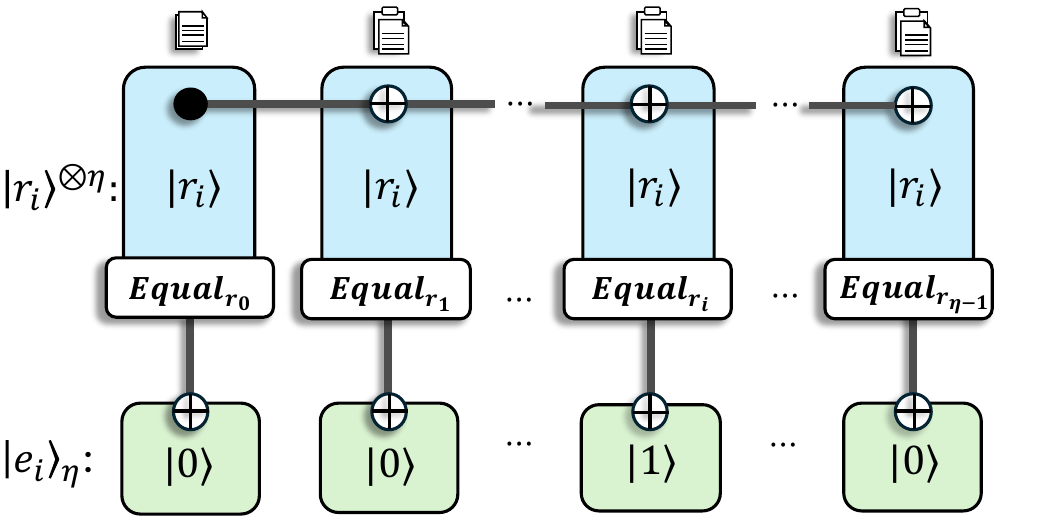}
    }\vfill
    \subfloat[\label{subfig:whole_scheme}]{
    \includegraphics[width=1.0\linewidth]{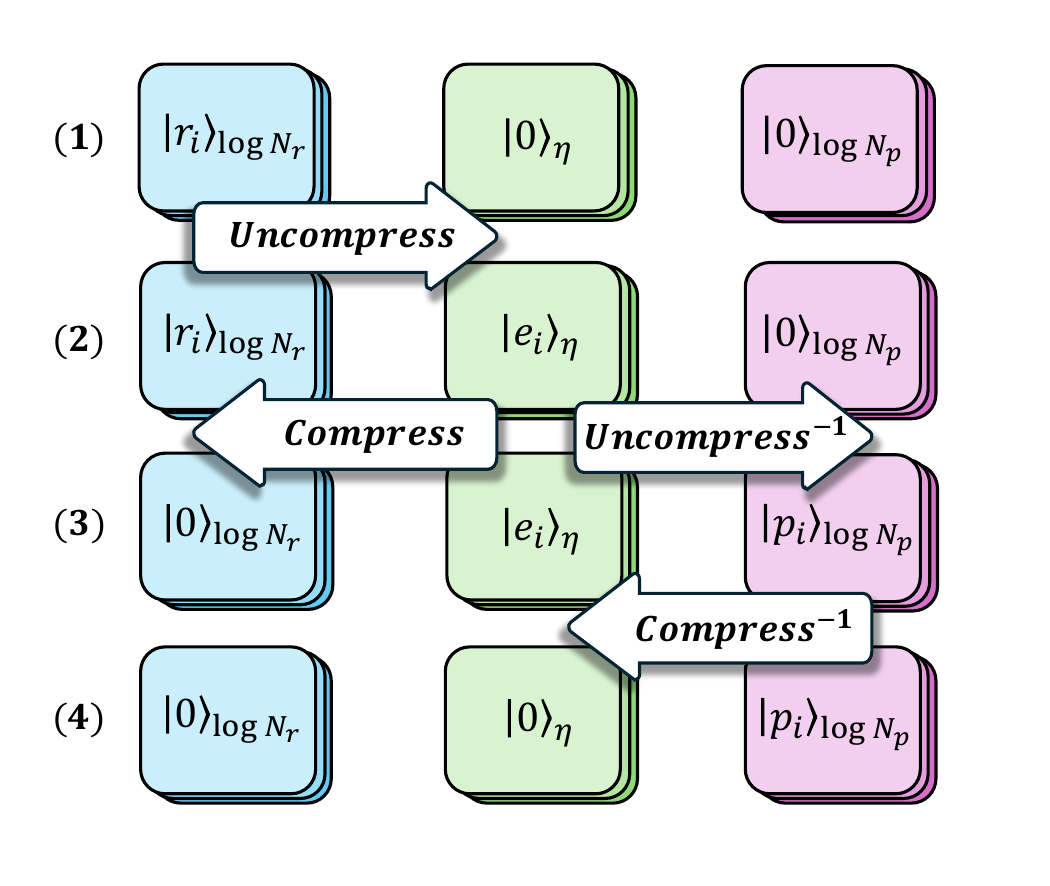}
    }
    \caption{The blue rectangles represent the $\{r_i\}_{i=0}^{\eta-1}$ registers, the green rectangle represents the registers, and the pink rectangles represent $\{p_i\}_{i=0}^{\eta-1}$ registers. (a) The quantum circuit consists of \textit{Equal} gates during \textit{Uncompress} operation on $\ket{r_i}$. The quantum fan-out gates duplicate the quantum state $\ket{r_i}$. (b) The entire process of transforming $\ket{r_i}_{\log N_r}$ into $\ket{p_i}_{\log N_p}$ via unary encoding as an intermediate step. The quantum circuit proceeds sequentially from (1) to (4). The operators $Uncompress^{-1}$ and $Compress^{-1}$ denote the inverse of $Uncompress$ and $Compress$, respectively. Note that this process can be simultaneously applied to the superpositioned state, $\sum_{i=0}^{\eta-1}\ket{r_i} \longrightarrow\sum_{i=0}^{\eta-1}\ket{p_i}$.}
    \label{fig:unary}
\end{figure}
With the fact that the quantum fan-out gate can be applied in a constant-depth quantum circuit, additional gates beyond the Clifford ladder can also be implemented in parallel. One example is a series of controlled pairwise commuting gates, which will be discussed in the next section \cite{hoyer2005quantum}. Before proceeding, it is useful to introduce unary encoding (also known as one-hot encoding) to extend the applicability of controlled pairwise commuting gates. The unary encoding provides the advantage of easier access to information at the cost of increased qubit usage, and it has gathered a lot of attraction in the context of quantum simulation, and quantum machine learning \cite{gard2020efficient, weigold2020data, ramos2021quantum, johri2021nearest, kerenidis2021classical,kerenidis2022quantum, zhang2023efficient, farias2025quantum}. Unlike binary encoding, unary encoding represents information about the qubit's position. For example, consider a system with four total data, and we want to encode the number two in the quantum state. In binary encoding, the state is represented as $\ket{2}_{\log 4}:=\ket{10}_{\log 4}$. In contrast, in unary encoding, the state is represented as $\ket{e_2}_{4}:=\ket{0100}_{4}$, where the state $\ket{e_i}$ denotes the unary encoding of $\ket{i}$. We note that the subscript $n$ in $\ket{x}_n$ indicates the number of qubits, and this notation will be used throughout the paper to ensure clarity. 

Another reason for introducing unary encoding is that the transformation between binary encoding and unary encoding can be achieved by a constant-depth quantum circuit with measurements and feedforward. This process is reversible, as the amount of information contained in the state is independent of the encoding method. Inspired by Ref. \cite{buhrman2024state}, we refer to the transformation from binary encoding to unary encoding as \textit{Uncompress}, and the reverse transformation as \textit{Compress}. Unlike other proofs in this paper, we propose a detailed exposition of these proofs in the main text, as they constitute the core of our parallelization techniques. We explicitly denote the number of qubits to avoid ambiguity by attaching subscripts. In addition, we assume that $N$ is a power of $2$. If this is not the case, we define the number of qubits required to represent integers from $0$ to $N-1$ as $\lceil \log N \rceil$. We also note that the notation $\tilde{O}$ indicates that coefficients of terms equal to or smaller than $\log\log$ are omitted.
% binary number 표현하려면 0부터 시작해야한다

\begin{lemma}\label{lem:uncompress}{\textit{(Uncompress)}}
    For two positive integers $\eta$ and $N$ satisfying $\eta\leq 2^{\lceil\log N\rceil}$, let $\{r_i\}_{i=0}^{\eta-1}$ be a set of non-negative integers where $0\leq r_0<r_1<\dots<r_{\eta-1}< N-1$. Then, the Uncompress operation refers to 
    \begin{equation}
        \sum_{i=0}^{\eta-1} \alpha_i \ket{r_i}_{\log N} \ket{0}_{\eta} \longrightarrow \sum_{i=0}^{\eta-1} \alpha_i\ket{r_i}_{\log N}\ket{e_{i}}_{\eta},
    \end{equation}
    where $\alpha_i\in \mathbb{C}$ satisfy $\sum_{i=0}^{\eta-1}|\alpha_i|^2=1$. This operation can be implemented in a constant-depth quantum circuit of width $\tilde{O}(\eta\log N)$.
\end{lemma}
\begin{proof}
    Note that we omit the amplitude factor in the proof, as only the bit values are altered. We refer to qubits encoded in binary as the binary register, and those encoded in unary as the unary register. To simultaneously access each number $r_i$ in the unary register, we need several copies of the binary register. This duplication can be achieved in a constant-depth quantum circuit using the quantum fan-out gate with $O(\log N)$ auxiliary qubits. Since $\eta$ copies of the binary register are needed, the operation requires $O(\eta\log N)$ auxiliary qubits. The state then transforms as
    \begin{equation*}
         \sum_{i=0}^{\eta-1} \ket{r_i}_{\log N}\ket{0}^{\otimes \eta-1}_{\log N}\ket{0}_\eta \longrightarrow \sum_{i=0}^{\eta-1} \ket{r_i}_{\log N}^{\otimes \eta}\ket{0}_\eta.
    \end{equation*}
    Next, apply the $Equal_{r_i}$ operation in the Table. \ref{tab:operations}, where the control qubits are the $i$-th binary register and the target qubit is the $i$-th qubit of the unary register $\ket{0}_\eta$. By simultaneously applying $Equal_{r_i}$ for $0\leq i \leq \eta-1$, we get the state
    \begin{equation*}
        \sum_{i=0}^{\eta-1} \ket{r_i}_{\log N}^{\otimes \eta}\ket{0}_\eta \longrightarrow \sum_{i=0}^{\eta-1} \ket{r_i}^{\otimes \eta}\ket{e_i}_\eta,
    \end{equation*}
    and this process is illustrated in Fig. \ref{subfig:uncompress}. According to the Table. \ref{tab:operations}, $\eta$ $Equal_{r_i}$ gates require $O(\eta \log N \log \log N)$ quantum circuit width. Finally, the copied binary registers can be uncomputed by reapplying the quantum fan-out gate, resulting in the desired state,
    \begin{equation*}
        \sum_{i=0}^{\eta-1} \ket{r_i}_{\log N}^{\otimes \eta}\ket{e_i}_\eta \longrightarrow \sum_{i=0}^{\eta-1} \ket{r_i}_{\log N}\ket{0}_{\log N}^{\otimes \eta-1}\ket{e_i}_\eta.
    \end{equation*}
    The dominant quantum circuit width $O(\eta \log N \log \log N)$ comes from operating $Equal$ gates. Given that $\log\log \eta$ is negligible compared to other things, we disregard it and express the complexity as $\tilde{O}(\eta\log N)$.
\end{proof}

\begin{lemma}\label{lem:compress}{(Compress)}
    For two positive integers $\eta$ and $N$, satisfying $\eta\leq 2^{\lceil\log N\rceil}$, let $\{r_i\}_{i=0}^{\eta-1}$ be a set of non-negative integers where $0\leq r_0<r_1<\dots<r_{\eta-1}\leq N-1$. Then, the Compress operation
    \begin{equation}
        \sum_{i=0}^{\eta-1} \alpha_i \ket{r_i}_{\log N} \ket{e_{i}}_\eta \longrightarrow \alpha_i \ket{0}_{\log N}\ket{e_{i}}_\eta,
    \end{equation}
    where $\alpha_i\in \mathbb{C}$ satisfy $\sum_{i=0}^{\eta-1}|\alpha_i|^2=1$. The Compress operation can be implemented in a constant-depth quantum circuit of width $O(\eta\log N)$.
\end{lemma}
\begin{proof}
Due to the same reason with the proof of Lemma \ref{lem:uncompress}, we omit the amplitude factor. In addition, we refer to the first $\log N$ qubits as the binary register and the next $\eta$ qubits as the unary register. Although unitary operators cannot generally remove information from a quantum state, we can erase the information of $\{r_i\}_{i=0}^{\eta-1}$ in the binary register since the values of $\{r_i\}_{i=0}^{\eta-1}$ are known. The first step involves applying Hadamard gates to each qubit in the binary register, transferring the information of $\{r_i\}_{i=0}^{\eta-1}$ into the phases, 
\begin{equation*}
    \sum_{i=0}^{\eta-1} \ket{r_i}_{\log N} \ket{e_{i}}_\eta \longrightarrow \sum_{i=0}^{\eta-1}\sum_{j=0}^{N-1} (-1)^{r_i \cdot j}\ket{j}_{\log N}\ket{e_i}_\eta,
\end{equation*}
where $r_i \cdot j$ represents the bitwise product of the binary numbers $r_i$ and $j$. Next, we use quantum fan-out gates to duplicate the binary register $\eta$ times,
\begin{equation*}
    \sum_{i=0}^{\eta-1}\sum_{j=0}^{N-1} (-1)^{r_i \cdot j} \ket{j}_{\log N}\ket{e_i}_\eta \longrightarrow \sum_{i=0}^{\eta-1}\sum_{j=0}^{N-1} (-1)^{r_i\cdot j}\ket{j}_{\log N}^{\otimes \eta}\ket{e_i}_\eta.
\end{equation*}
This duplication can be implemented in a constant-depth quantum circuit of width $O(\eta \log N)$. Using the unary register, we access each phase $(-1)^{r_i \cdot j}$. For each $0\leq l \leq \eta-1$ and $0\leq k \leq \log N-1$, we apply controlled-Z gates where the controlled qubit is the $k$-th qubit of the unary register, and the target qubit is the $k$-th qubit of $l$-th binary register from $\ket{j}_{\log N}^{\otimes \eta}$ if the $k$-th bit of $r_l$ is 1. Since a multi-target controlled-Z gate is a simple variant of the quantum fan-out gate, this operation can also be performed in parallel with a circuit width of $O(\eta\log N)$. This process completely removes the phase $(-1)^{r_i\cdot j}$, 
\begin{equation*}
    \sum_{i=0}^{\eta-1}\sum_{j=0}^{N-1} (-1)^{r_i \cdot j}\ket{j}^{\otimes \eta}\ket{e_i}_\eta \longrightarrow \sum_{i=0}^{\eta-1}\sum_{j=0}^{N-1} \ket{j}^{\otimes \eta}\ket{e_i}_\eta.
\end{equation*}
Finally, we obtain the desired state by reapplying the quantum fan-out gate and Hadamard gates,
\begin{equation*}
 \sum_{i=0}^{\eta-1}\sum_{j=0}^{N-1} \ket{j}_{\log N}^{\otimes \eta}\ket{e_i}_\eta \longrightarrow 
 \sum_{i=0}^{\eta-1} \ket{0}^{\otimes \eta-1}_{\log N}\ket{e_i}_\eta.
\end{equation*}
The dominant contribution to the circuit width arises from duplicating the binary register $\eta$ times, resulting in $O(\eta\log N)$.
\end{proof}
Given that both \textit{Uncompress} and \textit{Compress} operations are reversible, combining Lemmas \ref{lem:uncompress} and \ref{lem:compress} allows us to construct a constant-depth quantum circuit that transforms a quantum state between two different sets of non-negative integers, $\{r_i\}_{i=0}^{\eta-1}$ and $\{p_i\}_{i=0}^{\eta-1}$ satisfying $0\leq r_0<\dots<r_{\eta-1}\leq N-1$ and $0\leq p_o<\dots<p_{\eta-1}\leq N-1$, respectively. We illustrate the entire process of transforming the quantum state $\sum_{i=0}^{\eta-1}\ket{r_i}_{\log N_r}$ into $\sum_{i=0}^{\eta-1}\ket{p_i}_{\log N_p}$ in Fig. \ref{subfig:whole_scheme}, implemented within a constant-depth quantum circuit of width $\tilde{O}(\eta \log N_p)$, under the assumption that $N_r \leq N_p$. In Fig. \ref{subfig:whole_scheme}, the inverse of \textit{Uncompress} is represented as $Uncompress^{-1}$ and \textit{Compress} as $Compress^{-1}$, by leveraging the unitarity of both operations.
% 여기에 음..

It is worth noting that this unary-encoding-aided transformation was also utilized in Ref. \cite{luo2024circuit}. However, the difference lies in the fact that the authors in Ref. \cite{luo2024circuit} did not exploit parallelization techniques summarized in the Table. \ref{tab:operations}, and as a result, their quantum circuit was not a constant-depth. We emphasize that unary encoding is employed solely to represent the index of each integer, rather than the integer itself. This approach enables the design of an efficient algorithm with complexity that scales logarithmically with $N$. 

This transformation can be directly applied to sparse quantum state preparation \cite{zhang2022quantum, luo2024circuit, mao2024toward, ramacciotti2024simple}, which is widely applicable in various areas of quantum information processing. 
\begin{theorem}\label{thm:sparse}
    The $d$-sparse $n$-qubit quantum state is defined as
    \begin{equation}
    \ket{\psi_{\rm{sparse}}}=\sum_{i=0}^{d-1} \psi_i \ket{q_i}_n,
\end{equation}
where $\psi_i\in \mathbb{C}$ satisfy $\sum_{i=0}^{d-1}|\psi_i^2|=1$, and $\{q_i\}_{i=0}^{d-1}$ are binary numbers satisfying $0\leq q_i < 2^n$ for all $i$. This state can be prepared by a quantum circuit of depth $O(\log d)$ and width $O(dn\log n)$ using measurements and feedforward.
\end{theorem}
\begin{proof}
    According to Ref. \cite{zhang2022quantum}, which also employs measurements and feedforward to implement non-local CNOT gates, preparing an arbitrary $\tilde{n}=\lceil \log d \rceil$-qubit quantum state $\sum_{i=0}^{d-1}\psi_i\ket{i}_{\tilde{n}}$ requires quantum circuit of depth $O(\log d)$ and width $O(d)$. The desired state can be obtained by first applying the \textit{Uncompress} and \textit{Compress} operations, transforming the state to $\sum_{i=0}^{d-1} \psi_i\ket{e_i}_{d}$, requiring $\tilde{O}(d\log d)$ auxiliary qubits. Next, the inverse of \textit{Compress} and \textit{Uncompress} operations are applied to encode the information of the set $\{q_i\}$ into the quantum state, requiring $O(d\log 2^n \log \log 2^n)=O(dn\log n)$ auxiliary qubits.
\end{proof}
Note that the authors in Ref. \cite{zi2025constant} also proposed an algorithm for sparse quantum state preparation, using measurements and feedforward. Their quantum circuit have achieved constant-depth, but the width scales as $O(d^2\log n)$. Considering that $d$ is generally polynomial in $n$, we observe a tendency in these results \cite{luo2024circuit,zi2025constant} that the more we sacrifice width for the sake of depth, the worse the product of depth and width becomes. This tendency can also be found in Table. \ref{tab:summary}.

However, that tendency does not indicate that unary encoding makes things worse. In Ref. \cite{luo2024circuit}, the authors showed that by introducing unary encoding, a sparse quantum state can be prepared by quantum circuit size of $O(dn/\log d)$, depth of $O(\log dn)$, and width of $O(dn/\log d)$. To achieve quantum circuit depth logarithmic in $d$ and $n$, binary encoding requires circuit size of $O(dn \log d)$ where circuit depth is $O(\log dn)$, and width is $O(dn\log d)$ \cite{zhang2022quantum}. It is worth noting that Hamming weight encoding, which mediates unary and binary encoding, can be another solution for quantum state preparation \cite{kerenidis2022quantum, zhang2023efficient, farias2025quantum}. However, to the best of our knowledge, no previous work has employed Hamming weight encoding in the way this paper does, which is to drastically reduce circuit depth (logarithmically dependent on $d$ and $n$). 

\subsection{Series of Controlled-pairwise commuting gates}
\begin{figure}[t!]  % 여기 그림에 T도 넣고, n, N 반영하기
    \centering
    \subfloat[\label{subfig:series1}]{
    \includegraphics[width=1.0\linewidth]{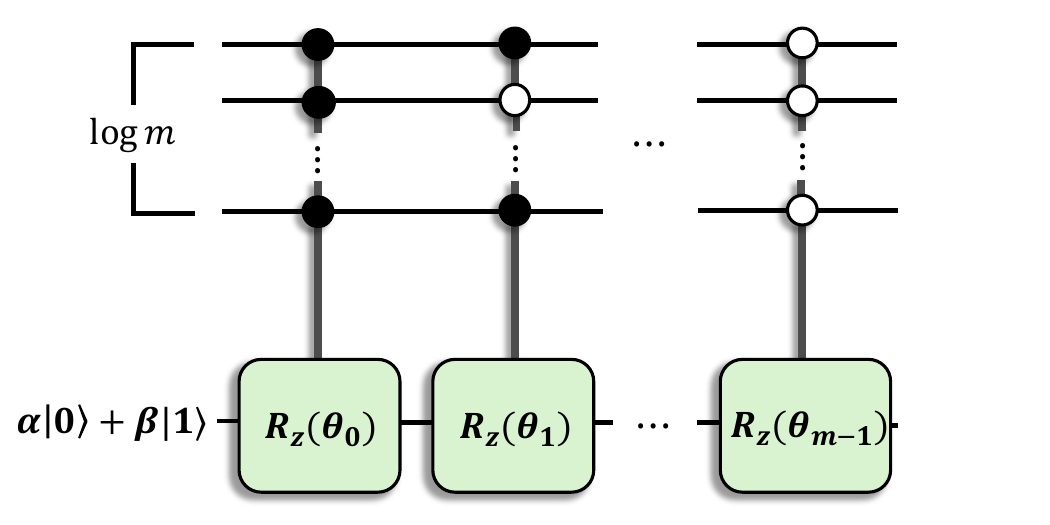}
    }\vfill
    \subfloat[\label{subfig:series2}]{
    \includegraphics[width=1.0\linewidth]{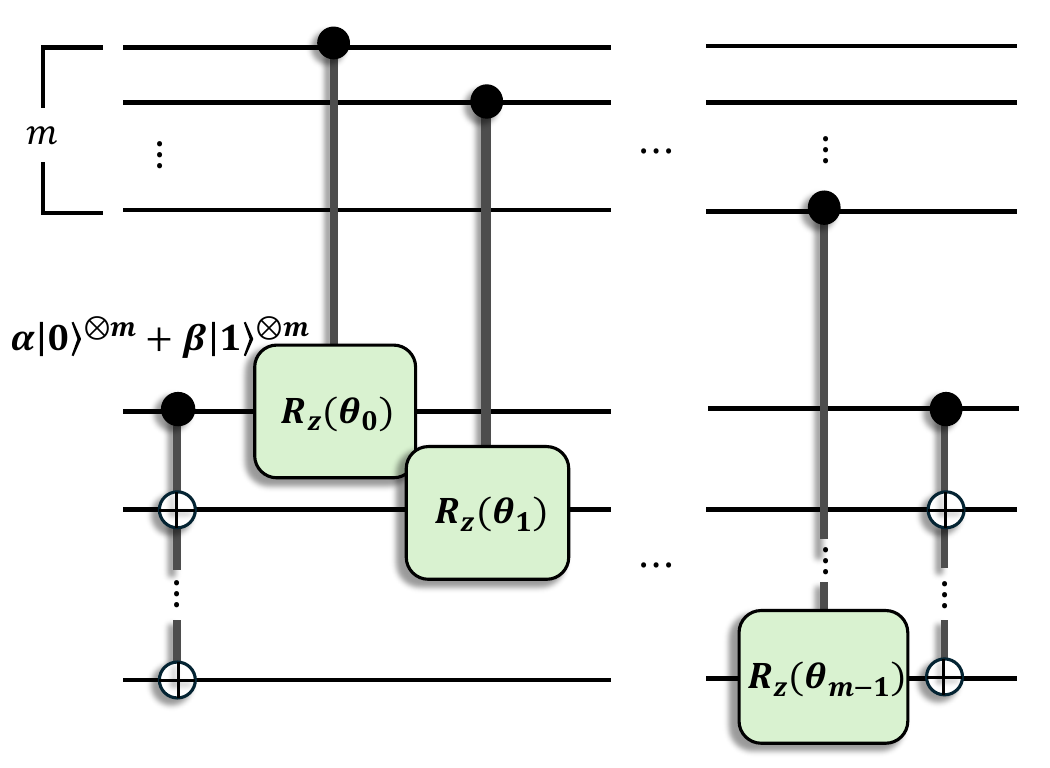}
    }
    \caption{(a) A series of multi-controlled pairwise commuting gates, with $R_z$ gates shown as an example of pairwise commuting gates. (b) The series of multi-controlled gates in (a) is replaced with controlled gates using the quantum fan-out gate and unary encoding.}
    \label{fig:series}
\end{figure}
At first glance, operating a series of $m$ controlled gates to $\log m$ qubits, as illustrated in Fig. \ref{subfig:series1}, appears to require at least $O(m)$ circuit depth. The situation becomes even more challenging when multi-controlled gates are involved \cite{vale2023decomposition}. However, this series of multi-controlled gates can be parallelized by leveraging a combination of measurements, feedforward, and certain assumptions. The parallelization of a series of controlled pairwise commuting gates is grounded in a fundamental theorem of linear algebra, which states that two commuting matrices can be simultaneously diagonalized. Combining these diagonalized gates with a replicated unary encoded controlled register-preparable by the quantum fan-out gate-it becomes possible to parallelize a series of controlled pairwise commuting gates as below.
\begin{lemma}\label{lem:controlled}
    Let $\{U_i\}^{m-1}_{i=0}$ be a set of pairwise commuting gates acting on the same $k$ qubits, where $U^{(x_i)}_i$ denotes the gate $U_i$, controlled by the $\log m$-qubit quantum state $\ket{x_i}_{\log m}$. Furthermore, assume that a unitary operator $T$ simultaneously diagonalizes the set $\{U_i\}^{m-1}_{i=0}$. Then, a series of multi-controlled pairwise commuting gates
    \begin{equation}
        U:=\prod^{m-1}_{i=0} U^{(x_i)}
    \end{equation}
    can be implemented in a quantum circuit of depth $O(\max^{m-1}_{i=0} [\rm{depth}$$(U_i)]+4\cdot \rm{depth}(T))$ and width $O(km)$, where the function $\rm{depth}(U)$ refers to the circuit depth of $U$ when it is decomposed as CNOT gates and single qubit rotation gates (or other conventional gate decomposition sets).
\end{lemma}
\begin{proof}
    See Theorem 3.2 in Ref. \cite{hoyer2005quantum}. The entire circuit after parallelization is depicted in Fig. \ref{subfig:series2}. 
\end{proof}
From the perspective of quantum state preparation, Lemma \ref{lem:controlled} serves as a powerful tool for preparing specific states with a large degree of freedom in their phase. However, directly converting Fig. \ref{subfig:series1} into Fig. \ref{subfig:series2} would replace $\log m$ with $m$, resulting in an exponential increase in circuit width. Therefore, we apply the idea that multiple-phase gates can be applied simultaneously using quantum fan-out gates to a state that is already given in unary encoding. In section \ref{sec:bethe_ansatz}, we show how this tool can be combined with others to construct a constant-depth quantum circuit for the preparation of the Bethe wavefunction.

\section{Symmetric and antisymmetric quantum states} \label{sec:antisymmetric_quantum_states}
Symmetric and antisymmetric quantum states are two of the most fundamental concepts in quantum mechanics, arising from the characteristics of indistinguishable particles: bosons and fermions. From the perspective of second quantization, these symmetries are inherently ensured by specially designed annihilation and creation operators, and are not explicitly encoded at the state level \cite{jordan1928paulische, bravyi2002fermionic, coleman2015introduction, macridin2022bosonic}. Since the second quantization representation of fermions aligns more intuitively with the spin-$\frac{1}{2}$ formalism used in digital quantum computers compared to the first quantization representation, numerous studies have focused on preparing arbitrary Slater determinants on quantum computers using the second quantization approach way \cite{ortiz2001quantum,wecker2015solving,kivlichan2018quantum, jiang2018quantum,chee2023shallow}.

Compared to the second quantization, the representation of quantum states under the first quantization scheme requires exponentially fewer qubits in the number basis sets, and exhibits lower time complexity in certain cases \cite{abrams1997simulation}. This significantly reduced spatial complexity has motivated the adoption of the plane-wave basis, which requires a substantially larger number of basis functions than the Gaussian basis, but offers higher accuracy \cite{su2021fault}. Recently, the authors of Refs. \cite{babbush2023quantum, huggins2024efficient} proposed efficient basis transformation of Slater determinants within the first quantization framework. Additionally, efficient (anti)symmetrization algorithms have been introduced in Ref. \cite{berry2018improved}. However, the preparation of sums of Slater determinants under the first quantization, which is a widely used initial state in classical computational chemistry, has received relatively less attention \cite{su2021fault}. In this section, we first show how to prepare a single (anti)symmetric state using a constant-depth quantum circuit with measurements and feedforward. Then, leveraging the prepared single (anti)symmetric state, we construct $d$-sums of Slater determinants with a quantum circuit depth of either $O(d)$ or $O(\log d)$.

\subsection{Preparing single (anti)symmetric state}
The single (anti)symmetric state in our study is defined as follows.
\begin{definition}
    For two positive integers $\eta$ and $N$ satisfying $\eta\leq N$, let $\{r_i\}_{i=0}^{\eta-1}$ be a set of non-negative integers where $0\leq r_0<r_1<\dots<r_{\eta-1}< N$. Then, the single symmetric state is defined as
    \begin{equation}
        \ket{\psi_{\mathrm{sym}}}:= \frac{1}{\sqrt{\eta !}}\sum_{\sigma\in S_\eta} \ket{r_{\sigma(0)}} \dots \ket{r_{\sigma(\eta-1)}},
    \end{equation}
    and the single antisymmetric state is defined as
    \begin{equation}
        \ket{\psi_{\rm{anti}}}:=\frac{1}{\sqrt{\eta !}}\sum_{\sigma\in S_\eta}(-1)^{\pi(\sigma)} \ket{r_{\sigma(0)}}\dots\ket{r_{\sigma(\eta-1)}},
    \end{equation}
    where $S_\eta$ and $\pi(\sigma)$ denote the permutation group of size $\eta$, and the parity of the permutation $\sigma$, respectively. 
\end{definition}
The constant-depth preparation of a single (anti)symmetric state involves two key Lemmas \ref{lem:filling_filtering} and \ref{lem:parity}, as well as the \textit{Uncompress} and \textit{Compress} operations. Before that, we define one notation that will be frequently used in our paper.
\begin{definition}
For two positive integers $\eta$ and $N$ satisfying $\eta \leq N$, we define the set $\Xi_{\eta}^{N}$ as the collection of $\eta$-combinations of non-negative integers from $0$ to $N-1$, represented as
\begin{equation*}
    \Xi_{\eta}^{N}:=\{(x_0, \dots, x_{\eta-1}) \in \mathbb{Z}^\eta \mid 0 \leq x_0  < \dots < x_{\eta-1} < N \}.
\end{equation*}
\end{definition}
Our preparation process begins by applying the \textit{Filling} and \textit{Filtering} techniques described in Ref. \cite{buhrman2024state}.
\begin{lemma}\label{lem:filling_filtering}
    For two positive integers $\eta$ and $\zeta$ satisfying $\eta^2\leq \zeta$, the quantum state
    \begin{equation}
    \frac{1}{\sqrt{\binom{\zeta}{\eta}\eta !}}\sum_{\vec{j}\in \Xi_\eta^\zeta} \sum_{\sigma\in S_\eta} \ket{j_{\sigma(0)}}\dots\ket{j_{\sigma(\eta-1)}}\ket{\bigoplus_{i=0}^{\eta-1}e_{j_i}}
    \end{equation}
    can be prepared by a constant-depth quantum circuit of width $\tilde{O}(\eta\zeta\log\zeta)$ using measurements and feedforward. Here, $\bigoplus_{i=0}^{\eta-1}e_{j_i}=e_{j_0}\oplus\dots\oplus e_{j_{\eta-1}}$ denotes the bitwise summation of $e_{j_0},\dots,e_{j_{\eta-1}}$.
\end{lemma}
\begin{proof}
    See the \textit{Filling} and \textit{Filtering} processes in Ref. \cite{buhrman2024state}.
\end{proof}
According to Lemma \ref{lem:filling_filtering}, we obtain superpositions of $\ket{j_{\sigma(0)}}\dots\ket{j_{\sigma(\eta-1)}}$ and their corresponding indices $\ket{\bigoplus_{i=0}^{\eta-1}e_{j_i}}$. To construct the antisymmetric state, it is necessary to attach the phase $(-1)^{\pi(\sigma)}$ to each permutation. This can be achieved in a constant-depth quantum circuit by leveraging the definition of the parity $\pi(\sigma)$.
\begin{lemma}\label{lem:parity}
For two positive integers $\eta$ and $\zeta$ satisfying $\eta\leq \zeta$, let $\vec{j}=(j_0,\dots,j_{\eta-1})$ be any element of the set $\Xi_{\eta}^{\zeta}$. Then, the transformation
\begin{align}
    &\frac{1}{\sqrt{\eta!}}\sum_{\sigma\in S_\eta}\ket{j_{\sigma(0)}}\dots\ket{j_{\sigma(\eta-1)}}\\
    &~~\longrightarrow \frac{1}{\sqrt{\eta!}}\sum_{\sigma\in S_\eta}(-1)^{\pi(\sigma)}\ket{j_{\sigma(0)}}\dots\ket{j_{\sigma(\eta-1)}}
\end{align}
can be implemented in a constant-depth quantum circuit of width $O(\eta^2\log^2\zeta)$ with measurements and feedforward.
\end{lemma}
\begin{proof}
See Appendix \ref{app:proof_parity}.
\end{proof}
By combining these Lemmas, we now propose a constant-depth quantum circuit for preparing a single (anti)symmetric state using measurements and feedforward.
\begin{theorem}\label{thm:single_symmetric}
For two positive integers $\eta$ and $N$ satisfying $\eta\leq N$, let $\vec{r}=(r_0,\dots,r_{\eta-1})$ be any element of the set $\Xi_\eta^N$. Then, the single symmetric state $\ket{\psi_{\mathrm{sym}}}$ and antisymmetric state $\ket{\psi_{\rm{anti}}}$ can be prepared in a constant-depth quantum circuit of width $\tilde{O}(\eta^2(\eta\log \eta+\log N))$ with measurements and feedforward.    
\end{theorem}
\begin{proof}
See Appendix \ref{app:proof_single_symmetric}.
\end{proof}

\subsection{Preparaing sum of Slater determinants}
In Ref. \cite{berry2018improved}, the authors proposed an algorithm to (anti)symmetrizes any ordered state $\sum_{i=0}^{d-1}\psi_i \ket{r_0^i}\dots\ket{r_{\eta-1}^i}$, where each $\vec{r}^i\in \Xi_{\eta}^N$ and $\psi_i\in \mathds{C}$, with a circuit depth of $O((\log\eta)^2\log\log N)$ through quantum bitonic sort. It is evident that if the desired state is a single (anti)symmetric state, our Theorem \ref{thm:single_symmetric} achieves a significantly shorter circuit depth. However, for more complex problems, a single (anti)symmetric state may lack sufficient quality to provide a solution. In the following, we address the preparation of the sums of Slater determinants (SOS), representing a more advanced initial state for quantum simulation. To the best of our knowledge, there has been only one approach to preparing SOS in the first quantization scheme: encoding the coefficients of each determinant, encoding the corresponding sets of non-negative integers, and then performing antisymmetrization. Our results show improvement by combining the process of encoding integer sets and antisymmetrization.

Antisymmetrizing the state $\sum_{i=0}^{d-1}\psi_i \ket{r_0^i}\dots\ket{r_{\eta-1}^i}$ yields the SOS in the first quantization scheme, expressed as follows:
\begin{equation}\label{eq:SOS}
    \ket{\psi_{\rm{SOS}}}=\sum_{\sigma\in S_\eta}\sum_{i=0}^{d-1}\psi_i(-1)^{\pi(\sigma)}\ket{r_{\sigma(0)}^i}\dots\ket{r_{\sigma(\eta-1)}^i},
\end{equation}
where $d$ is the number of superposed states and $\psi_i\in \mathbb{C}$ represents the coefficient of the $i$-th Slater determinant. The SOS is a widely used initial state in computational chemistry, as it can incorporate more correlation than a single Slater determinant \cite{veis2010quantum,tubman2018postponing}. This concept is a key component of configuration interaction (CI) in chemistry, which is further categorized as configuration interaction with singles (CIS), configuration interaction with singles and doubles (CISD), and full configuration interaction (FCI), depending on the desired level of correlation \cite{wang2009efficient,helgaker2013molecular}. 

We now propose two different methods for preparing SOS. The primary distinction between these approaches lies in the number of qubits required. To ensure the feasibility of these algorithms under the first quantization scheme, it is essential to maintain a circuit width that is polynomial in $\eta$ and logarithmic in $N$. However, the scale of $d$ is highly problem-dependent: a small $d$ suffices for simple problems, $d$ scales linearly with the number of basis functions $N$ for CIS, and becomes quadratic in $N$ for CISD. If the circuit width depends on $O(\rm{poly}(d))$ and $d$ is polynomial in $N$, this approach becomes unsuitable for first quantization simulations. To address this, we propose two algorithms with circuit depths of $O(d)$ and $O(\log d)$, respectively, allowing the choice of method based on the scale of $d$. Notably, our algorithm can also prepare $d$-sums of permanent in the same manner, by omitting the phase factor $(-1)^{\pi(\sigma)}$.

The following theorem provides two algorithms for preparing sums of Slater determinants.
\begin{theorem}\label{thm:sos}
For positive integers $\eta$, $d$, and $N$ satisfying $\eta, d\leq N$, let $\vec{r}^0,\dots,\vec{r}^{d-1}$ be distinct elements of the set $\Xi_{\eta}^N$. Then, the $d$-sums of Slater determinants in the first quantization scheme $\ket{\psi_{\rm{SOS}}}$ can be prepared by a quantum circuit of depth $O(d)$ and width $\tilde{O}(\eta^2(\eta\log\eta+\log (N))+\log d)$ with measurements and feedforward, or a quantum circuit of depth $O(\log d)$ and width $\tilde{O}(\eta^2d^2\log N)$ with measurements and feedforward.
\end{theorem}
\begin{proof}
See Appendix \ref{app:proof_sos}.
\end{proof}
We now compare our algorithm with other approaches. The same state can be constructed by first preparing $\sum_{i=0}^{d-1}\psi_i \ket{r_0^i}\dots\ket{r_{\eta-1}^i}$ and subsequently performing (anti)symmetrization. The first part corresponds to sparse quantum state preparation, which has been extensively studied in the literature \cite{zhang2022quantum, luo2024circuit, mao2024toward, ramacciotti2024simple}. Since we propose two versions of our SOS preparation algorithm, we compare it with two methods for sparse quantum state preparation. For circuit depths logarithmic in $d$, we choose our Theorem \ref{thm:sparse} for sparse quantum state preparation. For circuit depths polynomial in $d$, sparse quantum states can be prepared using a quantum circuit with an overall Toffoli cost of $O(d\log d)$ and an additional qubit cost of $O(\log d)$ \cite{fomichev2023initial}. Notably, the authors in Ref. \cite{fomichev2023initial} adopted an end-to-end perspective, focusing on the preparation of initial states rather than solely on sparse quantum state preparation. They also proposed a method for preparing sparse quantum state with a Toffoli cost of $\min(2\sqrt{32\eta\log N d}, d)+(7\log d+2\sqrt{32\log d})\sqrt{d}$ with a clean qubit cost of $2\sqrt{d}\log d$. The second part, (anti)symmetrization, can be implemented using a quantum circuit with a depth of $O(\log^2\eta \log\log N)$ and a worst-case width of $O(\eta \log^2 \eta \log\log N)$ by employing quantum bitonic sort \cite{berry2018improved}.

Assume that $d=O(\eta)$, and we choose Theorem \ref{thm:sparse} to prepare sparse quantum state and perform antisymmetrization using the method proposed in Ref. \cite{berry2018improved}. In this approach, the dominant circuit depth is $\tilde{O}(\log^2\eta)$, which exceeds the depth of $O(\log d)$ achieved by Theorem \ref{thm:sos}. The circuit width of this approach is $\tilde{O}(\eta d\log N \log \eta)$ and ours is $\tilde{O}(\eta^2 d^2\log N).$ For the case where $d=\rm{poly}(N)$, a sparse quantum state can be prepared using the algorithm in Ref. \cite{fomichev2023initial}, followed by antisymmetrization. Although the result cannot be directly compared due to differences in measurement metrics, we emphasize that our Theorem \ref{thm:sos} provides significant value by focusing on the extreme reduction of circuit depth through measurements and feedforward.

Until now, we have represented our SOS on a simple molecular orbital basis. This basis is one of the most intuitive and widely used bases, with significant value in many applications \cite{tubman2018postponing,georges2024quantum}. However, it is known that this simple basis form can lead to a dense Hamiltonian representation, which imposes a substantial overhead for the quantum circuits to process \cite{babbush2018low, kivlichan2018quantum}. From this perspective, the plane wave basis has garnered attention due to its sparse representation and the reduced circuit depth required for its implementation in the first quantization scheme \cite{babbush2019quantum, su2021fault}. We note that the state-of-the-art algorithm proposed in Ref. \cite{huggins2024efficient}, which transforms molecular orbital basis to the plane-wave basis, can be directly applied to our SOS as defined Eq. (\ref{eq:SOS}).
% 음 이거 논문이 이내용이 맞나? 내일 다시보자. gaussian basis에서 변환? molecular orbital?
% --------------------
\section{Bethe wavefunction}\label{sec:bethe_ansatz}
\subsection{About Bethe Ansatz}
The Bethe ansatz was initially proposed by Hans Bethe in 1931 \cite{bethe1931theorie} as an analytic solution for the spin-$\frac{1}{2}$ antiferromagnetic Heisenberg chain. This groundbreaking proposal provided a novel approach to address interacting Hamiltonians beyond free Hamiltonians. Following this monumental study, it was discovered that the Bethe ansatz could be applied to a wide range of systems, such as the XXZ chain \cite{yang1966one}, Bose gas \cite{lieb1963exact}, the Kondo model \cite{andrei1983solution}, and the Hubbard model \cite{lieb1968absence}. To address these numerous applications, several variants of the Bethe ansatz have been developed, such as coordinate Bethe ansatz (the original formulation) \cite{giamarchi2003quantum}, the algebraic Bethe ansatz \cite{slavnov2018algebraic}, and thermodynamic Bethe ansatz \cite{van2016introduction}. Given its versatility and profound influence across various fields, we recommend several references for readers interested in exploring the Bethe ansatz in greater depth \cite{karbach1998introduction, faddeev1995algebraic, takahashi1999thermodynamics, levkovich2016bethe}. 

With advancements in quantum information theory and experimental techniques, it has become possible to treat Bethe ansatz on a quantum computer. While the analytic form of given Hamiltonian's eigenstates can be determined through the Bethe ansatz, extracting high-order correlations or complex physical properties from this analytic form remains challenging \cite{giamarchi2003quantum}. Furthermore, highly excited eigenstates within the Bethe ansatz are challenging to express as matrix product states, which are among the most influential numerical tools for studying 1D quantum systems \cite{ruiz2024bethe1}. In this context, several studies have explored the preparation of wavefunction from the Bethe ansatz on a quantum computer, primarily based on the algebraic Bethe ansatz \cite{sopena2022algebraic, ruiz2024bethe1,ruiz2024bethe2}, the coordinate Bethe ansatz \cite{nepomechie2020bethe, van2021preparing, van2022preparing, li2022bethe, raveh2024deterministic, raveh2024estimating, sahu2024fractal}, and the thermodynamic Bethe ansatz \cite{lutz2025adiabatic}. In this section, we focus on preparing a Bethe wavefunction which is closely related to the coordinate Bethe ansatz on a quantum computer using a constant-depth quantum circuit with measurements and feedforward.

Let $L, M$ be the total number of sites and particles, respectively. Then, we follow the definition of the Bethe wavefunction proposed in Ref. \cite{sahu2024fractal} as follows,
\begin{equation}\label{eq:bethe_wavefunction}
    \ket{\psi(\bm{\theta},\vec{k})}  =\sum_{\vec{x}\in \Xi_{M}^L}\sum_{\sigma\in S_\eta}  A_\sigma(\bm{\theta}) e^{i\sum_{l=0}^{M-1} k_{\sigma(l)}x_l}  |\bigoplus_{m=0}^{M-1}e_{x_m}\rangle_L,
\end{equation}
where the coefficients $A_\sigma(\bm{\theta})$ are defined as
\begin{equation}
     A_\sigma(\bm{\theta}) = \exp\bigg(\frac{i}{2}\sum_{l<m}\theta_{\sigma(l),\sigma(m)}\bigg).\label{eq:A}
\end{equation}
We omit the normalization factor of $|\psi(\bm{\theta},\vec{k})\rangle$, due to its length expression. We denote the $M \times M$ matrix $\bm{\theta}$ as scattering matrix, and $\vec{k}$ as quasi momentum or rapidity. We only consider real-valued $\bm{\theta}$ and $\vec{k}$, and assume that they are obtained by solving Bethe equations. Several eigenstates of integrable Hamiltonian can be expressed in the form of Eq. (\ref{eq:bethe_wavefunction}) under the coordinate Bethe ansatz including spin-1/2 periodic XXZ Hamiltonian. However, applicability of our algorithm is not in the integrability, just in the structure of wavefunction and real variables. Of course, integrability is deeply related to two-body reducibility so there is some connection with the above wavefunction. Note that the quantum state without phases corresponds to the Dicke-$(L, M)$ state \cite{bartschi2019deterministic}, which is represented as
\begin{equation}
    \ket{\psi_{\mathrm{Dicke}}} =\frac{1}{\sqrt{\binom{L}{M}}} \sum_{\vec{x}\in \Xi_{M}^L} \ket{\bigoplus_{m=0}^{M-1}e_{x_m}}.
\end{equation}
The Dicke-$(L, M)$ state is a superposition of $L$-qubit states containing $M$ ones and $L-M$ zeros, and is the generalization of the W state. 

\subsection{Preparing Bethe wavefunction}
As observed from Eq. (\ref{eq:bethe_wavefunction}), the Bethe wavefunction possesses a phase with a degree of freedom that increases exponentially with the number of down spins. Despite this huge degree of freedom, the authors in Ref. \cite{van2021preparing} proposed an explicit algorithm for preparing Eq. (\ref{eq:bethe_wavefunction}) using a quantum circuit of depth $O(L)$, $O(M^2)$ auxiliary qubits, and a success probability of $O(1/M!)$ \cite{li2022bethe}. Additionally, the authors in Ref. \cite{raveh2024deterministic} proposed a deterministic algorithm to prepare the Bethe wavefunction using a quantum circuit of depth $O({L \choose M})$. In this section, we propose a constant-depth quantum circuit utilizing measurements and feedforward to prepare the Bethe wavefunction in Eq. (\ref{eq:bethe_wavefunction}), specifically for real-valued $\bm{\theta}$ and $\vec{k}$. 

To prepare the state $\ket{\psi(\bm{\theta},\vec{k})}$, we first require information about permutations and the locations of down spins encoded as a quantum state.
\begin{lemma}\label{lem:resource_state}
    For two positive integers $M$ and $L$ satisfying $M\leq L/2$, the quantum state
    \begin{align}\label{eq:resource_state}
    \nonumber
        &\frac{1}{\sqrt{\binom{M^2}{M}\binom{L}{M}M!}}\sum_{\vec{j}\in \Xi_{M}^{M^2}}\sum_{\sigma \in S_M} \bigg[\ket{j_{\sigma(0)}} \dots \ket{j_{\sigma(M-1)}}  \ket{\sigma(0)}  \\ 
        &\dots \ket{\sigma(M-1)}\ket{\bigoplus_{l=0}^{M-1}e_{j_l}}  \sum_{\vec{x}\in \Xi_{M}^L} \ket{x_0}\dots\ket{x_{M-1}}  |\bigoplus_{m=0}^{M-1}e_{x_m}\rangle\bigg]
    \end{align} 
    can be prepared using constant-depth quantum circuits of width $\tilde{O}(L^2\log L)$ with measurements and feedforward. This preparation can be achieved through a quantum circuit of width $\tilde{O}(L^2\log L)$ for exact Dicke-$(L,M)$ state with the constraint of $M = O(\sqrt{L})$ \cite{buhrman2024state} or a quantum circuit of width $O(Ll_{M,\epsilon})$ for the Dicke-$(L,M)$ state with the infidelity of $\epsilon$, and probability of $O(1/\sqrt{M})$ \cite{piroli2024approximating}, where $l_{M,\epsilon}$ is given by
        \begin{equation}
        l_{M,\epsilon}=\max \bigg\{ \log_2 (4M), 1+\log_2\ln(\sqrt{8\pi M}/\epsilon) \bigg\}.
    \end{equation}
\end{lemma}
\begin{proof}
    See Appendix \ref{app:proof_of_resource_state}.
\end{proof}
If the second subroutine in Lemma \ref{lem:resource_state} is chosen, the approximation can be implemented with high precision due to its $\log\log$ dependence on $1/\epsilon$. Furthermore, the probabilistic nature of the preparation is not a significant issue, as our Bethe state preparation is inevitably probabilistic to avoid exponentially large circuit depths.  

The next step is to attach the phases defined in Eq. (\ref{eq:bethe_wavefunction}).
\begin{lemma}\label{lem:phase_k}
    For two positive integers $M$ and $L$ satisfying $M \leq L$, the transformation
    \begin{align}
        &\frac{1}{\sqrt{\binom{L}{M}M!}}\sum_{\sigma\in S_M}\sum_{\vec{x}\in\Xi_M^L} \ket{e_{\sigma(0)}}\dots\ket{e_{\sigma(M-1)}} \ket{e_{x_0}} \dots\ket{e_{x_{M-1}}} \\ \nonumber
        &\longrightarrow  \frac{1}{\sqrt{\binom{L}{M}M!}}\sum_{\sigma\in S_M}\sum_{\vec{x}\in \Xi^L_M} e^{ i\sum_{l=0}^{M-1} k_{\sigma(l)}x_l }\bigg[\ket{e_{\sigma(0)}} \\
        &\qquad\qquad\qquad\qquad\qquad~~\dots\ket{e_{\sigma(M-1)}}
        \otimes\ket{e_{x_0}}\dots\ket{e_{x_{M-1}}}\bigg]
    \end{align}
    can be implemented using a constant-depth quantum circuit of width $O(M^2 L)$ with measurements and feedforward. 
\end{lemma}
\begin{proof}
 See Appendix \ref{app:proof_of_phase_k}. 
\end{proof}
To attach the phase $A_\sigma$ to the quantum state, the inverse of the permutation must be encoded in advance as follows.
\begin{lemma}\label{lem:inverse_of_permutation}
For a positive integer $M$ and a permutation $\sigma\in S_M$, let $\sigma^{-1}$ denote the inverse of permutation $\sigma$, satisfying $\sigma^{-1}(\sigma(i))=i$ for all $i$. Then, the transformation
\begin{align}
    &\frac{1}{\sqrt{M!}}\sum_{\sigma\in S_M}\ket{\sigma(0)} \dots \ket{\sigma(M-1)}\ket{0}_{\log M}  \dots \ket{0}_{\log M}\\ \nonumber
    &~~\longrightarrow  \frac{1}{\sqrt{M!}}\sum_{\sigma\in S_M}\bigg[\ket{\sigma(0)}\dots \ket{\sigma(M-1)}\\
    & ~~\qquad\qquad\otimes\ket{\sigma^{-1}(0)}_{\log M} \dots \ket{\sigma^{-1}(M-1)}_{\log M}\bigg]
\end{align}
can be implemented using a constant-depth quantum circuit of width $\tilde{O}(M^2 \log M)$ with measurements and feedforward.
\end{lemma}
\begin{proof}
See Appendix \ref{app:proof_of_inverse_of_permutation}. The main components of the quantum circuit are depicted in Fig. \ref{fig:inverse_permutation}.
\end{proof}
\begin{figure}
    \centering
    \includegraphics[width=1.0\linewidth]{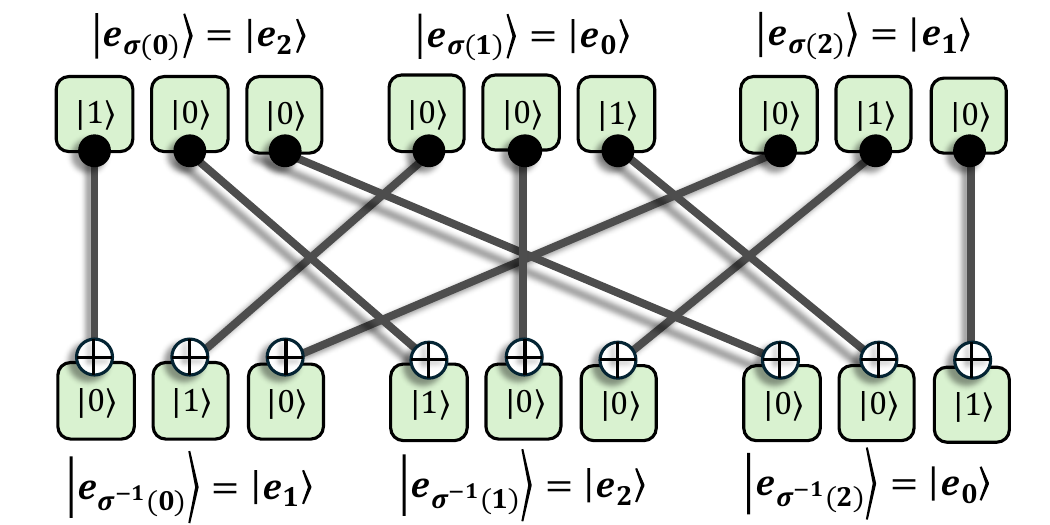}
    \caption{Quantum circuit for encoding inverse of permutation $\sigma$.}
    \label{fig:inverse_permutation}
\end{figure}
This encoded inverse of the permutation enables us to attach $A_\sigma$ to the corresponding quantum state.
\begin{lemma}\label{lem:phase_A}
    For a positive integer $M$, attaching the phase
    \begin{align}
        &\frac{1}{\sqrt{M!}}\sum_{\sigma\in S_M} \ket{\sigma(0)} \dots \ket{\sigma(M-1)} \\
        &~~\longrightarrow  \frac{1}{\sqrt{M!}}\sum_{\sigma\in S_M} A_\sigma \ket{\sigma(0)} \dots  \ket{\sigma(M-1)} \label{eq:A_second}
    \end{align}
    can be implemented using a constant-depth quantum circuit of width $O(M^2 \log^2 M)$ with measurements and feedforward, where $A_\sigma$ is defined in Eq. (\ref{eq:A}). 
\end{lemma}
\begin{proof}
 See Appendix \ref{app:proof_of_phase_A}.
\end{proof}
By combining these lemmas, we now arrive at the final theorem
\begin{theorem}\label{thm:bethe}
    The $L$-sites $M$-particles Bethe wavefunction $\ket{\psi(\bm{\theta},\vec{k})}$ satisfying $M \leq L/2$, can be prepared using a constant-depth quantum circuit of width $\tilde{O}(L^2\log L)$ with measurements and feedforward. This circuit can be structured by preparing either an exact Dicke-$(L,M)$ state with the probability of $O(1/M !)$ and the constraint $M=O(\sqrt{L})$ or an approximate Dicke-$(L,M)$ state with the probability of $O(1/M!)$, infidelity of $O(1/2^{L^L})$ and no constraint.
\end{theorem}
\begin{proof}
See Appendix \ref{app:proof_of_bethe}.
\end{proof} % 
It is important to note that if the second option is chosen, Theorem \ref{thm:bethe} approximates the Bethe wavefunction in principle. However, the $\log\log$ dependence on $1/\epsilon$ allows us to set the infidelity $\epsilon$ as small as $O(1/2^{L^L})$ by utilizing our circuit width of $\tilde{O}(L^2\log L)$, making the approximation error negligible.

% 결론에 우리는 기존에 존재하는 가능한 가장 짧은 circuit depth와 비교해보았을 때도 우리께 SOS 더 짧았다. 그러나 이것이 가장 효율적인 건 아닐 것이고, 최적의 조합을 찾아야한다. 
\section{conclusions}\label{sec:conclusion}
In this study, we employed mid-circuit measurement and feedforward to reduce circuit depth in preparing a quantum state, especially for quantum simulation. Our study leverages the constant-depth implementation of quantum logic gates, at the cost of increased quantum circuit width. This trade-off between quantum circuit depth and width was possible through measurement and feedforward. Utilizing these logic gates, we focused on two key subroutines with constant-depth quantum circuit implementation: the transformation between superpositions of non-negative integer sets and the execution of a series of multi-controlled pairwise commuting gates.

The first subroutine utilizes unary encoding solely as an intermediate step, enabling a constant-depth quantum circuit that has a logarithmic dependence on the maximum integer value. This transformation allows the preparation of arbitrary $d$-sparse $n$-qubit sparse quantum state by a quantum circuit of depth $O(\log d)$ and width $O(dn \log n)$. Furthermore, this subroutine facilitates the preparation of a single (anti)symmetric state in a constant-depth quantum circuit. Supported by these results, $d$-sums of Slater determinants in the first quantization scheme can be prepared using a quantum circuit of depth $O(\log d)$ or $O(d)$, depending on the complexity of the given problem. 

The second subroutine significantly reduces circuit depth in encoding a large degree of freedom into the phase, making it particularly suitable for the preparation of a Bethe wavefunction. Leveraging this capability, the Bethe wavefunction can be prepared using a constant-depth quantum circuit, although its success probability remains $O(1/M!)$ consistent with the previous study. It is worth noting that preparing this state deterministically requires a quantum circuit of depth $O({L \choose M})$, as shown in Ref. \cite{raveh2024deterministic}. Even though the structure of the Bethe wavefunction is closely related to the two-body reducibility, this wavefunction is not restricted to eigenstates of integrable Hamiltonian. We expect that incorporating constant quantum circuit depth arbitrary state preparation algorithm \cite{zi2025constant} may allow extending our methodology to the complex scattering matrix and quasi momentum.

Overall, our work is committed to reducing quantum circuit depth as much as possible, even at the expense of increasing circuit width by leveraging measurements and feedforward. This approach contrasts with conventional methods that typically focus on reducing circuit size or depth without such trade-offs. We emphasize that our strategy is not without its drawbacks; the increased circuit width may incur additional challenges related to qubit connectivity and overall circuit implementation. Nonetheless, we believe that the benefits of our approach are substantial.

A key advantage of our method is its potential to overcome the intrinsic limitations imposed by the physical coherence time of qubits. In the absence of—or with imperfect—quantum error correction, it is crucial to complete quantum circuits within the physical coherence time to mitigate decoherence. According to publicly available data from the IBM Quantum platform, as of the date of this writing, the median $T_2$ of the IBM Kingston system is approximately $125 \mu s$, and the median operation time for a two-qubit gate is about $68 ns$ \cite{ibm_quantum_resources}. In an ideal scenario, this allows for the execution of up to
$\frac{125\mu s}{68ns} \approx 1840$ quantum circuit layers. Therefore, for circuits exceeding this depth, parallelization via measurements and feedforward while keeping the circuit size constant may offer a promising solution. For instance, the authors of Ref.,\cite{baumer2024measurement} experimentally demonstrated that a quantum fan-out gate implemented using measurements and feedforward achieved higher gate fidelity compared to its unitary counterpart, even though this method incurs a modest increase in circuit size as shown in Table \ref{tab:summary}. Furthermore, even in scenarios where quantum error correction is fully established and circuit width constraints are less critical, reducing circuit depth remains the primary objective and is more in line with the fundamental goals of quantum computing.

Moreover, our parallelization techniques are not limited to the quantum state preparation. They can be partially applied to other parts of quantum circuits to reduce the circuit depth. For instance, by using this technique, the number of CNOT gates can be reduced for computing the exponential of Hamiltonians \cite{kaldenbach2024mapping}, and the space-time complexity of Linear Combination of Unitaries (LCU) can also be reduced \cite{boyd2023low}. However, research in this area remains limited. As a future direction, we aim to explore broader applications of these parallelization techniques beyond quantum state preparation. Additionally, identifying a \textit{Goldilocks zone}-where circuit depth and width are optimally balanced in the context of the entire quantum computation-would serve as a promising starting point.

\section*{Acknowledgements}
This work was supported by the National Research Foundation of Korea (NRF) through a grant funded by the Ministry of Science and ICT (NRF-2022M3H3A1098237, RS-2025-00515537), and partially supported by the Institute for Information \& Communications Technology Promotion (IITP) grant funded by the Korean government (MSIP) (No. 2019-0-00003; Research and Development of Core Technologies for Programming, Running, Implementing, and Validating of Fault-Tolerant Quantum Computing Systems). H.E.K. acknowledges support by the education and training program of the Quantum Information Research Support Center, funded through the National research foundation of Korea (NRF) by the Ministry of science and ICT (MSIT) of the Korean government (No.2021M3H3A1036573). I.K.S. acknowledges support by Quantum Computing based on Quantum Advantage challenge research through the National Research Foundation of Korea (NRF) funded by the Korean government (MSIT) (RS-2023-00256221).

\section*{Author Contributions}
H.Y. conceived and suggested the main idea. H.Y., H.E.K., and I.K.S. elaborated on refining the work, and wrote the manuscript together. K.J. supervised the research. All authors discussed and analyzed the results, and contributed to the final manuscript.

\appendix
\section{Proof of Lemma \ref{lem:parity}}\label{app:proof_parity}
The normalization factor and amplitude remain unchanged from those stated in the main text’s lemmas and theorems, and are mostly clear. Therefore, they are omitted for readability. This applies to all appendices. One of the definitions of the parity $\pi(\sigma)$ is the number of inversion in $\sigma\in S_\eta$, which can be represented as below,
\begin{align}
    \pi(\sigma)&=\sum_{0\leq i<j\leq \eta-1} \mathds{1}_{\sigma(i),\sigma(j)}\\
    where~~ \mathds{1}_{\sigma(i),\sigma(j)} &=
    \begin{cases}
        1 & \quad \text{if } \sigma(i)>\sigma(j)\\
        0 & \quad \text{if } \sigma(i)<\sigma(j).
    \end{cases}
\end{align}
Since the \textit{Greatherthan} gate in Table. \ref{tab:operations} can be operated in a constant-depth quantum circuit, we can encode the information of inversions in the quantum state using the \textit{Greatherthan} and quantum fan-out gate in a constant-depth quantum circuit as follows,
\begin{widetext}
\begin{align*}
    &\sum_{\sigma\in S_\eta}\ket{j_{\sigma(0)}}_{\log \zeta}\dots  \ket{j_{\sigma(\eta-1)}}_{\log \zeta} \\
    &~~\longrightarrow  \sum_{\sigma\in S_\eta} \ket{j_{\sigma(0)}}^{\otimes \eta-1}\ket{j_{\sigma(1)}}^{\otimes \eta-1} \dots  \ket{j_{\sigma(\eta-1)}}^{\otimes \eta-1} \\ 
    &~~\longrightarrow \sum_{\sigma\in S_\eta} \bigg[\ket{j_{\sigma(0)}}^{\otimes \eta-1}\bigg( \ket{\mathds{1}_{j_{\sigma(0)}>j_{\sigma(1)}}} \dots \ket{\mathds{1}_{j_{\sigma(0)}>j_{\sigma(\eta-1)})}} \bigg) \\ \nonumber
    &\qquad\qquad~~\otimes \ket{j_{\sigma(1)}}^{\otimes \eta-1}\bigg( \ket{\mathds{1}_{j_{\sigma(1)}>j_{\sigma(2)}}}\dots\ket{\mathds{1}_{j_{\sigma(1)}>j_{\sigma(\eta-1)}}} \bigg)  \dots \ket{j_{\sigma(\eta-2)}}^{\otimes \eta-1} \ket{\mathds{1}_{j_{\sigma(\eta-2)}>j_{\sigma(\eta-1)}}} \ket{j_{\sigma(\eta-1)}}^{\otimes \eta-1}\bigg].
\end{align*}
\end{widetext}
Here, ${\eta \choose 2}$ \textit{Greatherthan} gates, necessiating $\eta-1$ copies of the state $\ket{j_{\sigma(i)}}_{\log \zeta}$, which can be generated using quantum fan-out gates. This operation can be implemented in a constant-depth quantum circuit of width $O(\eta^2 \log^2 \zeta)$ as each \textit{Greatherthan} gate requires $O(\log^2\zeta)$ auxiliary qubits. Since $\{j_i\}_{i=0}^{\eta-1}$ is ordered, applying $Z$ gates to all quantum state $\ket{\mathds{1}_{j_{\sigma(i)}>j_{\sigma(j)}}}$ allows us to attach the phase $(-1)^{\pi(\sigma)}$. Finally, the information about inversions is erased by repeating the above process. Overall, the dominant quantum circuit width is $O(\eta^2\log^2\zeta)$, stemming from applying \textit{Greatherthan} gates ${\eta \choose 2}$ times.

\section{Proof of Theorem \ref{thm:single_symmetric}}\label{app:proof_single_symmetric}
From Lemma \ref{lem:filling_filtering}, we obtain the state
\begin{equation}\label{eq:app_thm2_resource}
    \sum_{\vec{j}\in \Xi_\eta^\zeta}\sum_{\sigma\in S_\eta} \ket{j_{\sigma(0)},\dots,j_{\sigma(\eta-1)}} \dots  \ket{\bigoplus_{i=0}^{\eta-1}e_{j_i}},
\end{equation}
using a constant-depth quantum circuit of width $O(\eta\zeta\log\zeta)$ with measurements and feedforward. If one aims to prepare an antisymmetric state, Lemma \ref{lem:parity} enables the attachment of the phase $(-1)^{\pi(\sigma)}$ via a constant-depth quantum circuit of width $O(\eta^2\log^2\zeta)$. To simplify the notation, we will disregard the phase $(-1)^{\pi(\sigma)}$ in subsequent discussions. From Eq. (\ref{eq:app_thm2_resource}), we can leave only one set of positive integers $\{k_i\}_{i=0}^{\eta-1}$ by measuring and tracing out the unary register $\ket{\bigoplus_{i=0}^{\eta-1}e_{k_i}}$. Then, the state becomes
\begin{equation}\label{eq:single_sym_first}
    \sum_{\sigma\in S_\eta}\ket{k_{\sigma(0)},\dots,k_{\sigma(\eta-1)}},
\end{equation}
where $0\leq k_0<\dots<k_{\eta-1}\leq \zeta-1$. To transform this state into the superposition of the desired integers $\{r_i\}_{i=0}^{\eta-1}$, we have to iteratively apply \textit{Uncompress} and \textit{Compress} to this state as follows, which is described in Fig. \ref{subfig:whole_scheme}.
\begin{enumerate}
    \item For each $\ket{k_{\sigma(i)}}$, apply the \textit{Uncompress} operation using prior knowledge of $\{k_i\}_{i=0}^{\eta-1}$. This can be implemented using a constant-depth quantum circuit of width $O(\eta^2\log\zeta\log\log \zeta)$. The resulting state is
    \begin{equation}\nonumber
        \sum_{\sigma\in S_\eta} \ket{k_{\sigma(0)}}_{\log\zeta}\ket{e_{\sigma(0)}}_{\eta} \dots \ket{k_{\sigma(\eta-1)}}_{\log\zeta}\ket{e_{\sigma(\eta-1)}}_{\eta}.
    \end{equation}
    \item Remove the information of $\{k_i\}_{i=0}^{\eta-1}$ using the \textit{Compress} operation. This step is achieved via a constant-depth quantum circuit of width $O(\eta^2\log\zeta)$, which results as
    \begin{equation}\nonumber
        \sum_{\sigma\in S_\eta} \ket{0}_{\log\zeta}\ket{e_{\sigma(0)}}_{\eta} \dots \ket{0}_{\log\zeta}\ket{e_{\sigma(\eta-1)}}_{\eta}.
    \end{equation}
    \item Using the prior knowledge of $\{r_i\}_{i=0}^{\eta-1}$, apply the inverse of \textit{Compress}. This step requires a constant-depth quantum circuit of width $O(\eta^2\log N\log\log N)$. The resulting state is
    \begin{equation}\nonumber
        \sum_{\sigma\in S_\eta} \ket{r_{\sigma(0)}}_{\log N}\ket{e_{\sigma(0)}}_{\eta} \dots \ket{r_{\sigma(\eta-1)}}_{\log N}\ket{e_{\sigma(\eta-1)}}_{\eta}.
    \end{equation}
   \item Apply the inverse of \textit{Uncompress} to erase information in the unary registers, by a constant-depth quantum circuit of width $O(\eta^2\log N\log\log N)$. The resulting state is
    \begin{equation}\nonumber
        \sum_{\sigma\in S_\eta} \ket{r_{\sigma(0)}}_{\log N}\ket{0}_{\eta} \dots \ket{r_{\sigma(\eta-1)}}_{\log N}\ket{0}_{\eta}.
    \end{equation}
\end{enumerate}
Given the condition $\zeta\geq \eta^2$, we can set $\zeta=\eta^2$ without loss of generality. Under this assumption, the dominant quantum circuit depth is $O(1)$, and the circuit width is $O(\eta^2(\log\eta\log\log \eta+\log N\log\log N))$.

\section{Proof of Theorem \ref{thm:sos}}\label{app:proof_sos}
    Our proof begins with the preparation of a single antisymmetric state
    \begin{equation*}
        \sum_{\sigma\in S_\eta}(-1)^{\pi(\sigma)}\ket{\sigma(0)}\dots\ket{\sigma(\eta-1)},
    \end{equation*}
    using a constant-depth quantum circuit of width $O(\eta^3\log\eta\log\log\eta)$ as established in Theorem \ref{thm:single_symmetric}. To encode the amplitude $\psi_i$ of each Slater determinant $\sum_{\sigma\in S_\eta}\ket{r_{\sigma(0)}^i} \dots\ket{r_{\sigma(\eta-1)}^i}$ into a quantum state, it is necessary to prepare the state $\sum_{i=0}^{d-1}\psi_i\ket{i}_{\log d}$. This preparation can be achieved via a quantum circuit of depth $O(d/\log d)$ and width $O(\log d)$ \cite{sun2023asymptotically} or a quantum circuit of depth $O(\log d)$ and width $O(d)$ \cite{zhang2022quantum}. In this theorem, we propose methods for preparing the $d$-sums of Slater determinants under both approaches, due to the fact that the circuit width polynomial in $d$ may not be compatible with the first quantization scheme.
     
    \subsection{A quantum circuit of width logarithmic in $d$}
    Although the following process will be iterated for $i=0$ to $d-1$, this proof focuses on the case of $i=0$, as the procedure for the other terms is analogous. Furthermore, we will refer to the qubits $\ket{i}$ in the state $\sum_{i=0}^{d-1}\psi_i \ket{i}_{\log d}$ as the coefficient register.
    \begin{enumerate}
        \item Our objective is to entangle each $\ket{i}_{\log d}$ in the coefficient register $\sum_{i=0}^{d-1}\psi_i\ket{i}_{\log d}$, with $\ket{\sigma(0)}\dots\ket{\sigma(\eta-1)}$ in order to encode $\vec{r}^0,\dots,\vec{r}^{d-1}$. To achieve this, we first apply the $Equal_0$ gate, where the coefficient register serves as the control, to generate the quantum state $\ket{\mathds{1}_{i=0}}_1$ which is $\ket{1}_1$ if $i=0$, and $\ket{0}_1$ otherwise. Then, the state becomes
        \begin{equation*}
            \sum_{i=0}^{d-1}\psi_i\ket{i}_{\log d} \sum_{\sigma\in S_\eta} \ket{\mathds{1}_{i=0}}_1\ket{\sigma(0)}\dots\ket{\sigma(\eta-1)}.
        \end{equation*}
        Secondly, we apply quantum fan-out gate to generate $\eta$ copies of the state $\ket{\mathds{1}_{i=0}}_1$,
        \begin{equation*}
            \sum_{i=0}^{d-1}\psi_i\ket{i}_{\log d} \sum_{\sigma\in S_\eta} \ket{\mathds{1}_{i=0}}_1\ket{\sigma(0)}\dots\ket{\mathds{1}_{i=0}}_1\ket{\sigma(\eta-1)}.
        \end{equation*}
        The single $Equal_0$ gate requires a circuit width of $O(\log d \log\log d)$. Additionally, copying the state $\ket{\mathds{1}_{i=0}}_1$ $\eta$ times requires a circuit width of $O(\eta)$.
        \item Apply the \textit{Uncompress} operation to each quantum state $\ket{\mathds{1}_{i=0}}_1\ket{\sigma(j)}$ for $0\leq j \leq \eta-1$. Instead of using the original $Equal_{\eta}$ gate as described in Lemma \ref{lem:uncompress}, we employ the $Equal_{\eta+m}$ gate to solely modify the states associated with $\ket{i=0}_{\log d}$ in the superposition $\sum_{i=0}^{d-1}\psi_i\ket{i}_{\log d}$. Consequently, the corresponding unary registers are activated only for $i=0$ as shown below,
        \begin{align*}
            \sum_{\sigma\in S_\eta}\Bigl[ &\psi_0\ket{0}_{\log d} \ket{1}_1\ket{\sigma(0)}\ket{e_{\sigma(0)}}_{\eta}\dots\ket{1}_1 \ket{\sigma(\eta-1)}\ket{e_{\sigma(\eta-1)}}_\eta \\ \nonumber
            &+ \sum_{i\neq 0}^{d-1}\psi_i\ket{i}_{\log d} \ket{0}_1\ket{\sigma(0)}\ket{0}_{\eta}\dots\ket{0}_1 \ket{\sigma(\eta-1)}\ket{0}_\eta  \Bigr].
        \end{align*}
        These operations require the quantum circuit width of $O(\eta^2\log \eta\log\log \eta)$. 
        
        \item Next, we need to delete the information encoded in $\ket{i=0}_{\log d}$ and $\ket{\mathds{1}_{i=0}}$, leaving only the antisymmetric state. To do so, we first apply a quantum fan-out gate to uncompute $\ket{\mathds{1}_{i=0}}^{\otimes \eta}$ back to $\ket{\mathds{1}_{i=0}}$. Subsequently, we apply quantum fan-out gates, where the controlled qubit is $\ket{\mathds{1}_{i=0}}$, and the target qubits correspond to the $l$-th qubit of the coefficient register if the $l$-th bit of $i=0$ is 1. After these operations, the state becomes
            \begin{align*}
            \sum_{\sigma\in S_\eta}\Bigl[ &\psi_0\ket{0}_{\log d} \ket{1}_1\ket{\sigma(0)} \ket{e_{\sigma(0)}}_{\eta}\dots \ket{\sigma(\eta-1)}\ket{e_{\sigma(\eta-1)}}_\eta \\ \nonumber
            &+ \sum_{i\neq 0}^{d-1}\psi_i\ket{i}_{\log d} \ket{0}_1 \ket{\sigma(0)}\ket{0}_{\eta}\dots \ket{\sigma(\eta-1)}\ket{0}_\eta  \Bigr].
        \end{align*}
        This process can be implemented using a constant-depth quantum circuit of width $O(\log d)$.
        To further remove the information encoded in $\ket{\mathds{1}_{i=0}}_1$, we apply OR gates, where the controlled qubits are $\ket{e_{\sigma(0)}}$, and the target qubit is $\ket{\mathds{1}_{i=0}}$. This operation allows us to successfully uncompute $\ket{\mathds{1}_{i=0}}$, requiring a constant-depth quantum circuit of width $O(\eta\log\eta)$. At this stage, the state becomes
        \begin{align*}
            \sum_{\sigma\in S_\eta}\Bigl[ &\psi_0\ket{0}_{\log d} \ket{0}_1\ket{\sigma(0)}\ket{e_{\sigma(0)}}_{\eta}\dots\ket{\sigma(\eta-1)} \ket{e_{\sigma(\eta-1)}}_\eta \\ \nonumber
            &+ \sum_{i\neq k}^{d-1}\psi_i\ket{i}_{\log d} \ket{0}_1\ket{\sigma(0)}\ket{0}_{\eta}\dots  \ket{\sigma(\eta-1)}\ket{0}_\eta  \Bigr].
        \end{align*}

        \item Henceforth, we dismiss the qubit $\ket{0}_1$ that was previously used to encode $\ket{\mathds{1}_{i=0}}$. To encode $\vec{r}^0$ using the quantum state $\ket{e_{\sigma(j)}}_\eta$, we apply \textit{Compress} operation to each state $\ket{\sigma(j)}$ for $0\leq j \leq \eta-1$. This transforms the state as follows, 
        \begin{align*}
            \sum_{\sigma\in S_\eta}\Bigl[ &\psi_0\ket{0}_{\log d} \ket{0}_{\log\eta}\ket{e_{\sigma(0)}}_{\eta}\dots \ket{0}_{\log\eta}\ket{e_{\sigma(\eta-1)}}_\eta \\ \nonumber
           & + \sum_{i\neq k}^{d-1}\psi_i\ket{i}_{\log d} \ket{\sigma(0)}_{\log\eta}\ket{0}_{\eta}\dots \ket{\sigma(\eta-1)}_{\log\eta}\ket{0}_\eta  \Bigr].
        \end{align*}
        This process can be implemented using a constant-depth quantum circuit of width $O(\eta^2\log\eta)$. Finally, applying the inverse operations of \textit{Compress} and \textit{Uncompress} ensures that only the information of $\vec{r}^0$ remains in the $\psi_0$ part as follows,
        \begin{widetext}
        \begin{align*}
            \sum_{\sigma\in S_\eta}\Bigl[ &\psi_0\ket{0}_{\log d} \ket{r^0_{\sigma(0)}}_{\log N}\ket{0}_{\eta}\dots \ket{r^0_{\sigma(\eta-1)}}_{\log N}\ket{0}_\eta \\ \nonumber
            &+ \sum_{i\neq k}^{d-1}\psi_i\ket{i}_{\log d} \ket{0}_{\log \frac{N}{\eta}} \ket{\sigma(0)}_{\log\eta}\ket{0}_{\eta}\dots\ket{0}_{\log \frac{N}{\eta}} \ket{\sigma(\eta-1)}_{\log \eta}\ket{0}_\eta  \Bigr].
        \end{align*}
        \end{widetext}
        The inverse operations of \textit{Compress} and \textit{Uncompress} can be implemented by a constant-depth quantum circuit of width $O(\eta^2\log N\log\log N)$.
    \end{enumerate}
    Overall, iterating this process from $i=0$ to $d-1$ requires a quantum circuit depth $O(d)$, and the dominant quantum circuit width is given by $O(\eta^2(\eta\log \eta\log\log\eta + \log N\log\log N)+\log d)$. 
    
    \subsection{With the quantum circuit width polynomial in $d$}
    If we aim to utilize a quantum circuit width of $O(d)$, it becomes possible to access the unary encoding of $\sum_{i=0}^{d-1}\ket{e_i}_d$ to reduce the circuit depth. To achieve this, we transform the state $\sum_{i=0}^{d-1}\psi_i\ket{i}_{\log d}$ as $\sum_{i=0}^{d-1}\psi_i\ket{e_i}_d$, which requires $\tilde{O}(d\log d)$ additional qubits. We refer to the qubits that encode $\ket{e_i}$ as the coefficient register.
    \begin{enumerate}
        \item Similar to the previous case, we need to entangle the coefficient register with the state $\ket{\sigma(0)}\dots\ket{\sigma(\eta-1)}$. To avoid a circuit depth of $O(d)$, we first copy the coefficient register $\eta^2$ times, transforming the state as
        \begin{equation*}
            \sum_{i=0}^{d-1}\psi_i \ket{e_i}_d^{\otimes \eta}\ket{\sigma(0)}_{\log\eta}\dots \ket{e_i}_d^{\otimes \eta} \ket{\sigma(\eta-1)}_{\log\eta}.
        \end{equation*}
        This can be implemented using a constant-depth quantum circuit of width $O(\eta^2 d)$. Second, we copy each $\ket{\sigma(j)}$ $\eta d$ times for $0\leq j \leq \eta-1$, where the resulting state is
        \begin{equation}\label{app:proof_sos_fromunary}
            \sum_{i=0}^{d-1}\psi_i \ket{e_i}_d^{\otimes \eta}\ket{\sigma(0)}_{\log\eta}^{\otimes \eta d}\dots \ket{e_i}_d^{\otimes \eta} \ket{\sigma(\eta-1)}_{\log\eta}^{\otimes \eta d}.
        \end{equation}
        The total number of qubits required to prepare this state is $O(\eta^2 d \log\eta))$.
        
        \item For simplicity, let us focus on the $(j+1)$-th pair of quantum states $\ket{e_i}^{\otimes \eta}\ket{\sigma(j)}^{\otimes \eta d}$ from Eq. (\ref{app:proof_sos_fromunary}). Now, we represent this state as follows,
        \begin{equation}\label{app:proof_sos_fromunary2}
            \ket{e_i^0}^{\otimes \eta}\ket{\sigma(j)}^{\otimes \eta}\dots \ket{e_i^{d-1}}^{\otimes \eta}\ket{\sigma(j)}^{\otimes \eta},
        \end{equation} 
        where the state $\ket{e_i^{l}}$ refers to the $(l+1)$-th qubit of $\ket{e_i}$, which is equal to $\ket{\mathds{1}_{i=l}}$. For the $(l+1)$-th pair of quantum states $\ket{e_i^{l}}^{\otimes \eta}\ket{\sigma(j)}^{\otimes \eta}$ in Eq. (\ref{app:proof_sos_fromunary2}), we attach additional $\eta$ qubits $\ket{0}_1^{\otimes \eta}$. Next, we apply the $Equal_{\eta+m}$ gate, where the controlled qubits are the $m$-th $\ket{e_i^{l}}\ket{\sigma(j)}$ pair and the target qubit is the $m$-th $\ket{0}$ among the newly added $\ket{0}_1^{\otimes\eta}$. By operating this $Equal$ gates simultaneously for all $0\leq m,j \leq \eta-1$, and $0\leq l \leq d-1$, we effectively \textit{Uncompress} each $\ket{\sigma(j)}$ while entangling it with each $\ket{e_i}$.
        After these operations, our quantum state becomes
        \begin{widetext}
        \begin{align*}
            &\psi_0\bigotimes_{j=0}^{\eta-1} \ket{e_0^0}^{\otimes \eta}\ket{\sigma(j)}^{\otimes \eta} \ket{e_{\sigma(j)}}_{\eta}\ket{e_0^1}^{\otimes \eta} \ket{\sigma(j)}^{\otimes \eta}\ket{0}_\eta \dots \ket{e_0^{d-1}}^{\otimes \eta}\ket{\sigma(j)}^{\otimes\eta}\ket{0}_\eta\\ \nonumber
            &+ \psi_1 \bigotimes_{j=0}^{\eta-1} \ket{e_1^0}^{\otimes \eta}\ket{\sigma(j)}^{\otimes \eta} \ket{0}_{\eta}\ket{e_1^1}^{\otimes \eta}\ket{\sigma(j)}^{\otimes \eta}\ket{e_{\sigma(j)}}_\eta \dots \ket{e_1^{d-1}}^{\otimes \eta}\ket{\sigma(j)}^{\otimes\eta}\ket{0}_\eta +\dots \\ \nonumber
            &+\psi_{d-1}\bigotimes_{j=0}^{\eta-1} \ket{e_{d-1}^0}^{\otimes \eta}\ket{\sigma(j)}^{\otimes \eta} \ket{0}_{\eta}\ket{e_{d-1}^1}^{\otimes \eta}\ket{\sigma(j)}^{\otimes \eta}\ket{0}_\eta \dots \ket{e_{d-1}^{d-1}}^{\otimes \eta} \ket{\sigma(j)}^{\otimes\eta}\ket{e_{\sigma(j)}}_\eta, 
        \end{align*}
        \end{widetext}
        where the total number of qubits used for $\eta^2 d$ Equal gates is $O(\eta^2 d \log \eta \log\log \eta)$. By removing the redundant information from $\ket{\sigma(j)}^{\otimes\eta d}$ and the coefficient register, the state is transformed into one entangled with  $\psi_i\ket{e_i}$ as follows,
        \begin{equation*}
            \psi_i \bigotimes_{j=0}^{\eta-1} \ket{e_i^0}\ket{\sigma(j)}\ket{0}_\eta \dots \ket{e_k^k} \ket{e_{\sigma(j)}}_\eta\dots \ket{e^{d-1}_0}\ket{0}_\eta.
        \end{equation*}
        This process is equivalent to controlled-\textit{Uncompress} operations, where the target qubits correspond to the $i$-th unary register if the controlled qubit is the $i$-th qubit of the coefficient register.
        
        \item To disentangle the coefficient register from the other states, we utilize the location of $\ket{e_{\sigma(j)}}_{\eta}$. Specifically, we operate \textit{OR} gates, where the controlled qubits are the qubits in the $p$-th unary register, and the target qubit is the $p$-th qubit of the coefficient register. Then, the state corresponding to $\psi_i\ket{e_i}$ becomes
        \begin{align*}
            &\psi_i\bigotimes_{j=0}^{\eta-1} \ket{e^0_i}_1 \ket{\sigma(j)}\ket{0}_{\eta}\dots \ket{e^i_i\oplus 1}_1 \ket{e_{\sigma(j)}}_\eta \dots \ket{0}_\eta 
            %\\&=\psi_i\bigotimes_{j=0}^{\eta-1} \ket{0}_d \ket{\sigma(j)}\ket{0}_{\eta}\dots  \ket{e_{\sigma(j)}}_\eta \dots \ket{0}_\eta.
            \end{align*}
        This process requires $\eta d$ \textit{OR} gates operations, and $O(\eta^2 d \log \eta)$ additional qubits.

        \item Since different $\psi_i$ share the same state $\ket{\sigma(j)}$, we can \textit{Compress} $\ket{\sigma(j)}$ collectively. To achieve this, we attach additional $\eta$ qubits $\ket{0}_1^{\otimes \eta}$ and apply \textit{OR} gates, where the controlled qubits are the $q$-th qubits of each unary registers, and the target qubit is the $q$-th qubit of the newly added $\ket{0}_1^{\otimes \eta}$ for $0\leq q\leq \eta-1$. As a result, the state $\ket{0}_1^{\otimes \eta}$ is transformed into the shared state $\ket{e_{\sigma(j)}}_\eta$ for $0\leq j \leq \eta-1$. After \textit{Compress} through this state, we obtain
        \begin{align*}
           &\psi_0 \bigotimes_{j=0}^{\eta-1} \ket{e_{\sigma(j)}}_\eta\ket{e_{\sigma(j)}}_{\eta} \ket{0}_\eta \dots \ket{0}_\eta\\ \nonumber
            &+\psi_1\bigotimes_{j=0}^{\eta-1} \ket{e_{\sigma(j)}}_\eta\ket{0}_{\eta} \ket{e_{\sigma(j)}}_\eta \dots \ket{0}_\eta +\dots\\ \nonumber
            &+\psi_{d-1}\bigotimes_{j=0}^{\eta-1} \ket{e_{\sigma(j)}}_\eta \ket{0}_{\eta} \ket{0}_\eta \dots \ket{e_{\sigma(j)}}_\eta.
        \end{align*}
        To accomplish this, $\eta^2$ OR gates across $d$ qubits are required. Consequently, the quantum circuit for this operation has a circuit width of $O(\eta^2 d\log d)$. The shared state $\ket{e_{\sigma(j)}}_\eta$ can be easily uncomputed by reversing the above process.
        
        \item Now, we need to encode the information of $\vec{r}^0,\dots,\vec{r}^{d-1}$. To do so, we operate the inverse of \textit{Compress} operations in parallel. First, we attach $\ket{0}_{\log N}$ to all unary registers and apply Hadamard gates to these newly added qubits, resulting as
        \begin{align*}
           &\psi_0 \bigotimes_{j=0}^{\eta-1}\sum_{l_j^0,\dots,l_j^{d-1}=0}^{N-1} \ket{l_j^0} 
           \ket{e_{\sigma(j)}}_{\eta} \ket{l_j^1}\ket{0}_\eta \dots \ket{l_j^{d-1}}\ket{0}_\eta\\ \nonumber
           &+ \psi_1 \bigotimes_{j=0}^{\eta-1}\sum_{l_j^0,\dots,l_j^{d-1}=0}^{N-1} \ket{l_j^0} \ket{0}_{\eta} \ket{l_j^1}\ket{e_{\sigma(j)}}_\eta \dots \ket{l_j^{d-1}}\ket{0}_\eta +\dots\\ \nonumber
            &+\psi_{d-1}  \bigotimes_{j=0}^{\eta-1}\sum_{l_j^0,\dots,l_j^{d-1}=0}^{N-1} \ket{l_j^0} \ket{0}_{\eta} \ket{l_j^1}\ket{0}_\eta \dots \ket{l_j^{d-1}}\ket{e_{\sigma(j)}}_\eta.
        \end{align*}
        The total number of additional qubits required in this process is $O(\eta d \log N)$. Next, we copy each $\ket{l_j^s}$ $\eta d$ times for all $0\leq j \leq \eta-1$ and $0\leq s \leq d-1$. Then, the entire quantum state becomes
        \begin{widetext}
            \begin{align*}
           &\psi_0 \bigotimes_{j=0}^{\eta-1}\sum_{l_j^0,\dots,l_j^{d-1}=0}^{N-1} \ket{e_{\sigma(j)}}_\eta \bigg(\ket{l_j^0}\ket{l_j^1}\dots \ket{l_j^{d-1}}\bigg)^{\otimes \eta} 
           \ket{0}_{\eta}  \bigg(\ket{l_j^0}\ket{l_j^1}\dots \ket{l_j^{d-1}}\bigg)^{\otimes \eta}  \dots \ket{0}_{\eta} \bigg(\ket{l_j^0}\ket{l_j^1}\dots \ket{l_j^{d-1}}\bigg)^{\otimes \eta} \\ \nonumber
           &+ \psi_1 \bigotimes_{j=0}^{\eta-1}\sum_{l_j^0,\dots,l_j^{d-1}=0}^{N-1} \ket{0}_\eta \bigg(\ket{l_j^0}\ket{l_j^1}\dots \ket{l_j^{d-1}}\bigg)^{\otimes \eta} 
           \ket{e_{\sigma(j)}}_{\eta}  \bigg(\ket{l_j^0}\ket{l_j^1}\dots \ket{l_j^{d-1}}\bigg)^{\otimes \eta}  \dots \ket{0}_{\eta} \bigg(\ket{l_j^0}\ket{l_j^1}\dots \ket{l_j^{d-1}}\bigg)^{\otimes \eta} +\dots \\
            &+\psi_{d-1}  \bigotimes_{j=0}^{\eta-1}\sum_{l_j^0,\dots,l_j^{d-1}=0}^{N-1}  \ket{0}_\eta \bigg(\ket{l_j^0}\ket{l_j^1}\dots \ket{l_j^{d-1}}\bigg)^{\otimes \eta} 
           \ket{0}_{\eta}  \bigg(\ket{l_j^0}\ket{l_j^1}\dots \ket{l_j^{d-1}}\bigg)^{\otimes \eta}  \dots \ket{e_{\sigma(j)}}_{\eta} \bigg(\ket{l_j^0}\ket{l_j^1}\dots \ket{l_j^{d-1}}\bigg)^{\otimes \eta}.
        \end{align*}
        \end{widetext}
        The number of additional qubits required for reproducing $\ket{l_i}_{\log N}$ is $O(\eta^2 d^2\log N)$. Next, we apply appropriate CZZ$\dots$Z gates between a unary register and the corresponding phase register $\bigg(\ket{l_j^0}\ket{l_j^1}\dots \ket{l_j^{d-1}}\bigg)^{\otimes \eta}$. The controlled qubit is the $s$-th qubit of the $t$-th unary register and the targets are the qubits in the $s$-th $\bigg(\ket{l_j^0}\ket{l_j^1}\dots \ket{l_j^{d-1}}\bigg)$ of the $t$-th phase register. Similar to the \textit{Compress} operation, we select the $u$-th qubits of each $\ket{l_j^0},\dots,\ket{l_j^{d-1}}$ as target qubits, if the $u$-th bit of $r_t^s$ is 1. Here, the states $\ket{l_j^{v}}$ belong to the $s$-th $\bigg(\ket{l_j^0}\ket{l_j^1}\dots \ket{l_j^{d-1}}\bigg)$. By simultaneously applying these gates for $0\leq t \leq d-1$, $0\leq s \leq \eta-1$ and $0\leq u \leq \log N -1$, and uncomputing the copied $\ket{l_j^v}$s, the resulting quantum state becomes
        \begin{widetext}
            \begin{align*}
           &\psi_0 \bigotimes_{j=0}^{\eta-1}\sum_{l_j^0,\dots,l_j^{d-1}=0}^{N-1} (-1)^{r_{\sigma(j)}^0\cdot l_j^0}\ket{e_{\sigma(j)}}_\eta \ket{l_j^0}
           (-1)^{r_{\sigma(j)}^0\cdot l_j^1}\ket{0}_{\eta} \ket{l_j^1}\dots (-1)^{r_{\sigma(j)}^0\cdot l_j^{d-1}}\ket{0}_{\eta}\ket{l_j^{d-1}} \\ \nonumber
           &+  \psi_1 \bigotimes_{j=0}^{\eta-1}\sum_{l_j^0,\dots,l_j^{d-1}=0}^{N-1} (-1)^{r_{\sigma(j)}^1\cdot l_j^0}\ket{e_{\sigma(j)}}_\eta \ket{l_j^0}
           (-1)^{r_{\sigma(j)}^1\cdot l_j^1}\ket{0}_{\eta} \ket{l_j^1}\dots (-1)^{r_{\sigma(j)}^1\cdot l_j^{d-1}}\ket{0}_{\eta}\ket{l_j^{d-1}}+\dots
            \\ \nonumber
            &+  \psi_{d-1} \bigotimes_{j=0}^{\eta-1}\sum_{l_j^0,\dots,l_j^{d-1}=0}^{N-1} (-1)^{r_{\sigma(j)}^{d-1}\cdot l_j^0}\ket{e_{\sigma(j)}}_\eta \ket{l_j^0}
           (-1)^{r_{\sigma(j)}^{d-1}\cdot l_j^1}\ket{0}_{\eta} \ket{l_j^1}\dots (-1)^{r_{\sigma(j)}^{d-1}\cdot l_j^{d-1}}\ket{0}_{\eta}\ket{l_j^{d-1}}.
        \end{align*}
        \end{widetext}
        Given that the maximum number of target qubits of each CZZ$\dots$Z gate is $d\log N$, and there are $\eta^2 d$ CZZ$\dots$Z gates in this process, the operation can be implemented using a constant-depth quantum circuit of width $O(\eta^2 d^2\log N)$. At this stage, we are ready to encode $\vec{r}^0,\dots,\vec{r}^{d-1}$ into a quantum state by simply applying Hadamard gates to all $\ket{l_j^v}s$ states as below,
        \begin{align*}
            &\psi_0\bigotimes_{j=0}^{\eta-1}  \ket{r_{\sigma(j)}^0} \ket{e_{\sigma(j)}}_\eta \ket{r_{\sigma(j)}^0}\ket{0}_\eta \dots\ket{r_{\sigma(j)}^0}\ket{0}_\eta +\dots \\ \nonumber
            &+\psi_{d-1}\bigotimes_{j=0}^{\eta-1}\ket{r_{\sigma(j)}^{d-1}} \ket{0}_\eta \ket{r_{\sigma(j)}^{d-1}} \ket{0}_\eta\dots\ket{r_{\sigma(j)}^{d-1}}\ket{e_{\sigma(j)}}_\eta.
        \end{align*}
        %&+\psi_1\bigotimes_{j=0}^{\eta-1} \ket{r_{\sigma(j)}^1} \ket{0}_\eta \ket{r_{\sigma(j)}^1}\ket{e_{\sigma(j)}}_\eta\dots\ket{r_{\sigma(j)}^1}\ket{0}_\eta \\ \nonumber
        \item Now, we need to disentangle the unary registers from the binary registers. This can be achieved by extracting the information of the permutation $\sigma$ from the binary registers, as all binary registers share the same permutation. Specifically, we can apply \textit{Uncompress} operation 1) from the zeroth binary register to the zeroth unary register for each $j$ where the target integer set is $\vec{r}^0$, 2) from the first binary register to the corresponding one where the target integer set is $\vec{r}^1$, and so on. However, this approach fails if there are repetitions among different integer sets. To address this issue, we mark the set $\vec{r}^0$ in the $i$-th binary register for each $j$, by copying binary registers and applying the $Equal_{r^i_k}$ gates to each copy for $0\leq k \leq \eta-1$. If the state of the binary register is $\ket{r^i}$, the output register of the $Equal_{r^i_k}$ gates will become $\ket{11\dots 1}$. From now on, we refer to the output register of $Equal$ gates as the index register. Applying \textit{AND} gate to the index register allows us to transform the state as follows,
        \begin{widetext}
        \begin{align*}
           &\psi_0 \bigotimes_{j=0}^{\eta-1}\sum_{l_j^0,\dots,l_j^{d-1}=0}^{N-1} \ket{r^0_{\sigma(j)}}^{\otimes \eta} \bigg[\ket{\mathds{1}_{r^0_{\sigma(j)}=r^0_0}}\dots\ket{\mathds{1}_{r^0_{\sigma(j)}=r^0_{\eta-1}}}\ket{1}_1 \bigg] \ket{e_{\sigma(j)}}_\eta \dots \ket{r^0_{\sigma(j)}}^{\otimes \eta} \bigg[\ket{\mathds{1}_{r^0_{\sigma(j)}=r^{d-1}_0}}\dots\ket{\mathds{1}_{r^0_{\sigma(j)}=r^{d-1}_{\eta-1}}}\ket{0}_1\bigg]\ket{0}_\eta \\ \nonumber
           &+\dots\\ \nonumber
          &+\psi_{d-1} \bigotimes_{j=0}^{\eta-1}\sum_{l_j^0,\dots,l_j^{d-1}=0}^{N-1} \ket{r^{d-1}_{\sigma(j)}}^{\otimes \eta} \bigg[\ket{\mathds{1}_{r^{d-1}_{\sigma(j)}=r^0_0}}\dots\ket{\mathds{1}_{r^{d-1}_{\sigma(j)}=r^0_{\eta-1}}}\ket{0}_1 \bigg] \ket{0}_\eta \dots \ket{r^{d-1}_{\sigma(j)}}^{\otimes \eta} \bigg[\ket{\mathds{1}_{r^{d-1}_{\sigma(j)}=r^{d-1}_0}}\dots\ket{\mathds{1}_{r^{d-1}_{\sigma(j)}=r^{d-1}_{\eta-1}}}\ket{1}_1\bigg]\ket{e_{\sigma(j)}}_\eta.
        \end{align*}
        \end{widetext}
        This operation requires $O(\eta^2 d \log N\log\log N)$ auxiliary qubits. At this stage, the index register of the $i$-th unary register for any $j$ is $\ket{1}_1$ only if the corresponding coefficient is $\psi_i$. After copying each index register $\eta$ times, we apply the \textit{Uncompress} operation, where both the index register and the binary register serve as control register, and the target register is the corresponding unary register. This process removes the information of $\ket{e_{\sigma(j)}}_\eta$ from the unary register without changing other states, and it requires $O(\eta^2 d \log N)$ auxiliary qubits. Finally, we obtain the desired state after uncomputing all redundant states.
        \end{enumerate}
        
\section{Proof of Lemma \ref{lem:resource_state}}\label{app:proof_of_resource_state}
The first summand of the desired state, involving $\sum_{\vec{j}\in \Xi_{M}^{M^2}} \sum_{\sigma\in S_M}$ but excluding the state $\ket{\sigma(0)},\dots\ket{\sigma(M-1)}$ can be obtained using Lemma \ref{lem:filling_filtering}. We can obtain the state $\ket{\sigma(0)},\dots\ket{\sigma(M-1)}$ in the same way as we describe in the second version of the proof of Theorem \ref{thm:sos}, requiring $O(M^2\log M \log\log L)$ auxiliary qubits in this case. Therefore, the first part of the desired quantum state can be prepared using a constant-depth quantum circuit of width $O(M^3\log M)$ with measurements and feedforward. For the second part of the state, involving $\sum_{\vec{x}\in \Xi_{M}^L}$, we prepare the Dicke-$(L,M)$ state in a constant-depth quantum circuit as described in ref. \cite{buhrman2024state} or Ref. \cite{piroli2024approximating}. Finally, we apply the inverse of the \textit{Cleaning} operation described in Ref. \cite{buhrman2024state}, which transforms the Dicke state into the desired state using a constant-depth quantum circuit of width $\tilde{O}(L^2\log L)$.
    
The algorithm in Ref. \cite{buhrman2024state} deterministically prepares exact Dicke-$(L,M)$ using a constant-depth quantum circuit of width $\tilde{O}(L^2\log L)$ with measurements and feedforward. However, this method imposes the constraint $M\leq\sqrt{L}$. In contrast, the algorithms in Ref. \cite{piroli2024approximating} approximate the Dicke state with the infidelity of $\epsilon$ through different methods, each utilizing distinct resources. These methods do not impose a constraint such as $M\leq\sqrt{L}$. To reduce the circuit depth, we select a constant-depth quantum circuit among three of them, which has the circuit width of $O(Ll_{M,\epsilon})$ and a success probability of $O(1/\sqrt{M})$. Given that the form of $l_{M,\epsilon}$ which is 
\begin{equation*}
    l_{M,\epsilon}=\max \bigg\{ \log_2 (4M), 1+\log_2\ln(\sqrt{8\pi M}/\epsilon) \bigg\},
\end{equation*}
and reusability of our additional qubits, the infidelity $\epsilon$ can be decreased to extremely small value if we compare $l_{M,\epsilon}$ to $L\log L$.

\section{Proof of Lemma \ref{lem:phase_k}}\label{app:proof_of_phase_k}
Since the information of $\sigma(i)$ and $x_i$ is already encoded in the location, one can attach $\exp(k_{\sigma(i)}x_i)$ by applying controlled-phase gates $CP(\alpha, \beta)$ as follows,
\begin{equation*}
    CP(\alpha,\beta)=
    \begin{pmatrix}
    1 & 0 & 0 & 0\\
    0 & 1 & 0 & 0\\
    0 & 0 & 1 & 0\\
    0 & 0 & 0 & e^{ik_\alpha \beta}
    \end{pmatrix}.
\end{equation*}
Here, the controlled qubit is the $\alpha$-th qubits of $\ket{e_{\sigma(i)}}$, and the target qubit is the $\beta$-th qubit of $\ket{e_{x_i}}$, respectively. By applying these gates to all $0\leq i,\alpha \leq M-1,$ and $0\leq\beta\leq L-1$, we can attach the phase $\exp(\sum_{i=0}^{M-1}k_{\sigma(i)}x_i)$. Using Theorem \ref{lem:controlled}, these operations can be parallelized by creating replicas using quantum fan-out gates as follows,
\begin{widetext}
\begin{align*}
    &\sum_{\vec{x}\in\Xi_M^L} \sum_{\sigma\in S_M} \ket{e_{\sigma(0)}}\dots\ket{e_{\sigma(M-1)}} \ket{e_{x_0}}\dots\ket{e_{x_{M-1}}}\\
    &~~\longrightarrow  \sum_{\vec{x}\in\Xi_M^L} \sum_{\sigma\in S_M}\ket{e_{\sigma(0)}}^{\otimes L}\dots\ket{e_{\sigma(M-1)}}^{\otimes L} \ket{e_{x_0}}^{\otimes M}\dots\ket{e_{x_{M-1}}}^{\otimes M}\\
    &~~\longrightarrow \sum_{\vec{x}\in\Xi_M^L} \sum_{\sigma\in S_M} \exp\bigg(\sum_{i=0}^{M-1}k_{\sigma(i)}x_i\bigg)
    \ket{e_{\sigma(0)}}^{\otimes L}\dots\ket{e_{\sigma(M-1)}}^{\otimes L}  \ket{e_{x_0}}^{\otimes M}\dots\ket{e_{x_{M-1}}}^{\otimes M}
\end{align*}
\end{widetext}
The dominant width of this quantum circuit is $O(M^2 L)$.

\section{Proof of Lemma \ref{lem:inverse_of_permutation}}\label{app:proof_of_inverse_of_permutation}
First, we apply \textit{Uncompress} operation to the state $\ket{\sigma(i)}$,
\begin{align*}
    &\sum_{\sigma\in S_M}\ket{\sigma(0)}\dots\ket{\sigma(M-1)}\\
    &~~\longrightarrow  \sum_{\sigma\in S_M}\ket{\sigma(0)}\dots\ket{\sigma(M-1)}\ket{e_{\sigma(0)}}\dots\ket{e_{\sigma(M-1)}},
\end{align*}
where the quantum circuit width of this operation is $O(M^2\log M)$. Now we add $M^2$ qubits as follows,
\begin{equation*}
    \sum_{\sigma\in S_M}\ket{\sigma(0)}\dots\ket{\sigma(M-1)}\ket{e_{\sigma(0)}}\dots\ket{e_{\sigma(M-1)}}\ket{0}_{M}\dots\ket{0}_M,
\end{equation*}
and we notate the $i$-th $\ket{0}_M$ as the $i$-th inverse register. To prepare the quantum state $\ket{e_{\sigma^{-1}(i)}}$, we apply a CX gate, where the controlled qubit is the $i$-th qubit of $\ket{e_{\sigma(j)}}$, and the target qubit is the $j$-th qubit of the $i$-th inverse register. Iterating this operations for all $0\leq i,j\leq M-1$ in parallel transforms the state as
\begin{align*}
    &\sum_{\sigma\in S_M} \ket{\sigma(0)}\dots \ket{\sigma(M-1)}\ket{e_{\sigma(0)}}\dots\ket{e_{\sigma(M-1)}} \\ \nonumber
    &\qquad~\otimes  \ket{e_{\sigma^{-1}(0)}}\dots\ket{e_{\sigma^{-1}(M-1)}},
\end{align*}
and this process is well depicted in Fig. \ref{fig:inverse_permutation}.

Finally, we apply the inverse of \textit{Uncompress} operation to remove the state $\ket{e_{\sigma(0)}}\dots\ket{e_{\sigma(M-1)}}$, and transform the state $\ket{e_{\sigma^{-1}(i)}}$ to $\ket{\sigma^{-1}(i)}$. We obtain the desired state in a constant-depth quantum circuit of width $O(M^2\log M\log\log M)$ using measurements and feedforward. 

\section{Proof of Lemma \ref{lem:phase_A}}\label{app:proof_of_phase_A}
Due to the properties of the Bethe equation, $\theta_{i,j}=-\theta_{j,i}$ holds true, and every possible $\theta_{i,j}$ for all $i\neq j$ is included in $A_\sigma$. Thus, all we need to determine is whether the permutation $\sigma$ satisfies $\sigma(a)>\sigma(b)$ for two integers $0\leq a<b\leq M-1$. This information can be encoded in the quantum state by applying the \textit{Greatherthan} gates to the inverse of permutation $\sigma^{-1}$, as described in Lemma \ref{lem:inverse_of_permutation}. Once encoded, the phase $A_\sigma$ can be attached to the state by applying the appropriate phase gates. This process can be outlined in detail as follows.
\begin{enumerate}
    \item Using Lemma \ref{lem:inverse_of_permutation} to encode the inverse of permutation into a quantum state
    \begin{equation*}
        \sum_{\sigma\in S_M}\ket{\sigma(0)}\dots\ket{\sigma(M-1)}\ket{\sigma^{-1}(0)}\dots\ket{\sigma^{-1}(M-1)},
    \end{equation*}
    which can be implemented in a constant-depth quantum circuit of width $\tilde{O}(M^2\log M)$ using measurements and feedforward.
    \item From now on, we neglect the state $\ket{\sigma(j)}$, as it is not used in the remainder of our proof. By using quantum fan-out gates, we copy $\ket{\sigma^{-1}(i)}$ for all $i$ to obtain
    \begin{equation*}
        \sum_{\sigma\in S_M} \ket{\sigma^{-1}(0)}^{\otimes M-1}\dots\ket{\sigma^{-1}(M-1)}^{\otimes M-1},
    \end{equation*}
    where the quantum circuit width of $O(M^2 \log M)$ is required for this process.
    \item Apply \textit{Greatherthan} gates to all possible pairs of $\ket{\sigma^{-1}(i)}$ and $\ket{\sigma^{-1}(j)}$. The resulting state is represented as
    \begin{widetext}
    \begin{align*}
        &\sum_{\sigma\in S_M} \bigg[\ket{\sigma^{-1}(0)}^{\otimes M-1}\big[\ket{\mathds{1}_{\sigma^{-1}(0)<\sigma^{-1}(1)}}\dots\ket{\mathds{1}_{\sigma^{-1}(0)<\sigma^{-1}(M-1)}} \big]\\ \nonumber
        &\quad~~~~\otimes \ket{\sigma^{-1}(1)}^{\otimes M-1}\big[\ket{\mathds{1}_{\sigma^{-1}(1)<\sigma^{-1}(2)}}\dots\ket{\mathds{1}_{\sigma^{-1}(1)<\sigma^{-1}(M-1)}}\big]\otimes \dots \\ \nonumber
        &\quad~~~~\otimes  \ket{\sigma^{-1}(M-2)}^{\otimes M-1} \big[\ket{\mathds{1}_{\sigma^{-1}(M-2)<\sigma^{-1}(M-1)}}\big]\otimes\ket{\sigma^{-1}(M-1)}^{\otimes M-1}\bigg],
    \end{align*}
    \end{widetext}
    which requires $O(M^2(\log M)^2)$ quantum circuit width.
    \item Let us define a phase gate $P_\theta(i,j)$ as 
    \begin{equation}
    P_\theta(i,j)=
    \begin{pmatrix}
    \exp(i\theta_{j,i}/2) & 0 \\
    0 & \exp(i\theta_{i,j/2}) \\
    \end{pmatrix},
    \end{equation}
    By operating phase gates on the targets of the \textit{Greatherthan} gates, which are $\ket{\mathds{1}_{\sigma^{-1}(i)<\sigma^{-1}(j)}}$, we attach the appropriate phase $\exp(\frac{i}{2}\theta_{i,j})$ or $\exp(\frac{i}{2}\theta_{j,i})$, irrespective of the specific permutation $\sigma$. This is implemented simultaneously for all $0\leq i < j \leq M-1$. Finally, we uncompute all states except $\ket{\sigma(0)}\dots\ket{\sigma(M-1)}$ to obtain the desired state of Eq. (\ref{eq:A_second}). The dominant quantum circuit width arises from applying the \textit{Greatherthan} gates, which is $O(M^2\log^2M)$.
\end{enumerate}

\section{Proof of Theorem \ref{thm:bethe}}\label{app:proof_of_bethe}
    The detailed process of our proof can be outlined as follows:
    \begin{enumerate}
        \item By using Lemma \ref{lem:resource_state}, we obtain the state 
        \begin{align*} 
 &\sum_{\vec{j}\in \Xi_{M}^{M^2}}\sum_{\sigma \in S_M} \bigg[\ket{j_{\sigma(0)}} \dots  \ket{j_{\sigma(M-1)}}   \ket{\sigma(0)}\dots \ket{\sigma(M-1)}\ket{\bigoplus_{i=0}^{M-1}e_{j_i}}\\ \nonumber
        &\qquad\qquad~~~~\otimes \sum_{\vec{x}\in \Xi_M^L}  \ket{x_0}\dots\ket{x_{M-1}} \ket{\bigoplus_{k=0}^{M-1}e_{x_k}}\bigg]
        \end{align*}
        by a constant-depth quantum circuit of width $\tilde{O}(L^2\log L)$. While preparing this state, let us assume that we use the second subroutine described in Lemma \ref{lem:resource_state}. Even though this subroutine approximates the Dicke state, we can achieve an almost exact preparation with an infidelity of $O(2^{-L^L})$, by utilizing the full quantum circuit width of $\tilde{O}(L^2\log L)$ to minimize the infidelity. As a byproduct, this process requires iteration until it succeeds, with a success probability of $O(1/\sqrt{M})$.
    \item After performing \textit{Uncompress} operation on $\ket{x_0}\dots\ket{x_{M-1}}$ using a quantum circuit of width $\tilde{O}(M L\log L)$, we attach the phases $A_\sigma \exp (i\sum_{i=0}^{M-1} k_{\sigma(i)}x_i )$ by Lemmas \ref{lem:phase_k} and \ref{lem:phase_A}. This phase attachment can be implemented by a quantum circuit of width $\tilde{O}(M^2 (L+\log^2 M))$. By applying the inverse of \textit{Uncompress} operation to the state $\ket{e_{x_0}}\dots\ket{e_{x_{M-1}}}$, we obtain the state
    \begin{widetext}
    \begin{align*}
        \sum_{\sigma\in S_M}\sum_{\vec{x}\in \Xi_{M}^L}\sum_{\vec{j}\in \Xi_{M}^{M^2}} A_\sigma \exp\bigg(i\sum_{l=0}^{M-1} k_{\sigma(l)}x_l\bigg)
        \ket{j_{\sigma(0)}}\dots \ket{j_{\sigma(M-1)}} 
  \ket{\sigma(0)}\dots\ket{\sigma(M-1)} \ket{\bigoplus_{m=0}^{M-1}e_{j_m}} \ket{x_0}\dots\ket{x_{M-1}} \ket{\bigoplus_{p=0}^{M-1}e_{x_p}}.
    \end{align*}
    \end{widetext}
    \item Finally, we need to remove the information about the permutation $\sigma$ and the state $\ket{x_0}\dots\ket{x_{M-1}}$, only retaining the complete Bethe wavefunction. The latter can be removed by performing \textit{Cleaning} operation described in \cite{buhrman2024state}, which requires a quantum circuit of width $\tilde{O}(L^2\log L)$. For the former, in the same sprit as the Linear combination of Unitaries (LCU), applying the inverse gates used to encode the phases transforms the state into
    \begin{equation*}
        \sum_{\sigma\in S_M}\sum_{\vec{x}\in \Xi_{M}^L} A_\sigma \exp\bigg(i\sum_{l=0}^{M-1}k_{\sigma(l)}x_l\bigg)\ket{00\dots 0}\ket{\bigoplus_{p=0}^{M-1}e_{x_p}}+\ket{\psi_\bot}, 
    \end{equation*}
    where the state $\ket{\psi_\bot}$ refers to the residual state orthogonal to $\ket{00\dots 0}\ket{\bigoplus_{p=0}^{M-1}e_{x_p}}$. The probability of measuring $\ket{00\dots 0}$ in this case is known to be $O(1/M!)$ \cite{li2022bethe}.
    \end{enumerate}
    Overall, the dominant cost of preparing Bethe wavefunction arises from the \textit{Cleaning} operation, which is $\tilde{O}(L^2\log L)$, and the total probability of success is $O(\frac{1}{M!\sqrt{M}}).$

\nocite{*}

\bibliography{references}% Produces the bibliography via BibTeX.

\begin{thebibliography}{95}%
\makeatletter
\providecommand \@ifxundefined [1]{%
 \@ifx{#1\undefined}
}%
\providecommand \@ifnum [1]{%
 \ifnum #1\expandafter \@firstoftwo
 \else \expandafter \@secondoftwo
 \fi
}%
\providecommand \@ifx [1]{%
 \ifx #1\expandafter \@firstoftwo
 \else \expandafter \@secondoftwo
 \fi
}%
\providecommand \natexlab [1]{#1}%
\providecommand \enquote  [1]{``#1''}%
\providecommand \bibnamefont  [1]{#1}%
\providecommand \bibfnamefont [1]{#1}%
\providecommand \citenamefont [1]{#1}%
\providecommand \href@noop [0]{\@secondoftwo}%
\providecommand \href [0]{\begingroup \@sanitize@url \@href}%
\providecommand \@href[1]{\@@startlink{#1}\@@href}%
\providecommand \@@href[1]{\endgroup#1\@@endlink}%
\providecommand \@sanitize@url [0]{\catcode `\\12\catcode `\$12\catcode `\&12\catcode `\#12\catcode `\^12\catcode `\_12\catcode `\%12\relax}%
\providecommand \@@startlink[1]{}%
\providecommand \@@endlink[0]{}%
\providecommand \url  [0]{\begingroup\@sanitize@url \@url }%
\providecommand \@url [1]{\endgroup\@href {#1}{\urlprefix }}%
\providecommand \urlprefix  [0]{URL }%
\providecommand \Eprint [0]{\href }%
\providecommand \doibase [0]{https://doi.org/}%
\providecommand \selectlanguage [0]{\@gobble}%
\providecommand \bibinfo  [0]{\@secondoftwo}%
\providecommand \bibfield  [0]{\@secondoftwo}%
\providecommand \translation [1]{[#1]}%
\providecommand \BibitemOpen [0]{}%
\providecommand \bibitemStop [0]{}%
\providecommand \bibitemNoStop [0]{.\EOS\space}%
\providecommand \EOS [0]{\spacefactor3000\relax}%
\providecommand \BibitemShut  [1]{\csname bibitem#1\endcsname}%
\let\auto@bib@innerbib\@empty
%</preamble>
\bibitem [{\citenamefont {Piroli}\ \emph {et~al.}(2021)\citenamefont {Piroli}, \citenamefont {Styliaris},\ and\ \citenamefont {Cirac}}]{piroli2021quantum}%
  \BibitemOpen
  \bibfield  {author} {\bibinfo {author} {\bibfnamefont {L.}~\bibnamefont {Piroli}}, \bibinfo {author} {\bibfnamefont {G.}~\bibnamefont {Styliaris}},\ and\ \bibinfo {author} {\bibfnamefont {J.~I.}\ \bibnamefont {Cirac}},\ }\bibfield  {title} {\bibinfo {title} {{Quantum {Circuits} Assisted by Local Operations and Classical Communication: Transformations and Phases of Matter}},\ }\href {https://journals.aps.org/prl/abstract/10.1103/PhysRevLett.127.220503} {\bibfield  {journal} {\bibinfo  {journal} {Phys. Rev. Lett.}\ }\textbf {\bibinfo {volume} {127}},\ \bibinfo {pages} {220503} (\bibinfo {year} {2021})}\BibitemShut {NoStop}%
\bibitem [{\citenamefont {Lu}\ \emph {et~al.}(2022)\citenamefont {Lu}, \citenamefont {Lessa}, \citenamefont {Kim},\ and\ \citenamefont {Hsieh}}]{lu2022measurement}%
  \BibitemOpen
  \bibfield  {author} {\bibinfo {author} {\bibfnamefont {T.-C.}\ \bibnamefont {Lu}}, \bibinfo {author} {\bibfnamefont {L.~A.}\ \bibnamefont {Lessa}}, \bibinfo {author} {\bibfnamefont {I.~H.}\ \bibnamefont {Kim}},\ and\ \bibinfo {author} {\bibfnamefont {T.~H.}\ \bibnamefont {Hsieh}},\ }\bibfield  {title} {\bibinfo {title} {{Measurement as a Shortcut to Long-Range Entangled Quantum Matter}},\ }\href {https://journals.aps.org/prxquantum/abstract/10.1103/PRXQuantum.3.040337} {\bibfield  {journal} {\bibinfo  {journal} {PRX Quantum}\ }\textbf {\bibinfo {volume} {3}},\ \bibinfo {pages} {040337} (\bibinfo {year} {2022})}\BibitemShut {NoStop}%
\bibitem [{\citenamefont {Smith}\ \emph {et~al.}(2024)\citenamefont {Smith}, \citenamefont {Khan}, \citenamefont {Clark}, \citenamefont {Girvin},\ and\ \citenamefont {Wei}}]{smith2024constant}%
  \BibitemOpen
  \bibfield  {author} {\bibinfo {author} {\bibfnamefont {K.~C.}\ \bibnamefont {Smith}}, \bibinfo {author} {\bibfnamefont {A.}~\bibnamefont {Khan}}, \bibinfo {author} {\bibfnamefont {B.~K.}\ \bibnamefont {Clark}}, \bibinfo {author} {\bibfnamefont {S.}~\bibnamefont {Girvin}},\ and\ \bibinfo {author} {\bibfnamefont {T.-C.}\ \bibnamefont {Wei}},\ }\bibfield  {title} {\bibinfo {title} {{Constant-Depth Preparation of Matrix Product States with Adaptive Quantum Circuits}},\ }\href {https://journals.aps.org/prxquantum/abstract/10.1103/PRXQuantum.5.030344} {\bibfield  {journal} {\bibinfo  {journal} {PRX Quantum}\ }\textbf {\bibinfo {volume} {5}},\ \bibinfo {pages} {030344} (\bibinfo {year} {2024})}\BibitemShut {NoStop}%
\bibitem [{\citenamefont {Foss-Feig}\ \emph {et~al.}()\citenamefont {Foss-Feig}, \citenamefont {Tikku}, \citenamefont {Lu}, \citenamefont {Mayer}, \citenamefont {Iqbal}, \citenamefont {Gatterman}, \citenamefont {Gerber}, \citenamefont {Gilmore}, \citenamefont {Gresh}, \citenamefont {Hankin} \emph {et~al.}}]{foss2023experimental}%
  \BibitemOpen
  \bibfield  {author} {\bibinfo {author} {\bibfnamefont {M.}~\bibnamefont {Foss-Feig}}, \bibinfo {author} {\bibfnamefont {A.}~\bibnamefont {Tikku}}, \bibinfo {author} {\bibfnamefont {T.-C.}\ \bibnamefont {Lu}}, \bibinfo {author} {\bibfnamefont {K.}~\bibnamefont {Mayer}}, \bibinfo {author} {\bibfnamefont {M.}~\bibnamefont {Iqbal}}, \bibinfo {author} {\bibfnamefont {T.~M.}\ \bibnamefont {Gatterman}}, \bibinfo {author} {\bibfnamefont {J.~A.}\ \bibnamefont {Gerber}}, \bibinfo {author} {\bibfnamefont {K.}~\bibnamefont {Gilmore}}, \bibinfo {author} {\bibfnamefont {D.}~\bibnamefont {Gresh}}, \bibinfo {author} {\bibfnamefont {A.}~\bibnamefont {Hankin}}, \emph {et~al.},\ }\href@noop {} {\bibinfo {title} {{Experimental demonstration of the advantage of adaptive quantum circuits}}},\ \Eprint {https://arxiv.org/abs/2302.03029} {arXiv:2302.03029} \BibitemShut {NoStop}%
\bibitem [{\citenamefont {B{\"a}umer}\ \emph {et~al.}(2024)\citenamefont {B{\"a}umer}, \citenamefont {Tripathi}, \citenamefont {Wang}, \citenamefont {Rall}, \citenamefont {Chen}, \citenamefont {Majumder}, \citenamefont {Seif},\ and\ \citenamefont {Minev}}]{baumer2024efficient}%
  \BibitemOpen
  \bibfield  {author} {\bibinfo {author} {\bibfnamefont {E.}~\bibnamefont {B{\"a}umer}}, \bibinfo {author} {\bibfnamefont {V.}~\bibnamefont {Tripathi}}, \bibinfo {author} {\bibfnamefont {D.~S.}\ \bibnamefont {Wang}}, \bibinfo {author} {\bibfnamefont {P.}~\bibnamefont {Rall}}, \bibinfo {author} {\bibfnamefont {E.~H.}\ \bibnamefont {Chen}}, \bibinfo {author} {\bibfnamefont {S.}~\bibnamefont {Majumder}}, \bibinfo {author} {\bibfnamefont {A.}~\bibnamefont {Seif}},\ and\ \bibinfo {author} {\bibfnamefont {Z.~K.}\ \bibnamefont {Minev}},\ }\bibfield  {title} {\bibinfo {title} {{Efficient Long-Range Entanglement using Dynamic Circuits}},\ }\href {https://journals.aps.org/prxquantum/abstract/10.1103/PRXQuantum.5.030339} {\bibfield  {journal} {\bibinfo  {journal} {PRX Quantum}\ }\textbf {\bibinfo {volume} {5}},\ \bibinfo {pages} {030339} (\bibinfo {year} {2024})}\BibitemShut {NoStop}%
\bibitem [{\citenamefont {Malz}\ \emph {et~al.}(2024)\citenamefont {Malz}, \citenamefont {Styliaris}, \citenamefont {Wei},\ and\ \citenamefont {Cirac}}]{malz2024preparation}%
  \BibitemOpen
  \bibfield  {author} {\bibinfo {author} {\bibfnamefont {D.}~\bibnamefont {Malz}}, \bibinfo {author} {\bibfnamefont {G.}~\bibnamefont {Styliaris}}, \bibinfo {author} {\bibfnamefont {Z.-Y.}\ \bibnamefont {Wei}},\ and\ \bibinfo {author} {\bibfnamefont {J.~I.}\ \bibnamefont {Cirac}},\ }\bibfield  {title} {\bibinfo {title} {{Preparation of Matrix Product States with Log-Depth Quantum Circuits}},\ }\href {https://journals.aps.org/prl/abstract/10.1103/PhysRevLett.132.040404} {\bibfield  {journal} {\bibinfo  {journal} {Phys. Rev. Lett.}\ }\textbf {\bibinfo {volume} {132}},\ \bibinfo {pages} {040404} (\bibinfo {year} {2024})}\BibitemShut {NoStop}%
\bibitem [{\citenamefont {Bravyi}\ \emph {et~al.}()\citenamefont {Bravyi}, \citenamefont {Kim}, \citenamefont {Kliesch},\ and\ \citenamefont {Koenig}}]{bravyi2022adaptive}%
  \BibitemOpen
  \bibfield  {author} {\bibinfo {author} {\bibfnamefont {S.}~\bibnamefont {Bravyi}}, \bibinfo {author} {\bibfnamefont {I.}~\bibnamefont {Kim}}, \bibinfo {author} {\bibfnamefont {A.}~\bibnamefont {Kliesch}},\ and\ \bibinfo {author} {\bibfnamefont {R.}~\bibnamefont {Koenig}},\ }\href@noop {} {\bibinfo {title} {{Adaptive constant-depth circuits for manipulating non-abelian anyons}}},\ \Eprint {https://arxiv.org/abs/2205.01933} {arXiv:2205.01933} \BibitemShut {NoStop}%
\bibitem [{\citenamefont {Tantivasadakarn}\ \emph {et~al.}(2023)\citenamefont {Tantivasadakarn}, \citenamefont {Vishwanath},\ and\ \citenamefont {Verresen}}]{tantivasadakarn2023hierarchy}%
  \BibitemOpen
  \bibfield  {author} {\bibinfo {author} {\bibfnamefont {N.}~\bibnamefont {Tantivasadakarn}}, \bibinfo {author} {\bibfnamefont {A.}~\bibnamefont {Vishwanath}},\ and\ \bibinfo {author} {\bibfnamefont {R.}~\bibnamefont {Verresen}},\ }\bibfield  {title} {\bibinfo {title} {{Hierarchy of Topological Order from Finite-Depth Unitaries, Measurement, and Feedforward}},\ }\href {https://journals.aps.org/prxquantum/abstract/10.1103/PRXQuantum.4.020339} {\bibfield  {journal} {\bibinfo  {journal} {PRX Quantum}\ }\textbf {\bibinfo {volume} {4}},\ \bibinfo {pages} {020339} (\bibinfo {year} {2023})}\BibitemShut {NoStop}%
\bibitem [{\citenamefont {Boyd}()}]{boyd2023low}%
  \BibitemOpen
  \bibfield  {author} {\bibinfo {author} {\bibfnamefont {G.}~\bibnamefont {Boyd}},\ }\href@noop {} {\bibinfo {title} {{Low-Overhead Parallelisation of LCU via Commuting Operators}}},\ \Eprint {https://arxiv.org/abs/2312.00696} {arXiv:2312.00696} \BibitemShut {NoStop}%
\bibitem [{\citenamefont {Gidney}\ and\ \citenamefont {Jones}()}]{gidney2021cccz}%
  \BibitemOpen
  \bibfield  {author} {\bibinfo {author} {\bibfnamefont {C.}~\bibnamefont {Gidney}}\ and\ \bibinfo {author} {\bibfnamefont {N.~C.}\ \bibnamefont {Jones}},\ }\href@noop {} {\bibinfo {title} {{A CCCZ gate performed with 6 T gates}}},\ \Eprint {https://arxiv.org/abs/2106.11513} {arXiv:2106.11513} \BibitemShut {NoStop}%
\bibitem [{\citenamefont {Kim}\ \emph {et~al.}()\citenamefont {Kim}, \citenamefont {Baek}, \citenamefont {Hwang},\ and\ \citenamefont {Bang}}]{kim2024resource}%
  \BibitemOpen
  \bibfield  {author} {\bibinfo {author} {\bibfnamefont {T.}~\bibnamefont {Kim}}, \bibinfo {author} {\bibfnamefont {K.}~\bibnamefont {Baek}}, \bibinfo {author} {\bibfnamefont {Y.}~\bibnamefont {Hwang}},\ and\ \bibinfo {author} {\bibfnamefont {J.}~\bibnamefont {Bang}},\ }\href@noop {} {\bibinfo {title} {{Resource-compact time-optimal quantum computation}}},\ \Eprint {https://arxiv.org/abs/2405.00191} {arXiv:2405.00191} \BibitemShut {NoStop}%
\bibitem [{\citenamefont {Smith}\ \emph {et~al.}(2023)\citenamefont {Smith}, \citenamefont {Crane}, \citenamefont {Wiebe},\ and\ \citenamefont {Girvin}}]{smith2023deterministic}%
  \BibitemOpen
  \bibfield  {author} {\bibinfo {author} {\bibfnamefont {K.~C.}\ \bibnamefont {Smith}}, \bibinfo {author} {\bibfnamefont {E.}~\bibnamefont {Crane}}, \bibinfo {author} {\bibfnamefont {N.}~\bibnamefont {Wiebe}},\ and\ \bibinfo {author} {\bibfnamefont {S.~M.}\ \bibnamefont {Girvin}},\ }\bibfield  {title} {\bibinfo {title} {{Deterministic Constant-Depth Preparation of the AKLT State on a Quantum Processor Using Fusion Measurements}},\ }\href {https://journals.aps.org/prxquantum/abstract/10.1103/PRXQuantum.4.020315} {\bibfield  {journal} {\bibinfo  {journal} {PRX Quantum}\ }\textbf {\bibinfo {volume} {4}},\ \bibinfo {pages} {020315} (\bibinfo {year} {2023})}\BibitemShut {NoStop}%
\bibitem [{\citenamefont {Piroli}\ \emph {et~al.}(2024)\citenamefont {Piroli}, \citenamefont {Styliaris},\ and\ \citenamefont {Cirac}}]{piroli2024approximating}%
  \BibitemOpen
  \bibfield  {author} {\bibinfo {author} {\bibfnamefont {L.}~\bibnamefont {Piroli}}, \bibinfo {author} {\bibfnamefont {G.}~\bibnamefont {Styliaris}},\ and\ \bibinfo {author} {\bibfnamefont {J.~I.}\ \bibnamefont {Cirac}},\ }\bibfield  {title} {\bibinfo {title} {{Approximating Many-Body Quantum States with Quantum Circuits and Measurements}},\ }\href {https://doi.org/10.1103/PhysRevLett.133.230401} {\bibfield  {journal} {\bibinfo  {journal} {Phys. Rev. Lett.}\ }\textbf {\bibinfo {volume} {133}},\ \bibinfo {pages} {230401} (\bibinfo {year} {2024})}\BibitemShut {NoStop}%
\bibitem [{\citenamefont {Gottesman}\ and\ \citenamefont {Chuang}(1999)}]{gottesman1999demonstrating}%
  \BibitemOpen
  \bibfield  {author} {\bibinfo {author} {\bibfnamefont {D.}~\bibnamefont {Gottesman}}\ and\ \bibinfo {author} {\bibfnamefont {I.~L.}\ \bibnamefont {Chuang}},\ }\bibfield  {title} {\bibinfo {title} {{Demonstrating the viability of universal quantum computation using teleportation and single-qubit operations}},\ }\href {https://www.nature.com/articles/46503} {\bibfield  {journal} {\bibinfo  {journal} {Nature}\ }\textbf {\bibinfo {volume} {402}},\ \bibinfo {pages} {390} (\bibinfo {year} {1999})}\BibitemShut {NoStop}%
\bibitem [{\citenamefont {Jozsa}(2006)}]{jozsa2006introduction}%
  \BibitemOpen
  \bibfield  {author} {\bibinfo {author} {\bibfnamefont {R.}~\bibnamefont {Jozsa}},\ }\bibfield  {title} {\bibinfo {title} {{An introduction to measurement based quantum computation}},\ }\href@noop {} {\bibfield  {journal} {\bibinfo  {journal} {NATO Science Series, III: Computer and Systems Sciences. Quantum Information Processing-From Theory to Experiment}\ }\textbf {\bibinfo {volume} {199}},\ \bibinfo {pages} {137} (\bibinfo {year} {2006})}\BibitemShut {NoStop}%
\bibitem [{\citenamefont {Low}\ \emph {et~al.}(2024)\citenamefont {Low}, \citenamefont {Kliuchnikov},\ and\ \citenamefont {Schaeffer}}]{low2024trading}%
  \BibitemOpen
  \bibfield  {author} {\bibinfo {author} {\bibfnamefont {G.~H.}\ \bibnamefont {Low}}, \bibinfo {author} {\bibfnamefont {V.}~\bibnamefont {Kliuchnikov}},\ and\ \bibinfo {author} {\bibfnamefont {L.}~\bibnamefont {Schaeffer}},\ }\bibfield  {title} {\bibinfo {title} {{Trading T gates for dirty qubits in state preparation and unitary synthesis}},\ }\href {https://quantum-journal.org/papers/q-2024-06-17-1375/} {\bibfield  {journal} {\bibinfo  {journal} {Quantum}\ }\textbf {\bibinfo {volume} {8}},\ \bibinfo {pages} {1375} (\bibinfo {year} {2024})}\BibitemShut {NoStop}%
\bibitem [{\citenamefont {Buhrman}\ \emph {et~al.}(2024)\citenamefont {Buhrman}, \citenamefont {Folkertsma}, \citenamefont {Loff},\ and\ \citenamefont {Neumann}}]{buhrman2024state}%
  \BibitemOpen
  \bibfield  {author} {\bibinfo {author} {\bibfnamefont {H.}~\bibnamefont {Buhrman}}, \bibinfo {author} {\bibfnamefont {M.}~\bibnamefont {Folkertsma}}, \bibinfo {author} {\bibfnamefont {B.}~\bibnamefont {Loff}},\ and\ \bibinfo {author} {\bibfnamefont {N.~M.~P.}\ \bibnamefont {Neumann}},\ }\bibfield  {title} {\bibinfo {title} {{State preparation by shallow circuits using feed forward}},\ }\href {https://quantum-journal.org/papers/q-2024-12-09-1552/} {\bibfield  {journal} {\bibinfo  {journal} {Quantum}\ }\textbf {\bibinfo {volume} {8}},\ \bibinfo {pages} {1552} (\bibinfo {year} {2024})}\BibitemShut {NoStop}%
\bibitem [{\citenamefont {H{\o}yer}\ and\ \citenamefont {{\v{S}}palek}(2005)}]{hoyer2005quantum}%
  \BibitemOpen
  \bibfield  {author} {\bibinfo {author} {\bibfnamefont {P.}~\bibnamefont {H{\o}yer}}\ and\ \bibinfo {author} {\bibfnamefont {R.}~\bibnamefont {{\v{S}}palek}},\ }\bibfield  {title} {\bibinfo {title} {{Quantum Fan-out is Powerful}},\ }\href {https://theoryofcomputing.org/articles/v001a005/} {\bibfield  {journal} {\bibinfo  {journal} {Theory of Computing}\ }\textbf {\bibinfo {volume} {1}},\ \bibinfo {pages} {81} (\bibinfo {year} {2005})}\BibitemShut {NoStop}%
\bibitem [{\citenamefont {Takahashi}\ and\ \citenamefont {Tani}(2016)}]{takahashi2016collapse}%
  \BibitemOpen
  \bibfield  {author} {\bibinfo {author} {\bibfnamefont {Y.}~\bibnamefont {Takahashi}}\ and\ \bibinfo {author} {\bibfnamefont {S.}~\bibnamefont {Tani}},\ }\bibfield  {title} {\bibinfo {title} {{Collapse of the Hierarchy of Constant-Depth Exact Quantum Circuits}},\ }\href {https://link.springer.com/article/10.1007/s00037-016-0140-0} {\bibfield  {journal} {\bibinfo  {journal} {comput. complex.}\ }\textbf {\bibinfo {volume} {25}},\ \bibinfo {pages} {849} (\bibinfo {year} {2016})}\BibitemShut {NoStop}%
\bibitem [{\citenamefont {Kitaev}()}]{kitaev1995quantum}%
  \BibitemOpen
  \bibfield  {author} {\bibinfo {author} {\bibfnamefont {A.~Y.}\ \bibnamefont {Kitaev}},\ }\href@noop {} {\bibinfo {title} {{Quantum measurements and the Abelian stabilizer problem}}},\ \Eprint {https://arxiv.org/abs/quant-ph/9511026} {arXiv:quant-ph/9511026} \BibitemShut {NoStop}%
\bibitem [{\citenamefont {Abrams}\ and\ \citenamefont {Lloyd}(1999)}]{abrams1999quantum}%
  \BibitemOpen
  \bibfield  {author} {\bibinfo {author} {\bibfnamefont {D.~S.}\ \bibnamefont {Abrams}}\ and\ \bibinfo {author} {\bibfnamefont {S.}~\bibnamefont {Lloyd}},\ }\bibfield  {title} {\bibinfo {title} {{Quantum Algorithm Providing Exponential Speed Increase for Finding Eigenvalues and Eigenvectors}},\ }\href {https://journals.aps.org/prl/abstract/10.1103/PhysRevLett.83.5162} {\bibfield  {journal} {\bibinfo  {journal} {Phys. Rev. Lett.}\ }\textbf {\bibinfo {volume} {83}},\ \bibinfo {pages} {5162} (\bibinfo {year} {1999})}\BibitemShut {NoStop}%
\bibitem [{\citenamefont {Dong}\ \emph {et~al.}(2022)\citenamefont {Dong}, \citenamefont {Lin},\ and\ \citenamefont {Tong}}]{dong2022ground}%
  \BibitemOpen
  \bibfield  {author} {\bibinfo {author} {\bibfnamefont {Y.}~\bibnamefont {Dong}}, \bibinfo {author} {\bibfnamefont {L.}~\bibnamefont {Lin}},\ and\ \bibinfo {author} {\bibfnamefont {Y.}~\bibnamefont {Tong}},\ }\bibfield  {title} {\bibinfo {title} {{Ground-State Preparation and Energy Estimation on Early Fault-Tolerant Quantum Computers via Quantum Eigenvalue Transformation of Unitary Matrices}},\ }\href {https://journals.aps.org/prxquantum/abstract/10.1103/PRXQuantum.3.040305} {\bibfield  {journal} {\bibinfo  {journal} {PRX Quantum}\ }\textbf {\bibinfo {volume} {3}},\ \bibinfo {pages} {040305} (\bibinfo {year} {2022})}\BibitemShut {NoStop}%
\bibitem [{\citenamefont {Wang}\ \emph {et~al.}(2023)\citenamefont {Wang}, \citenamefont {Fran{\c{c}}a}, \citenamefont {Zhang}, \citenamefont {Zhu},\ and\ \citenamefont {Johnson}}]{wang2023quantum}%
  \BibitemOpen
  \bibfield  {author} {\bibinfo {author} {\bibfnamefont {G.}~\bibnamefont {Wang}}, \bibinfo {author} {\bibfnamefont {D.~S.}\ \bibnamefont {Fran{\c{c}}a}}, \bibinfo {author} {\bibfnamefont {R.}~\bibnamefont {Zhang}}, \bibinfo {author} {\bibfnamefont {S.}~\bibnamefont {Zhu}},\ and\ \bibinfo {author} {\bibfnamefont {P.~D.}\ \bibnamefont {Johnson}},\ }\bibfield  {title} {\bibinfo {title} {{Quantum algorithm for ground state energy estimation using circuit depth with exponentially improved dependence on precision}},\ }\href {https://quantum-journal.org/papers/q-2023-11-06-1167/} {\bibfield  {journal} {\bibinfo  {journal} {Quantum}\ }\textbf {\bibinfo {volume} {7}},\ \bibinfo {pages} {1167} (\bibinfo {year} {2023})}\BibitemShut {NoStop}%
\bibitem [{\citenamefont {Romero}\ \emph {et~al.}(2018)\citenamefont {Romero}, \citenamefont {Babbush}, \citenamefont {McClean}, \citenamefont {Hempel}, \citenamefont {Love},\ and\ \citenamefont {Aspuru-Guzik}}]{romero2018strategies}%
  \BibitemOpen
  \bibfield  {author} {\bibinfo {author} {\bibfnamefont {J.}~\bibnamefont {Romero}}, \bibinfo {author} {\bibfnamefont {R.}~\bibnamefont {Babbush}}, \bibinfo {author} {\bibfnamefont {J.~R.}\ \bibnamefont {McClean}}, \bibinfo {author} {\bibfnamefont {C.}~\bibnamefont {Hempel}}, \bibinfo {author} {\bibfnamefont {P.~J.}\ \bibnamefont {Love}},\ and\ \bibinfo {author} {\bibfnamefont {A.}~\bibnamefont {Aspuru-Guzik}},\ }\bibfield  {title} {\bibinfo {title} {{Strategies for quantum computing molecular energies using the unitary coupled cluster ansatz}},\ }\href {https://iopscience.iop.org/article/10.1088/2058-9565/aad3e4/meta?casa_token=5oPxT7f_q_AAAAAA:sxQ4FTPg0ugPiyzyUq3fTLnCEz0k1ajRE7nkcj4qHTusxfvLmcD4Sqc65VH7hGF2kDnk-zgZOtsy6NkK3wZv3T6n8lYLNg} {\bibfield  {journal} {\bibinfo  {journal} {Quantum Sci. Technol.}\ }\textbf {\bibinfo {volume} {4}},\ \bibinfo {pages} {014008} (\bibinfo {year} {2018})}\BibitemShut {NoStop}%
\bibitem [{\citenamefont {Fomichev}\ \emph {et~al.}(2024)\citenamefont {Fomichev}, \citenamefont {Hejazi}, \citenamefont {Zini}, \citenamefont {Kiser}, \citenamefont {Fraxanet}, \citenamefont {Casares}, \citenamefont {Delgado}, \citenamefont {Huh}, \citenamefont {Voigt}, \citenamefont {Mueller},\ and\ \citenamefont {Arrazola}}]{fomichev2023initial}%
  \BibitemOpen
  \bibfield  {author} {\bibinfo {author} {\bibfnamefont {S.}~\bibnamefont {Fomichev}}, \bibinfo {author} {\bibfnamefont {K.}~\bibnamefont {Hejazi}}, \bibinfo {author} {\bibfnamefont {M.~S.}\ \bibnamefont {Zini}}, \bibinfo {author} {\bibfnamefont {M.}~\bibnamefont {Kiser}}, \bibinfo {author} {\bibfnamefont {J.}~\bibnamefont {Fraxanet}}, \bibinfo {author} {\bibfnamefont {P.~A.~M.}\ \bibnamefont {Casares}}, \bibinfo {author} {\bibfnamefont {A.}~\bibnamefont {Delgado}}, \bibinfo {author} {\bibfnamefont {J.}~\bibnamefont {Huh}}, \bibinfo {author} {\bibfnamefont {A.-C.}\ \bibnamefont {Voigt}}, \bibinfo {author} {\bibfnamefont {J.~E.}\ \bibnamefont {Mueller}},\ and\ \bibinfo {author} {\bibfnamefont {J.~M.}\ \bibnamefont {Arrazola}},\ }\bibfield  {title} {\bibinfo {title} {{Initial State Preparation for Quantum Chemistry on Quantum Computers}},\ }\href {https://doi.org/10.1103/PRXQuantum.5.040339} {\bibfield  {journal} {\bibinfo  {journal} {PRX Quantum}\ }\textbf {\bibinfo {volume} {5}},\ \bibinfo {pages} {040339}
  (\bibinfo {year} {2024})}\BibitemShut {NoStop}%
\bibitem [{\citenamefont {Albash}\ and\ \citenamefont {Lidar}(2018)}]{albash2018adiabatic}%
  \BibitemOpen
  \bibfield  {author} {\bibinfo {author} {\bibfnamefont {T.}~\bibnamefont {Albash}}\ and\ \bibinfo {author} {\bibfnamefont {D.~A.}\ \bibnamefont {Lidar}},\ }\bibfield  {title} {\bibinfo {title} {{Adiabatic quantum computation}},\ }\href {https://journals.aps.org/rmp/abstract/10.1103/RevModPhys.90.015002} {\bibfield  {journal} {\bibinfo  {journal} {Rev. Mod. Phys.}\ }\textbf {\bibinfo {volume} {90}},\ \bibinfo {pages} {015002} (\bibinfo {year} {2018})}\BibitemShut {NoStop}%
\bibitem [{\citenamefont {Ebadi}\ \emph {et~al.}(2021)\citenamefont {Ebadi}, \citenamefont {Wang}, \citenamefont {Levine}, \citenamefont {Keesling}, \citenamefont {Semeghini}, \citenamefont {Omran}, \citenamefont {Bluvstein}, \citenamefont {Samajdar}, \citenamefont {Pichler}, \citenamefont {Ho} \emph {et~al.}}]{ebadi2021quantum}%
  \BibitemOpen
  \bibfield  {author} {\bibinfo {author} {\bibfnamefont {S.}~\bibnamefont {Ebadi}}, \bibinfo {author} {\bibfnamefont {T.~T.}\ \bibnamefont {Wang}}, \bibinfo {author} {\bibfnamefont {H.}~\bibnamefont {Levine}}, \bibinfo {author} {\bibfnamefont {A.}~\bibnamefont {Keesling}}, \bibinfo {author} {\bibfnamefont {G.}~\bibnamefont {Semeghini}}, \bibinfo {author} {\bibfnamefont {A.}~\bibnamefont {Omran}}, \bibinfo {author} {\bibfnamefont {D.}~\bibnamefont {Bluvstein}}, \bibinfo {author} {\bibfnamefont {R.}~\bibnamefont {Samajdar}}, \bibinfo {author} {\bibfnamefont {H.}~\bibnamefont {Pichler}}, \bibinfo {author} {\bibfnamefont {W.~W.}\ \bibnamefont {Ho}}, \emph {et~al.},\ }\bibfield  {title} {\bibinfo {title} {{Quantum phases of matter on a 256-atom programmable quantum simulator}},\ }\href {https://www.nature.com/articles/s41586-021-03582-4} {\bibfield  {journal} {\bibinfo  {journal} {Nature}\ }\textbf {\bibinfo {volume} {595}},\ \bibinfo {pages} {227} (\bibinfo {year} {2021})}\BibitemShut {NoStop}%
\bibitem [{\citenamefont {{\v{C}}epait{\.e}}\ \emph {et~al.}(2023)\citenamefont {{\v{C}}epait{\.e}}, \citenamefont {Polkovnikov}, \citenamefont {Daley},\ and\ \citenamefont {Duncan}}]{vcepaite2023counterdiabatic}%
  \BibitemOpen
  \bibfield  {author} {\bibinfo {author} {\bibfnamefont {I.}~\bibnamefont {{\v{C}}epait{\.e}}}, \bibinfo {author} {\bibfnamefont {A.}~\bibnamefont {Polkovnikov}}, \bibinfo {author} {\bibfnamefont {A.~J.}\ \bibnamefont {Daley}},\ and\ \bibinfo {author} {\bibfnamefont {C.~W.}\ \bibnamefont {Duncan}},\ }\bibfield  {title} {\bibinfo {title} {{Counterdiabatic Optimized Local Driving}},\ }\href {https://journals.aps.org/prxquantum/abstract/10.1103/PRXQuantum.4.010312} {\bibfield  {journal} {\bibinfo  {journal} {PRX Quantum}\ }\textbf {\bibinfo {volume} {4}},\ \bibinfo {pages} {010312} (\bibinfo {year} {2023})}\BibitemShut {NoStop}%
\bibitem [{\citenamefont {van Vreumingen}\ and\ \citenamefont {Schoutens}(2023)}]{van2023adiabatic}%
  \BibitemOpen
  \bibfield  {author} {\bibinfo {author} {\bibfnamefont {D.}~\bibnamefont {van Vreumingen}}\ and\ \bibinfo {author} {\bibfnamefont {K.}~\bibnamefont {Schoutens}},\ }\bibfield  {title} {\bibinfo {title} {{Adiabatic ground-state preparation of fermionic many-body systems from a two-body perspective}},\ }\href {https://journals.aps.org/pra/abstract/10.1103/PhysRevA.108.062603} {\bibfield  {journal} {\bibinfo  {journal} {Phys. Rev. A}\ }\textbf {\bibinfo {volume} {108}},\ \bibinfo {pages} {062603} (\bibinfo {year} {2023})}\BibitemShut {NoStop}%
\bibitem [{\citenamefont {Mc~Keever}\ and\ \citenamefont {Lubasch}(2024)}]{mc2024towards}%
  \BibitemOpen
  \bibfield  {author} {\bibinfo {author} {\bibfnamefont {C.}~\bibnamefont {Mc~Keever}}\ and\ \bibinfo {author} {\bibfnamefont {M.}~\bibnamefont {Lubasch}},\ }\bibfield  {title} {\bibinfo {title} {{Towards Adiabatic Quantum Computing Using Compressed Quantum Circuits}},\ }\href {https://journals.aps.org/prxquantum/abstract/10.1103/PRXQuantum.5.020362} {\bibfield  {journal} {\bibinfo  {journal} {PRX Quantum}\ }\textbf {\bibinfo {volume} {5}},\ \bibinfo {pages} {020362} (\bibinfo {year} {2024})}\BibitemShut {NoStop}%
\bibitem [{\citenamefont {Zhang}\ \emph {et~al.}(2022)\citenamefont {Zhang}, \citenamefont {Li},\ and\ \citenamefont {Yuan}}]{zhang2022quantum}%
  \BibitemOpen
  \bibfield  {author} {\bibinfo {author} {\bibfnamefont {X.-M.}\ \bibnamefont {Zhang}}, \bibinfo {author} {\bibfnamefont {T.}~\bibnamefont {Li}},\ and\ \bibinfo {author} {\bibfnamefont {X.}~\bibnamefont {Yuan}},\ }\bibfield  {title} {\bibinfo {title} {{Quantum State Preparation with Optimal Circuit Depth: Implementations and Applications}},\ }\href {https://journals.aps.org/prl/abstract/10.1103/PhysRevLett.129.230504} {\bibfield  {journal} {\bibinfo  {journal} {Phys. Rev. Lett.}\ }\textbf {\bibinfo {volume} {129}},\ \bibinfo {pages} {230504} (\bibinfo {year} {2022})}\BibitemShut {NoStop}%
\bibitem [{\citenamefont {Zhang}\ and\ \citenamefont {Yuan}(2024)}]{zhang2024circuit}%
  \BibitemOpen
  \bibfield  {author} {\bibinfo {author} {\bibfnamefont {X.-M.}\ \bibnamefont {Zhang}}\ and\ \bibinfo {author} {\bibfnamefont {X.}~\bibnamefont {Yuan}},\ }\bibfield  {title} {\bibinfo {title} {{Circuit complexity of quantum access models for encoding classical data}},\ }\href {https://www.nature.com/articles/s41534-024-00835-8} {\bibfield  {journal} {\bibinfo  {journal} {npj Quantum Inf.}\ }\textbf {\bibinfo {volume} {10}},\ \bibinfo {pages} {42} (\bibinfo {year} {2024})}\BibitemShut {NoStop}%
\bibitem [{\citenamefont {Sahu}\ and\ \citenamefont {Vidal}()}]{sahu2024fractal}%
  \BibitemOpen
  \bibfield  {author} {\bibinfo {author} {\bibfnamefont {S.}~\bibnamefont {Sahu}}\ and\ \bibinfo {author} {\bibfnamefont {G.}~\bibnamefont {Vidal}},\ }\href@noop {} {\bibinfo {title} {{Fractal decompositions and tensor network representations of Bethe wavefunctions}}},\ \Eprint {https://arxiv.org/abs/2412.00923} {arXiv:2412.00923} \BibitemShut {NoStop}%
\bibitem [{\citenamefont {Bethe}(1931)}]{bethe1931theorie}%
  \BibitemOpen
  \bibfield  {author} {\bibinfo {author} {\bibfnamefont {H.}~\bibnamefont {Bethe}},\ }\bibfield  {title} {\bibinfo {title} {{Zur theorie der metalle: I. Eigenwerte und eigenfunktionen der linearen atomkette}},\ }\href@noop {} {\bibfield  {journal} {\bibinfo  {journal} {Zeitschrift f{\"u}r Physik}\ }\textbf {\bibinfo {volume} {71}},\ \bibinfo {pages} {205} (\bibinfo {year} {1931})}\BibitemShut {NoStop}%
\bibitem [{\citenamefont {Giamarchi}(2003)}]{giamarchi2003quantum}%
  \BibitemOpen
  \bibfield  {author} {\bibinfo {author} {\bibfnamefont {T.}~\bibnamefont {Giamarchi}},\ }\href@noop {} {\emph {\bibinfo {title} {{Quantum physics in one dimension}}}},\ Vol.\ \bibinfo {volume} {121}\ (\bibinfo  {publisher} {Clarendon press},\ \bibinfo {year} {2003})\BibitemShut {NoStop}%
\bibitem [{\citenamefont {van Tongeren}(2016)}]{van2016introduction}%
  \BibitemOpen
  \bibfield  {author} {\bibinfo {author} {\bibfnamefont {S.~J.}\ \bibnamefont {van Tongeren}},\ }\bibfield  {title} {\bibinfo {title} {{Introduction to the thermodynamic Bethe ansatz}},\ }\href {https://iopscience.iop.org/article/10.1088/1751-8113/49/32/323005/meta?casa_token=7EL-8HSk6FQAAAAA:F1KhEsVcXdvn5ZLBwDZ9o-aQta-ps5IyZkgCqFCzLDuNwWRl2anr45qxaY3iH7tJ6M8C2IuRG8rNaf2-QZt-t0x5zZaFmA} {\bibfield  {journal} {\bibinfo  {journal} {J. Phys. A: Math. Theor.}\ }\textbf {\bibinfo {volume} {49}},\ \bibinfo {pages} {323005} (\bibinfo {year} {2016})}\BibitemShut {NoStop}%
\bibitem [{\citenamefont {Slavnov}()}]{slavnov2018algebraic}%
  \BibitemOpen
  \bibfield  {author} {\bibinfo {author} {\bibfnamefont {N.~A.}\ \bibnamefont {Slavnov}},\ }\href@noop {} {\bibinfo {title} {{Algebraic bethe ansatz}}},\ \Eprint {https://arxiv.org/abs/1804.07350} {arXiv:1804.07350} \BibitemShut {NoStop}%
\bibitem [{\citenamefont {Gard}\ \emph {et~al.}(2020)\citenamefont {Gard}, \citenamefont {Zhu}, \citenamefont {Barron}, \citenamefont {Mayhall}, \citenamefont {Economou},\ and\ \citenamefont {Barnes}}]{gard2020efficient}%
  \BibitemOpen
  \bibfield  {author} {\bibinfo {author} {\bibfnamefont {B.~T.}\ \bibnamefont {Gard}}, \bibinfo {author} {\bibfnamefont {L.}~\bibnamefont {Zhu}}, \bibinfo {author} {\bibfnamefont {G.~S.}\ \bibnamefont {Barron}}, \bibinfo {author} {\bibfnamefont {N.~J.}\ \bibnamefont {Mayhall}}, \bibinfo {author} {\bibfnamefont {S.~E.}\ \bibnamefont {Economou}},\ and\ \bibinfo {author} {\bibfnamefont {E.}~\bibnamefont {Barnes}},\ }\bibfield  {title} {\bibinfo {title} {Efficient symmetry-preserving state preparation circuits for the variational quantum eigensolver algorithm},\ }\href {https://www.nature.com/articles/s41534-019-0240-1} {\bibfield  {journal} {\bibinfo  {journal} {npj Quantum Inf.}\ }\textbf {\bibinfo {volume} {6}},\ \bibinfo {pages} {10} (\bibinfo {year} {2020})}\BibitemShut {NoStop}%
\bibitem [{\citenamefont {Van~Dyke}\ \emph {et~al.}(2021)\citenamefont {Van~Dyke}, \citenamefont {Barron}, \citenamefont {Mayhall}, \citenamefont {Barnes},\ and\ \citenamefont {Economou}}]{van2021preparing}%
  \BibitemOpen
  \bibfield  {author} {\bibinfo {author} {\bibfnamefont {J.~S.}\ \bibnamefont {Van~Dyke}}, \bibinfo {author} {\bibfnamefont {G.~S.}\ \bibnamefont {Barron}}, \bibinfo {author} {\bibfnamefont {N.~J.}\ \bibnamefont {Mayhall}}, \bibinfo {author} {\bibfnamefont {E.}~\bibnamefont {Barnes}},\ and\ \bibinfo {author} {\bibfnamefont {S.~E.}\ \bibnamefont {Economou}},\ }\bibfield  {title} {\bibinfo {title} {{Preparing Bethe Ansatz Eigenstates on a Quantum Computer}},\ }\href {https://journals.aps.org/prxquantum/abstract/10.1103/PRXQuantum.2.040329} {\bibfield  {journal} {\bibinfo  {journal} {PRX Quantum}\ }\textbf {\bibinfo {volume} {2}},\ \bibinfo {pages} {040329} (\bibinfo {year} {2021})}\BibitemShut {NoStop}%
\bibitem [{\citenamefont {Li}\ \emph {et~al.}(2022)\citenamefont {Li}, \citenamefont {Okyay},\ and\ \citenamefont {Nepomechie}}]{li2022bethe}%
  \BibitemOpen
  \bibfield  {author} {\bibinfo {author} {\bibfnamefont {W.}~\bibnamefont {Li}}, \bibinfo {author} {\bibfnamefont {M.}~\bibnamefont {Okyay}},\ and\ \bibinfo {author} {\bibfnamefont {R.~I.}\ \bibnamefont {Nepomechie}},\ }\bibfield  {title} {\bibinfo {title} {{Bethe states on a quantum computer: success probability and correlation functions}},\ }\href {https://iopscience.iop.org/article/10.1088/1751-8121/ac8255} {\bibfield  {journal} {\bibinfo  {journal} {J. Phys. A: Math. Theor.}\ }\textbf {\bibinfo {volume} {55}},\ \bibinfo {pages} {355305} (\bibinfo {year} {2022})}\BibitemShut {NoStop}%
\bibitem [{\citenamefont {Raveh}\ and\ \citenamefont {Nepomechie}(2024{\natexlab{a}})}]{raveh2024deterministic}%
  \BibitemOpen
  \bibfield  {author} {\bibinfo {author} {\bibfnamefont {D.}~\bibnamefont {Raveh}}\ and\ \bibinfo {author} {\bibfnamefont {R.~I.}\ \bibnamefont {Nepomechie}},\ }\bibfield  {title} {\bibinfo {title} {{Deterministic Bethe state preparation}},\ }\href {https://quantum-journal.org/papers/q-2024-10-24-1510/} {\bibfield  {journal} {\bibinfo  {journal} {Quantum}\ }\textbf {\bibinfo {volume} {8}},\ \bibinfo {pages} {1510} (\bibinfo {year} {2024}{\natexlab{a}})}\BibitemShut {NoStop}%
\bibitem [{\citenamefont {Farias}\ \emph {et~al.}(2025)\citenamefont {Farias}, \citenamefont {Maciel}, \citenamefont {Camilo}, \citenamefont {Lin}, \citenamefont {Ramos-Calderer},\ and\ \citenamefont {Aolita}}]{farias2025quantum}%
  \BibitemOpen
  \bibfield  {author} {\bibinfo {author} {\bibfnamefont {R.~M.}\ \bibnamefont {Farias}}, \bibinfo {author} {\bibfnamefont {T.~O.}\ \bibnamefont {Maciel}}, \bibinfo {author} {\bibfnamefont {G.}~\bibnamefont {Camilo}}, \bibinfo {author} {\bibfnamefont {R.}~\bibnamefont {Lin}}, \bibinfo {author} {\bibfnamefont {S.}~\bibnamefont {Ramos-Calderer}},\ and\ \bibinfo {author} {\bibfnamefont {L.}~\bibnamefont {Aolita}},\ }\bibfield  {title} {\bibinfo {title} {Quantum encoder for fixed-hamming-weight subspaces},\ }\href {https://journals.aps.org/prapplied/abstract/10.1103/PhysRevApplied.23.044014} {\bibfield  {journal} {\bibinfo  {journal} {Phys. Rev. Applied}\ }\textbf {\bibinfo {volume} {23}},\ \bibinfo {pages} {044014} (\bibinfo {year} {2025})}\BibitemShut {NoStop}%
\bibitem [{\citenamefont {Sopena}\ \emph {et~al.}(2022)\citenamefont {Sopena}, \citenamefont {Gordon}, \citenamefont {Garc{\'\i}a-Mart{\'\i}n}, \citenamefont {Sierra},\ and\ \citenamefont {L{\'o}pez}}]{sopena2022algebraic}%
  \BibitemOpen
  \bibfield  {author} {\bibinfo {author} {\bibfnamefont {A.}~\bibnamefont {Sopena}}, \bibinfo {author} {\bibfnamefont {M.~H.}\ \bibnamefont {Gordon}}, \bibinfo {author} {\bibfnamefont {D.}~\bibnamefont {Garc{\'\i}a-Mart{\'\i}n}}, \bibinfo {author} {\bibfnamefont {G.}~\bibnamefont {Sierra}},\ and\ \bibinfo {author} {\bibfnamefont {E.}~\bibnamefont {L{\'o}pez}},\ }\bibfield  {title} {\bibinfo {title} {{Algebraic Bethe Circuits}},\ }\href {https://quantum-journal.org/papers/q-2022-09-08-796/} {\bibfield  {journal} {\bibinfo  {journal} {Quantum}\ }\textbf {\bibinfo {volume} {6}},\ \bibinfo {pages} {796} (\bibinfo {year} {2022})}\BibitemShut {NoStop}%
\bibitem [{\citenamefont {Ruiz}\ \emph {et~al.}(2024)\citenamefont {Ruiz}, \citenamefont {Sopena}, \citenamefont {Gordon}, \citenamefont {Sierra},\ and\ \citenamefont {L{\'o}pez}}]{ruiz2024bethe1}%
  \BibitemOpen
  \bibfield  {author} {\bibinfo {author} {\bibfnamefont {R.}~\bibnamefont {Ruiz}}, \bibinfo {author} {\bibfnamefont {A.}~\bibnamefont {Sopena}}, \bibinfo {author} {\bibfnamefont {M.~H.}\ \bibnamefont {Gordon}}, \bibinfo {author} {\bibfnamefont {G.}~\bibnamefont {Sierra}},\ and\ \bibinfo {author} {\bibfnamefont {E.}~\bibnamefont {L{\'o}pez}},\ }\bibfield  {title} {\bibinfo {title} {{The Bethe Ansatz as a Quantum Circuit}},\ }\href {https://quantum-journal.org/papers/q-2024-05-23-1356/} {\bibfield  {journal} {\bibinfo  {journal} {Quantum}\ }\textbf {\bibinfo {volume} {8}},\ \bibinfo {pages} {1356} (\bibinfo {year} {2024})}\BibitemShut {NoStop}%
\bibitem [{\citenamefont {Ruiz}\ \emph {et~al.}({\natexlab{a}})\citenamefont {Ruiz}, \citenamefont {Sopena}, \citenamefont {L{\'o}pez}, \citenamefont {Sierra},\ and\ \citenamefont {Pozsgay}}]{ruiz2024bethe2}%
  \BibitemOpen
  \bibfield  {author} {\bibinfo {author} {\bibfnamefont {R.}~\bibnamefont {Ruiz}}, \bibinfo {author} {\bibfnamefont {A.}~\bibnamefont {Sopena}}, \bibinfo {author} {\bibfnamefont {E.}~\bibnamefont {L{\'o}pez}}, \bibinfo {author} {\bibfnamefont {G.}~\bibnamefont {Sierra}},\ and\ \bibinfo {author} {\bibfnamefont {B.}~\bibnamefont {Pozsgay}},\ }\href@noop {} {\bibinfo {title} {{Bethe Ansatz, Quantum Circuits, and the F-basis}}} ({\natexlab{a}}),\ \Eprint {https://arxiv.org/abs/2411.02519} {arXiv:2411.02519} \BibitemShut {NoStop}%
\bibitem [{\citenamefont {Ruiz}\ \emph {et~al.}({\natexlab{b}})\citenamefont {Ruiz}, \citenamefont {Sopena}, \citenamefont {Pozsgay},\ and\ \citenamefont {L{\'o}pez}}]{ruiz2024efficient}%
  \BibitemOpen
  \bibfield  {author} {\bibinfo {author} {\bibfnamefont {R.}~\bibnamefont {Ruiz}}, \bibinfo {author} {\bibfnamefont {A.}~\bibnamefont {Sopena}}, \bibinfo {author} {\bibfnamefont {B.}~\bibnamefont {Pozsgay}},\ and\ \bibinfo {author} {\bibfnamefont {E.}~\bibnamefont {L{\'o}pez}},\ }\href@noop {} {\bibinfo {title} {{Efficient Eigenstate Preparation in an Integrable Model with Hilbert Space Fragmentation}}} ({\natexlab{b}}),\ \Eprint {https://arxiv.org/abs/2411.15132} {arXiv:2411.15132} \BibitemShut {NoStop}%
\bibitem [{\citenamefont {Luo}\ and\ \citenamefont {Li}()}]{luo2024circuit}%
  \BibitemOpen
  \bibfield  {author} {\bibinfo {author} {\bibfnamefont {J.}~\bibnamefont {Luo}}\ and\ \bibinfo {author} {\bibfnamefont {L.}~\bibnamefont {Li}},\ }\href@noop {} {\bibinfo {title} {{Circuit complexity of sparse quantum state preparation}}},\ \Eprint {https://arxiv.org/abs/2406.16142} {arXiv:2406.16142} \BibitemShut {NoStop}%
\bibitem [{\citenamefont {Zi}\ \emph {et~al.}()\citenamefont {Zi}, \citenamefont {Nie},\ and\ \citenamefont {Sun}}]{zi2025constant}%
  \BibitemOpen
  \bibfield  {author} {\bibinfo {author} {\bibfnamefont {W.}~\bibnamefont {Zi}}, \bibinfo {author} {\bibfnamefont {J.}~\bibnamefont {Nie}},\ and\ \bibinfo {author} {\bibfnamefont {X.}~\bibnamefont {Sun}},\ }\href@noop {} {\bibinfo {title} {{Constant-Depth Quantum Circuits for Arbitrary Quantum State Preparation via Measurement and Feedback}}},\ \Eprint {https://arxiv.org/abs/2503.16208} {arXiv:2503.16208} \BibitemShut {NoStop}%
\bibitem [{\citenamefont {Berry}\ \emph {et~al.}(2018)\citenamefont {Berry}, \citenamefont {Kieferov{\'a}}, \citenamefont {Scherer}, \citenamefont {Sanders}, \citenamefont {Low}, \citenamefont {Wiebe}, \citenamefont {Gidney},\ and\ \citenamefont {Babbush}}]{berry2018improved}%
  \BibitemOpen
  \bibfield  {author} {\bibinfo {author} {\bibfnamefont {D.~W.}\ \bibnamefont {Berry}}, \bibinfo {author} {\bibfnamefont {M.}~\bibnamefont {Kieferov{\'a}}}, \bibinfo {author} {\bibfnamefont {A.}~\bibnamefont {Scherer}}, \bibinfo {author} {\bibfnamefont {Y.~R.}\ \bibnamefont {Sanders}}, \bibinfo {author} {\bibfnamefont {G.~H.}\ \bibnamefont {Low}}, \bibinfo {author} {\bibfnamefont {N.}~\bibnamefont {Wiebe}}, \bibinfo {author} {\bibfnamefont {C.}~\bibnamefont {Gidney}},\ and\ \bibinfo {author} {\bibfnamefont {R.}~\bibnamefont {Babbush}},\ }\bibfield  {title} {\bibinfo {title} {{Improved techniques for preparing eigenstates of fermionic Hamiltonians}},\ }\href {https://www.nature.com/articles/s41534-018-0071-5} {\bibfield  {journal} {\bibinfo  {journal} {npj Quantum Inf.}\ }\textbf {\bibinfo {volume} {4}},\ \bibinfo {pages} {22} (\bibinfo {year} {2018})}\BibitemShut {NoStop}%
\bibitem [{\citenamefont {B{\"a}umer}\ and\ \citenamefont {Woerner}(2025)}]{baumer2024measurement}%
  \BibitemOpen
  \bibfield  {author} {\bibinfo {author} {\bibfnamefont {E.}~\bibnamefont {B{\"a}umer}}\ and\ \bibinfo {author} {\bibfnamefont {S.}~\bibnamefont {Woerner}},\ }\bibfield  {title} {\bibinfo {title} {Measurement-based long-range entangling gates in constant depth},\ }\href {https://doi.org/10.1103/PhysRevResearch.7.023120} {\bibfield  {journal} {\bibinfo  {journal} {Phys. Rev. Res.}\ }\textbf {\bibinfo {volume} {7}},\ \bibinfo {pages} {023120} (\bibinfo {year} {2025})}\BibitemShut {NoStop}%
\bibitem [{\citenamefont {Weigold}\ \emph {et~al.}(2020)\citenamefont {Weigold}, \citenamefont {Barzen}, \citenamefont {Leymann},\ and\ \citenamefont {Salm}}]{weigold2020data}%
  \BibitemOpen
  \bibfield  {author} {\bibinfo {author} {\bibfnamefont {M.}~\bibnamefont {Weigold}}, \bibinfo {author} {\bibfnamefont {J.}~\bibnamefont {Barzen}}, \bibinfo {author} {\bibfnamefont {F.}~\bibnamefont {Leymann}},\ and\ \bibinfo {author} {\bibfnamefont {M.}~\bibnamefont {Salm}},\ }\bibfield  {title} {\bibinfo {title} {{Data encoding patterns for quantum computing}},\ }in\ \href {https://dl.acm.org/doi/abs/10.5555/3511065.3511068} {\emph {\bibinfo {booktitle} {Proceedings of the 27th Conference on Pattern Languages of Programs}}}\ (\bibinfo {year} {2020})\ pp.\ \bibinfo {pages} {1--11}\BibitemShut {NoStop}%
\bibitem [{\citenamefont {Ramos-Calderer}\ \emph {et~al.}(2021)\citenamefont {Ramos-Calderer}, \citenamefont {P{\'e}rez-Salinas}, \citenamefont {Garc{\'\i}a-Mart{\'\i}n}, \citenamefont {Bravo-Prieto}, \citenamefont {Cortada}, \citenamefont {Planagum{\`a}},\ and\ \citenamefont {Latorre}}]{ramos2021quantum}%
  \BibitemOpen
  \bibfield  {author} {\bibinfo {author} {\bibfnamefont {S.}~\bibnamefont {Ramos-Calderer}}, \bibinfo {author} {\bibfnamefont {A.}~\bibnamefont {P{\'e}rez-Salinas}}, \bibinfo {author} {\bibfnamefont {D.}~\bibnamefont {Garc{\'\i}a-Mart{\'\i}n}}, \bibinfo {author} {\bibfnamefont {C.}~\bibnamefont {Bravo-Prieto}}, \bibinfo {author} {\bibfnamefont {J.}~\bibnamefont {Cortada}}, \bibinfo {author} {\bibfnamefont {J.}~\bibnamefont {Planagum{\`a}}},\ and\ \bibinfo {author} {\bibfnamefont {J.~I.}\ \bibnamefont {Latorre}},\ }\bibfield  {title} {\bibinfo {title} {Quantum unary approach to option pricing},\ }\href {https://journals.aps.org/pra/abstract/10.1103/PhysRevA.103.032414} {\bibfield  {journal} {\bibinfo  {journal} {Phys. Rev. A}\ }\textbf {\bibinfo {volume} {103}},\ \bibinfo {pages} {032414} (\bibinfo {year} {2021})}\BibitemShut {NoStop}%
\bibitem [{\citenamefont {Johri}\ \emph {et~al.}(2021)\citenamefont {Johri}, \citenamefont {Debnath}, \citenamefont {Mocherla}, \citenamefont {SINGK}, \citenamefont {Prakash}, \citenamefont {Kim},\ and\ \citenamefont {Kerenidis}}]{johri2021nearest}%
  \BibitemOpen
  \bibfield  {author} {\bibinfo {author} {\bibfnamefont {S.}~\bibnamefont {Johri}}, \bibinfo {author} {\bibfnamefont {S.}~\bibnamefont {Debnath}}, \bibinfo {author} {\bibfnamefont {A.}~\bibnamefont {Mocherla}}, \bibinfo {author} {\bibfnamefont {A.}~\bibnamefont {SINGK}}, \bibinfo {author} {\bibfnamefont {A.}~\bibnamefont {Prakash}}, \bibinfo {author} {\bibfnamefont {J.}~\bibnamefont {Kim}},\ and\ \bibinfo {author} {\bibfnamefont {I.}~\bibnamefont {Kerenidis}},\ }\bibfield  {title} {\bibinfo {title} {Nearest centroid classification on a trapped ion quantum computer},\ }\href {https://www.nature.com/articles/s41534-021-00456-5} {\bibfield  {journal} {\bibinfo  {journal} {npj Quantum Inf.}\ }\textbf {\bibinfo {volume} {7}},\ \bibinfo {pages} {122} (\bibinfo {year} {2021})}\BibitemShut {NoStop}%
\bibitem [{\citenamefont {Kerenidis}\ \emph {et~al.}()\citenamefont {Kerenidis}, \citenamefont {Landman},\ and\ \citenamefont {Mathur}}]{kerenidis2021classical}%
  \BibitemOpen
  \bibfield  {author} {\bibinfo {author} {\bibfnamefont {I.}~\bibnamefont {Kerenidis}}, \bibinfo {author} {\bibfnamefont {J.}~\bibnamefont {Landman}},\ and\ \bibinfo {author} {\bibfnamefont {N.}~\bibnamefont {Mathur}},\ }\href@noop {} {\bibinfo {title} {{Classical and quantum algorithms for orthogonal neural networks}}},\ \Eprint {https://arxiv.org/abs/2106.07198} {arXiv:2106.07198} \BibitemShut {NoStop}%
\bibitem [{\citenamefont {Kerenidis}\ and\ \citenamefont {Prakash}()}]{kerenidis2022quantum}%
  \BibitemOpen
  \bibfield  {author} {\bibinfo {author} {\bibfnamefont {I.}~\bibnamefont {Kerenidis}}\ and\ \bibinfo {author} {\bibfnamefont {A.}~\bibnamefont {Prakash}},\ }\href@noop {} {\bibinfo {title} {{Quantum machine learning with subspace states}}},\ \Eprint {https://arxiv.org/abs/2202.00054} {arXiv:2202.00054} \BibitemShut {NoStop}%
\bibitem [{\citenamefont {Zhang}\ \emph {et~al.}(2023)\citenamefont {Zhang}, \citenamefont {Wan}, \citenamefont {Ramos-Calderer}, \citenamefont {Zhan}, \citenamefont {Mok}, \citenamefont {Cai}, \citenamefont {Gao}, \citenamefont {Luo}, \citenamefont {Lo}, \citenamefont {Kwek} \emph {et~al.}}]{zhang2023efficient}%
  \BibitemOpen
  \bibfield  {author} {\bibinfo {author} {\bibfnamefont {H.}~\bibnamefont {Zhang}}, \bibinfo {author} {\bibfnamefont {L.}~\bibnamefont {Wan}}, \bibinfo {author} {\bibfnamefont {S.}~\bibnamefont {Ramos-Calderer}}, \bibinfo {author} {\bibfnamefont {Y.}~\bibnamefont {Zhan}}, \bibinfo {author} {\bibfnamefont {W.-K.}\ \bibnamefont {Mok}}, \bibinfo {author} {\bibfnamefont {H.}~\bibnamefont {Cai}}, \bibinfo {author} {\bibfnamefont {F.}~\bibnamefont {Gao}}, \bibinfo {author} {\bibfnamefont {X.}~\bibnamefont {Luo}}, \bibinfo {author} {\bibfnamefont {G.-Q.}\ \bibnamefont {Lo}}, \bibinfo {author} {\bibfnamefont {L.~C.}\ \bibnamefont {Kwek}}, \emph {et~al.},\ }\bibfield  {title} {\bibinfo {title} {Efficient option pricing with a unary-based photonic computing chip and generative adversarial learning},\ }\href {https://opg.optica.org/prj/fulltext.cfm?uri=prj-11-10-1703&id=540356} {\bibfield  {journal} {\bibinfo  {journal} {Photonics Research}\ }\textbf {\bibinfo {volume} {11}},\ \bibinfo {pages} {1703} (\bibinfo {year}
  {2023})}\BibitemShut {NoStop}%
\bibitem [{\citenamefont {Mao}\ \emph {et~al.}(2024)\citenamefont {Mao}, \citenamefont {Tian},\ and\ \citenamefont {Sun}}]{mao2024toward}%
  \BibitemOpen
  \bibfield  {author} {\bibinfo {author} {\bibfnamefont {R.}~\bibnamefont {Mao}}, \bibinfo {author} {\bibfnamefont {G.}~\bibnamefont {Tian}},\ and\ \bibinfo {author} {\bibfnamefont {X.}~\bibnamefont {Sun}},\ }\bibfield  {title} {\bibinfo {title} {{Toward optimal circuit size for sparse quantum State preparation}},\ }\href {https://journals.aps.org/pra/abstract/10.1103/PhysRevA.110.032439} {\bibfield  {journal} {\bibinfo  {journal} {Phys. Rev. A}\ }\textbf {\bibinfo {volume} {110}},\ \bibinfo {pages} {032439} (\bibinfo {year} {2024})}\BibitemShut {NoStop}%
\bibitem [{\citenamefont {Ramacciotti}\ \emph {et~al.}(2024)\citenamefont {Ramacciotti}, \citenamefont {Lefterovici},\ and\ \citenamefont {Rotundo}}]{ramacciotti2024simple}%
  \BibitemOpen
  \bibfield  {author} {\bibinfo {author} {\bibfnamefont {D.}~\bibnamefont {Ramacciotti}}, \bibinfo {author} {\bibfnamefont {A.~I.}\ \bibnamefont {Lefterovici}},\ and\ \bibinfo {author} {\bibfnamefont {A.~F.}\ \bibnamefont {Rotundo}},\ }\bibfield  {title} {\bibinfo {title} {Simple quantum algorithm to efficiently prepare sparse states},\ }\href {https://journals.aps.org/pra/abstract/10.1103/PhysRevA.110.032609} {\bibfield  {journal} {\bibinfo  {journal} {Phys. Rev. A}\ }\textbf {\bibinfo {volume} {110}},\ \bibinfo {pages} {032609} (\bibinfo {year} {2024})}\BibitemShut {NoStop}%
\bibitem [{\citenamefont {Vale}\ \emph {et~al.}()\citenamefont {Vale}, \citenamefont {Azevedo}, \citenamefont {Ara{\'u}jo}, \citenamefont {Araujo},\ and\ \citenamefont {da~Silva}}]{vale2023decomposition}%
  \BibitemOpen
  \bibfield  {author} {\bibinfo {author} {\bibfnamefont {R.}~\bibnamefont {Vale}}, \bibinfo {author} {\bibfnamefont {T.~M.~D.}\ \bibnamefont {Azevedo}}, \bibinfo {author} {\bibfnamefont {I.~C.~S.}\ \bibnamefont {Ara{\'u}jo}}, \bibinfo {author} {\bibfnamefont {I.~F.}\ \bibnamefont {Araujo}},\ and\ \bibinfo {author} {\bibfnamefont {A.~J.}\ \bibnamefont {da~Silva}},\ }\href@noop {} {\bibinfo {title} {{Decomposition of Multi-controlled Special Unitary Single-Qubit Gates}}},\ \Eprint {https://arxiv.org/abs/2302.06377} {arXiv:2302.06377} \BibitemShut {NoStop}%
\bibitem [{\citenamefont {Jordan}\ and\ \citenamefont {Wigner}(1928)}]{jordan1928paulische}%
  \BibitemOpen
  \bibfield  {author} {\bibinfo {author} {\bibfnamefont {P.}~\bibnamefont {Jordan}}\ and\ \bibinfo {author} {\bibfnamefont {E.}~\bibnamefont {Wigner}},\ }\bibfield  {title} {\bibinfo {title} {{{\"U}ber das Paulische {\"A}quivalenzverbot}},\ }\href@noop {} {\bibfield  {journal} {\bibinfo  {journal} {Zeitschrift f{\"u}r Physik}\ }\textbf {\bibinfo {volume} {47}},\ \bibinfo {pages} {631} (\bibinfo {year} {1928})}\BibitemShut {NoStop}%
\bibitem [{\citenamefont {Bravyi}\ and\ \citenamefont {Kitaev}(2002)}]{bravyi2002fermionic}%
  \BibitemOpen
  \bibfield  {author} {\bibinfo {author} {\bibfnamefont {S.~B.}\ \bibnamefont {Bravyi}}\ and\ \bibinfo {author} {\bibfnamefont {A.~Y.}\ \bibnamefont {Kitaev}},\ }\bibfield  {title} {\bibinfo {title} {{Fermionic Quantum Computation}},\ }\href {https://www.sciencedirect.com/science/article/abs/pii/S0003491602962548} {\bibfield  {journal} {\bibinfo  {journal} {Ann. Phys.}\ }\textbf {\bibinfo {volume} {298}},\ \bibinfo {pages} {210} (\bibinfo {year} {2002})}\BibitemShut {NoStop}%
\bibitem [{\citenamefont {Coleman}(2015)}]{coleman2015introduction}%
  \BibitemOpen
  \bibfield  {author} {\bibinfo {author} {\bibfnamefont {P.}~\bibnamefont {Coleman}},\ }\href@noop {} {\emph {\bibinfo {title} {{Introduction to many-body physics}}}}\ (\bibinfo  {publisher} {Cambridge University Press},\ \bibinfo {year} {2015})\BibitemShut {NoStop}%
\bibitem [{\citenamefont {Macridin}\ \emph {et~al.}(2022)\citenamefont {Macridin}, \citenamefont {Li}, \citenamefont {Mrenna},\ and\ \citenamefont {Spentzouris}}]{macridin2022bosonic}%
  \BibitemOpen
  \bibfield  {author} {\bibinfo {author} {\bibfnamefont {A.}~\bibnamefont {Macridin}}, \bibinfo {author} {\bibfnamefont {A.~C.~Y.}\ \bibnamefont {Li}}, \bibinfo {author} {\bibfnamefont {S.}~\bibnamefont {Mrenna}},\ and\ \bibinfo {author} {\bibfnamefont {P.}~\bibnamefont {Spentzouris}},\ }\bibfield  {title} {\bibinfo {title} {{Bosonic field digitization for quantum computers}},\ }\href {https://journals.aps.org/pra/abstract/10.1103/PhysRevA.105.052405} {\bibfield  {journal} {\bibinfo  {journal} {Phys. Rev. A}\ }\textbf {\bibinfo {volume} {105}},\ \bibinfo {pages} {052405} (\bibinfo {year} {2022})}\BibitemShut {NoStop}%
\bibitem [{\citenamefont {Ortiz}\ \emph {et~al.}(2001)\citenamefont {Ortiz}, \citenamefont {Gubernatis}, \citenamefont {Knill},\ and\ \citenamefont {Laflamme}}]{ortiz2001quantum}%
  \BibitemOpen
  \bibfield  {author} {\bibinfo {author} {\bibfnamefont {G.}~\bibnamefont {Ortiz}}, \bibinfo {author} {\bibfnamefont {J.~E.}\ \bibnamefont {Gubernatis}}, \bibinfo {author} {\bibfnamefont {E.}~\bibnamefont {Knill}},\ and\ \bibinfo {author} {\bibfnamefont {R.}~\bibnamefont {Laflamme}},\ }\bibfield  {title} {\bibinfo {title} {{Quantum algorithms for fermionic simulations}},\ }\href {https://journals.aps.org/pra/abstract/10.1103/PhysRevA.64.022319} {\bibfield  {journal} {\bibinfo  {journal} {Phys. Rev. A}\ }\textbf {\bibinfo {volume} {64}},\ \bibinfo {pages} {022319} (\bibinfo {year} {2001})}\BibitemShut {NoStop}%
\bibitem [{\citenamefont {Wecker}\ \emph {et~al.}(2015)\citenamefont {Wecker}, \citenamefont {Hastings}, \citenamefont {Wiebe}, \citenamefont {Clark}, \citenamefont {Nayak},\ and\ \citenamefont {Troyer}}]{wecker2015solving}%
  \BibitemOpen
  \bibfield  {author} {\bibinfo {author} {\bibfnamefont {D.}~\bibnamefont {Wecker}}, \bibinfo {author} {\bibfnamefont {M.~B.}\ \bibnamefont {Hastings}}, \bibinfo {author} {\bibfnamefont {N.}~\bibnamefont {Wiebe}}, \bibinfo {author} {\bibfnamefont {B.~K.}\ \bibnamefont {Clark}}, \bibinfo {author} {\bibfnamefont {C.}~\bibnamefont {Nayak}},\ and\ \bibinfo {author} {\bibfnamefont {M.}~\bibnamefont {Troyer}},\ }\bibfield  {title} {\bibinfo {title} {Solving strongly correlated electron models on a quantum computer},\ }\href {https://journals.aps.org/pra/abstract/10.1103/PhysRevA.92.062318} {\bibfield  {journal} {\bibinfo  {journal} {Phys. Rev. A}\ }\textbf {\bibinfo {volume} {92}},\ \bibinfo {pages} {062318} (\bibinfo {year} {2015})}\BibitemShut {NoStop}%
\bibitem [{\citenamefont {Kivlichan}\ \emph {et~al.}(2018)\citenamefont {Kivlichan}, \citenamefont {McClean}, \citenamefont {Wiebe}, \citenamefont {Gidney}, \citenamefont {Aspuru-Guzik}, \citenamefont {Chan},\ and\ \citenamefont {Babbush}}]{kivlichan2018quantum}%
  \BibitemOpen
  \bibfield  {author} {\bibinfo {author} {\bibfnamefont {I.~D.}\ \bibnamefont {Kivlichan}}, \bibinfo {author} {\bibfnamefont {J.}~\bibnamefont {McClean}}, \bibinfo {author} {\bibfnamefont {N.}~\bibnamefont {Wiebe}}, \bibinfo {author} {\bibfnamefont {C.}~\bibnamefont {Gidney}}, \bibinfo {author} {\bibfnamefont {A.}~\bibnamefont {Aspuru-Guzik}}, \bibinfo {author} {\bibfnamefont {G.~K.-L.}\ \bibnamefont {Chan}},\ and\ \bibinfo {author} {\bibfnamefont {R.}~\bibnamefont {Babbush}},\ }\bibfield  {title} {\bibinfo {title} {{Quantum Simulation of Electronic Structure with Linear Depth and Connectivity}},\ }\href {https://journals.aps.org/prl/abstract/10.1103/PhysRevLett.120.110501} {\bibfield  {journal} {\bibinfo  {journal} {Phys. Rev. Lett.}\ }\textbf {\bibinfo {volume} {120}},\ \bibinfo {pages} {110501} (\bibinfo {year} {2018})}\BibitemShut {NoStop}%
\bibitem [{\citenamefont {Jiang}\ \emph {et~al.}(2018)\citenamefont {Jiang}, \citenamefont {Sung}, \citenamefont {Kechedzhi}, \citenamefont {Smelyanskiy},\ and\ \citenamefont {Boixo}}]{jiang2018quantum}%
  \BibitemOpen
  \bibfield  {author} {\bibinfo {author} {\bibfnamefont {Z.}~\bibnamefont {Jiang}}, \bibinfo {author} {\bibfnamefont {K.~J.}\ \bibnamefont {Sung}}, \bibinfo {author} {\bibfnamefont {K.}~\bibnamefont {Kechedzhi}}, \bibinfo {author} {\bibfnamefont {V.~N.}\ \bibnamefont {Smelyanskiy}},\ and\ \bibinfo {author} {\bibfnamefont {S.}~\bibnamefont {Boixo}},\ }\bibfield  {title} {\bibinfo {title} {{Quantum Algorithms to Simulate Many-Body Physics of Correlated Fermions}},\ }\href {https://journals.aps.org/prapplied/abstract/10.1103/PhysRevApplied.9.044036} {\bibfield  {journal} {\bibinfo  {journal} {Phys. Rev. Applied}\ }\textbf {\bibinfo {volume} {9}},\ \bibinfo {pages} {044036} (\bibinfo {year} {2018})}\BibitemShut {NoStop}%
\bibitem [{\citenamefont {Chee}\ \emph {et~al.}(2023)\citenamefont {Chee}, \citenamefont {Leykam}, \citenamefont {Mak},\ and\ \citenamefont {Angelakis}}]{chee2023shallow}%
  \BibitemOpen
  \bibfield  {author} {\bibinfo {author} {\bibfnamefont {C.~H.}\ \bibnamefont {Chee}}, \bibinfo {author} {\bibfnamefont {D.}~\bibnamefont {Leykam}}, \bibinfo {author} {\bibfnamefont {A.~M.}\ \bibnamefont {Mak}},\ and\ \bibinfo {author} {\bibfnamefont {D.~G.}\ \bibnamefont {Angelakis}},\ }\bibfield  {title} {\bibinfo {title} {Shallow quantum circuits for efficient preparation of slater determinants and correlated states on a quantum computer},\ }\href {https://journals.aps.org/pra/abstract/10.1103/PhysRevA.108.022416} {\bibfield  {journal} {\bibinfo  {journal} {Phys. Rev. A}\ }\textbf {\bibinfo {volume} {108}},\ \bibinfo {pages} {022416} (\bibinfo {year} {2023})}\BibitemShut {NoStop}%
\bibitem [{\citenamefont {Abrams}\ and\ \citenamefont {Lloyd}(1997)}]{abrams1997simulation}%
  \BibitemOpen
  \bibfield  {author} {\bibinfo {author} {\bibfnamefont {D.~S.}\ \bibnamefont {Abrams}}\ and\ \bibinfo {author} {\bibfnamefont {S.}~\bibnamefont {Lloyd}},\ }\bibfield  {title} {\bibinfo {title} {{Simulation of Many-Body Fermi Systems on a Universal Quantum Computer}},\ }\href {https://journals.aps.org/prl/abstract/10.1103/PhysRevLett.79.2586} {\bibfield  {journal} {\bibinfo  {journal} {Phys. Rev. Lett.}\ }\textbf {\bibinfo {volume} {79}},\ \bibinfo {pages} {2586} (\bibinfo {year} {1997})}\BibitemShut {NoStop}%
\bibitem [{\citenamefont {Su}\ \emph {et~al.}(2021)\citenamefont {Su}, \citenamefont {Berry}, \citenamefont {Wiebe}, \citenamefont {Rubin},\ and\ \citenamefont {Babbush}}]{su2021fault}%
  \BibitemOpen
  \bibfield  {author} {\bibinfo {author} {\bibfnamefont {Y.}~\bibnamefont {Su}}, \bibinfo {author} {\bibfnamefont {D.~W.}\ \bibnamefont {Berry}}, \bibinfo {author} {\bibfnamefont {N.}~\bibnamefont {Wiebe}}, \bibinfo {author} {\bibfnamefont {N.}~\bibnamefont {Rubin}},\ and\ \bibinfo {author} {\bibfnamefont {R.}~\bibnamefont {Babbush}},\ }\bibfield  {title} {\bibinfo {title} {{Fault-Tolerant Quantum Simulations of Chemistry in First Quantization}},\ }\href {https://journals.aps.org/prxquantum/abstract/10.1103/PRXQuantum.2.040332} {\bibfield  {journal} {\bibinfo  {journal} {PRX Quantum}\ }\textbf {\bibinfo {volume} {2}},\ \bibinfo {pages} {040332} (\bibinfo {year} {2021})}\BibitemShut {NoStop}%
\bibitem [{\citenamefont {Babbush}\ \emph {et~al.}(2023)\citenamefont {Babbush}, \citenamefont {Huggins}, \citenamefont {Berry}, \citenamefont {Ung}, \citenamefont {Zhao}, \citenamefont {Reichman}, \citenamefont {Neven}, \citenamefont {Baczewski},\ and\ \citenamefont {Lee}}]{babbush2023quantum}%
  \BibitemOpen
  \bibfield  {author} {\bibinfo {author} {\bibfnamefont {R.}~\bibnamefont {Babbush}}, \bibinfo {author} {\bibfnamefont {W.~J.}\ \bibnamefont {Huggins}}, \bibinfo {author} {\bibfnamefont {D.~W.}\ \bibnamefont {Berry}}, \bibinfo {author} {\bibfnamefont {S.~F.}\ \bibnamefont {Ung}}, \bibinfo {author} {\bibfnamefont {A.}~\bibnamefont {Zhao}}, \bibinfo {author} {\bibfnamefont {D.~R.}\ \bibnamefont {Reichman}}, \bibinfo {author} {\bibfnamefont {H.}~\bibnamefont {Neven}}, \bibinfo {author} {\bibfnamefont {A.~D.}\ \bibnamefont {Baczewski}},\ and\ \bibinfo {author} {\bibfnamefont {J.}~\bibnamefont {Lee}},\ }\bibfield  {title} {\bibinfo {title} {{Quantum simulation of exact electron dynamics can be more efficient than classical mean-field methods}},\ }\href {https://www.nature.com/articles/s41467-023-39024-0} {\bibfield  {journal} {\bibinfo  {journal} {Nat. Commun.}\ }\textbf {\bibinfo {volume} {14}},\ \bibinfo {pages} {4058} (\bibinfo {year} {2023})}\BibitemShut {NoStop}%
\bibitem [{\citenamefont {Huggins}\ \emph {et~al.}(2025)\citenamefont {Huggins}, \citenamefont {Leimkuhler}, \citenamefont {Stetina},\ and\ \citenamefont {Whaley}}]{huggins2024efficient}%
  \BibitemOpen
  \bibfield  {author} {\bibinfo {author} {\bibfnamefont {W.~J.}\ \bibnamefont {Huggins}}, \bibinfo {author} {\bibfnamefont {O.}~\bibnamefont {Leimkuhler}}, \bibinfo {author} {\bibfnamefont {T.~F.}\ \bibnamefont {Stetina}},\ and\ \bibinfo {author} {\bibfnamefont {K.~B.}\ \bibnamefont {Whaley}},\ }\bibfield  {title} {\bibinfo {title} {Efficient state preparation for the quantum simulation of molecules in first quantiza- tion},\ }\href {https://doi.org/10.1103/PRXQuantum.6.020319} {\bibfield  {journal} {\bibinfo  {journal} {PRX Quantum}\ }\textbf {\bibinfo {volume} {6}},\ \bibinfo {pages} {020319} (\bibinfo {year} {2025})}\BibitemShut {NoStop}%
\bibitem [{\citenamefont {Veis}\ and\ \citenamefont {Pittner}(2010)}]{veis2010quantum}%
  \BibitemOpen
  \bibfield  {author} {\bibinfo {author} {\bibfnamefont {L.}~\bibnamefont {Veis}}\ and\ \bibinfo {author} {\bibfnamefont {J.}~\bibnamefont {Pittner}},\ }\bibfield  {title} {\bibinfo {title} {{Quantum computing applied to calculations of molecular energies: CH2 benchmark}},\ }\href {https://pubs.aip.org/aip/jcp/article-abstract/133/19/194106/188444/Quantum-computing-applied-to-calculations-of?redirectedFrom=fulltext} {\bibfield  {journal} {\bibinfo  {journal} {J. Chem. Phys.}\ }\textbf {\bibinfo {volume} {133}} (\bibinfo {year} {2010})}\BibitemShut {NoStop}%
\bibitem [{\citenamefont {Tubman}\ \emph {et~al.}()\citenamefont {Tubman}, \citenamefont {Mejuto-Zaera}, \citenamefont {Epstein}, \citenamefont {Hait}, \citenamefont {Levine}, \citenamefont {Huggins}, \citenamefont {Jiang}, \citenamefont {McClean}, \citenamefont {Babbush}, \citenamefont {Head-Gordon} \emph {et~al.}}]{tubman2018postponing}%
  \BibitemOpen
  \bibfield  {author} {\bibinfo {author} {\bibfnamefont {N.~M.}\ \bibnamefont {Tubman}}, \bibinfo {author} {\bibfnamefont {C.}~\bibnamefont {Mejuto-Zaera}}, \bibinfo {author} {\bibfnamefont {J.~M.}\ \bibnamefont {Epstein}}, \bibinfo {author} {\bibfnamefont {D.}~\bibnamefont {Hait}}, \bibinfo {author} {\bibfnamefont {D.~S.}\ \bibnamefont {Levine}}, \bibinfo {author} {\bibfnamefont {W.}~\bibnamefont {Huggins}}, \bibinfo {author} {\bibfnamefont {Z.}~\bibnamefont {Jiang}}, \bibinfo {author} {\bibfnamefont {J.~R.}\ \bibnamefont {McClean}}, \bibinfo {author} {\bibfnamefont {R.}~\bibnamefont {Babbush}}, \bibinfo {author} {\bibfnamefont {M.}~\bibnamefont {Head-Gordon}}, \emph {et~al.},\ }\href@noop {} {\bibinfo {title} {{Postponing the orthogonality catastrophe: efficient state preparation for electronic structure simulations on quantum devices}}},\ \Eprint {https://arxiv.org/abs/1809.05523} {arXiv:1809.05523} \BibitemShut {NoStop}%
\bibitem [{\citenamefont {Wang}\ \emph {et~al.}(2009)\citenamefont {Wang}, \citenamefont {Ashhab},\ and\ \citenamefont {Nori}}]{wang2009efficient}%
  \BibitemOpen
  \bibfield  {author} {\bibinfo {author} {\bibfnamefont {H.}~\bibnamefont {Wang}}, \bibinfo {author} {\bibfnamefont {S.}~\bibnamefont {Ashhab}},\ and\ \bibinfo {author} {\bibfnamefont {F.}~\bibnamefont {Nori}},\ }\bibfield  {title} {\bibinfo {title} {{Efficient quantum algorithm for preparing molecular-system-like states on a quantum computer}},\ }\href {https://journals.aps.org/pra/abstract/10.1103/PhysRevA.79.042335} {\bibfield  {journal} {\bibinfo  {journal} {Phys. Rev. A}\ }\textbf {\bibinfo {volume} {79}},\ \bibinfo {pages} {042335} (\bibinfo {year} {2009})}\BibitemShut {NoStop}%
\bibitem [{\citenamefont {Helgaker}\ \emph {et~al.}(2013)\citenamefont {Helgaker}, \citenamefont {Jorgensen},\ and\ \citenamefont {Olsen}}]{helgaker2013molecular}%
  \BibitemOpen
  \bibfield  {author} {\bibinfo {author} {\bibfnamefont {T.}~\bibnamefont {Helgaker}}, \bibinfo {author} {\bibfnamefont {P.}~\bibnamefont {Jorgensen}},\ and\ \bibinfo {author} {\bibfnamefont {J.}~\bibnamefont {Olsen}},\ }\href@noop {} {\emph {\bibinfo {title} {{Molecular electronic-structure theory}}}}\ (\bibinfo  {publisher} {John Wiley \& Sons},\ \bibinfo {year} {2013})\BibitemShut {NoStop}%
\bibitem [{\citenamefont {Georges}\ \emph {et~al.}(2025)\citenamefont {Georges}, \citenamefont {Bothe}, \citenamefont {S{\"u}nderhauf}, \citenamefont {Berntson}, \citenamefont {Izs{\'a}k},\ and\ \citenamefont {Ivanov}}]{georges2024quantum}%
  \BibitemOpen
  \bibfield  {author} {\bibinfo {author} {\bibfnamefont {T.~N.}\ \bibnamefont {Georges}}, \bibinfo {author} {\bibfnamefont {M.}~\bibnamefont {Bothe}}, \bibinfo {author} {\bibfnamefont {C.}~\bibnamefont {S{\"u}nderhauf}}, \bibinfo {author} {\bibfnamefont {B.~K.}\ \bibnamefont {Berntson}}, \bibinfo {author} {\bibfnamefont {R.}~\bibnamefont {Izs{\'a}k}},\ and\ \bibinfo {author} {\bibfnamefont {A.~V.}\ \bibnamefont {Ivanov}},\ }\bibfield  {title} {\bibinfo {title} {Quantum simulations of chemistry in first quantization with any basis set},\ }\href {https://www.nature.com/articles/s41534-025-00987-1} {\bibfield  {journal} {\bibinfo  {journal} {npj Quantum Inf.}\ }\textbf {\bibinfo {volume} {11}} (\bibinfo {year} {2025})}\BibitemShut {NoStop}%
\bibitem [{\citenamefont {Babbush}\ \emph {et~al.}(2018)\citenamefont {Babbush}, \citenamefont {Wiebe}, \citenamefont {McClean}, \citenamefont {McClain}, \citenamefont {Neven},\ and\ \citenamefont {Chan}}]{babbush2018low}%
  \BibitemOpen
  \bibfield  {author} {\bibinfo {author} {\bibfnamefont {R.}~\bibnamefont {Babbush}}, \bibinfo {author} {\bibfnamefont {N.}~\bibnamefont {Wiebe}}, \bibinfo {author} {\bibfnamefont {J.}~\bibnamefont {McClean}}, \bibinfo {author} {\bibfnamefont {J.}~\bibnamefont {McClain}}, \bibinfo {author} {\bibfnamefont {H.}~\bibnamefont {Neven}},\ and\ \bibinfo {author} {\bibfnamefont {G.~K.-L.}\ \bibnamefont {Chan}},\ }\bibfield  {title} {\bibinfo {title} {{Low-Depth Quantum Simulation of Materials}},\ }\href {https://journals.aps.org/prx/abstract/10.1103/PhysRevX.8.011044} {\bibfield  {journal} {\bibinfo  {journal} {Phys. Rev. X}\ }\textbf {\bibinfo {volume} {8}},\ \bibinfo {pages} {011044} (\bibinfo {year} {2018})}\BibitemShut {NoStop}%
\bibitem [{\citenamefont {Babbush}\ \emph {et~al.}(2019)\citenamefont {Babbush}, \citenamefont {Berry}, \citenamefont {McClean},\ and\ \citenamefont {Neven}}]{babbush2019quantum}%
  \BibitemOpen
  \bibfield  {author} {\bibinfo {author} {\bibfnamefont {R.}~\bibnamefont {Babbush}}, \bibinfo {author} {\bibfnamefont {D.~W.}\ \bibnamefont {Berry}}, \bibinfo {author} {\bibfnamefont {J.~R.}\ \bibnamefont {McClean}},\ and\ \bibinfo {author} {\bibfnamefont {H.}~\bibnamefont {Neven}},\ }\bibfield  {title} {\bibinfo {title} {{Quantum simulation of chemistry with sublinear scaling in basis size}},\ }\href {https://www.nature.com/articles/s41534-019-0199-y} {\bibfield  {journal} {\bibinfo  {journal} {npj Quantum Inf.}\ }\textbf {\bibinfo {volume} {5}},\ \bibinfo {pages} {92} (\bibinfo {year} {2019})}\BibitemShut {NoStop}%
\bibitem [{\citenamefont {Yang}\ and\ \citenamefont {Yang}(1966)}]{yang1966one}%
  \BibitemOpen
  \bibfield  {author} {\bibinfo {author} {\bibfnamefont {C.~N.}\ \bibnamefont {Yang}}\ and\ \bibinfo {author} {\bibfnamefont {C.~P.}\ \bibnamefont {Yang}},\ }\bibfield  {title} {\bibinfo {title} {{One-Dimensional Chain of Anisotropic Spin-Spin Interactions. I. Proof of Bethe's Hypothesis for Ground State in a Finite System}},\ }\href {https://journals.aps.org/pr/abstract/10.1103/PhysRev.150.321} {\bibfield  {journal} {\bibinfo  {journal} {Phys. Rev.}\ }\textbf {\bibinfo {volume} {150}},\ \bibinfo {pages} {327} (\bibinfo {year} {1966})}\BibitemShut {NoStop}%
\bibitem [{\citenamefont {Lieb}\ and\ \citenamefont {Liniger}(1963)}]{lieb1963exact}%
  \BibitemOpen
  \bibfield  {author} {\bibinfo {author} {\bibfnamefont {E.~H.}\ \bibnamefont {Lieb}}\ and\ \bibinfo {author} {\bibfnamefont {W.}~\bibnamefont {Liniger}},\ }\bibfield  {title} {\bibinfo {title} {{Exact Analysis of an Interacting Bose Gas. I. The General Solution and the Ground State}},\ }\href {https://journals.aps.org/pr/abstract/10.1103/PhysRev.130.1605} {\bibfield  {journal} {\bibinfo  {journal} {Phys. Rev.}\ }\textbf {\bibinfo {volume} {130}},\ \bibinfo {pages} {1605} (\bibinfo {year} {1963})}\BibitemShut {NoStop}%
\bibitem [{\citenamefont {Andrei}\ \emph {et~al.}(1983)\citenamefont {Andrei}, \citenamefont {Furuya},\ and\ \citenamefont {Lowenstein}}]{andrei1983solution}%
  \BibitemOpen
  \bibfield  {author} {\bibinfo {author} {\bibfnamefont {N.}~\bibnamefont {Andrei}}, \bibinfo {author} {\bibfnamefont {K.}~\bibnamefont {Furuya}},\ and\ \bibinfo {author} {\bibfnamefont {J.~H.}\ \bibnamefont {Lowenstein}},\ }\bibfield  {title} {\bibinfo {title} {{Solution of the Kondo problem}},\ }\href {https://journals.aps.org/rmp/abstract/10.1103/RevModPhys.55.331} {\bibfield  {journal} {\bibinfo  {journal} {Rev. Mod. Phys.}\ }\textbf {\bibinfo {volume} {55}},\ \bibinfo {pages} {331} (\bibinfo {year} {1983})}\BibitemShut {NoStop}%
\bibitem [{\citenamefont {Lieb}\ and\ \citenamefont {Wu}(1968)}]{lieb1968absence}%
  \BibitemOpen
  \bibfield  {author} {\bibinfo {author} {\bibfnamefont {E.~H.}\ \bibnamefont {Lieb}}\ and\ \bibinfo {author} {\bibfnamefont {F.~Y.}\ \bibnamefont {Wu}},\ }\bibfield  {title} {\bibinfo {title} {{Absence of Mott Transition in an Exact Solution of the Short-Range, One-Band Model in One Dimension}},\ }\href {https://journals.aps.org/prl/abstract/10.1103/PhysRevLett.20.1445} {\bibfield  {journal} {\bibinfo  {journal} {Phys. Rev. Lett.}\ }\textbf {\bibinfo {volume} {20}},\ \bibinfo {pages} {1445} (\bibinfo {year} {1968})}\BibitemShut {NoStop}%
\bibitem [{\citenamefont {Karbach}\ and\ \citenamefont {M{\"u}ller}()}]{karbach1998introduction}%
  \BibitemOpen
  \bibfield  {author} {\bibinfo {author} {\bibfnamefont {M.}~\bibnamefont {Karbach}}\ and\ \bibinfo {author} {\bibfnamefont {G.}~\bibnamefont {M{\"u}ller}},\ }\href@noop {} {\bibinfo {title} {{Introduction to the Bethe Ansatz I}}},\ \Eprint {https://arxiv.org/abs/cond-mat/9809162} {arXiv:cond-mat/9809162} \BibitemShut {NoStop}%
\bibitem [{\citenamefont {Faddeev}(1995)}]{faddeev1995algebraic}%
  \BibitemOpen
  \bibfield  {author} {\bibinfo {author} {\bibfnamefont {L.~D.}\ \bibnamefont {Faddeev}},\ }\bibfield  {title} {\bibinfo {title} {{Algebraic aspects of the Bethe ansatz}},\ }\href {https://www.worldscientific.com/doi/abs/10.1142/S0217751X95000905} {\bibfield  {journal} {\bibinfo  {journal} {Int. J. Mod. Phys. A}\ }\textbf {\bibinfo {volume} {10}},\ \bibinfo {pages} {1845} (\bibinfo {year} {1995})}\BibitemShut {NoStop}%
\bibitem [{\citenamefont {Takahashi}(1999)}]{takahashi1999thermodynamics}%
  \BibitemOpen
  \bibfield  {author} {\bibinfo {author} {\bibfnamefont {M.}~\bibnamefont {Takahashi}},\ }\href@noop {} {\emph {\bibinfo {title} {Thermodynamics of One-Dimensional Solvable Models}}}\ (\bibinfo  {publisher} {Cambridge University Press},\ \bibinfo {year} {1999})\BibitemShut {NoStop}%
\bibitem [{\citenamefont {Levkovich-Maslyuk}(2016)}]{levkovich2016bethe}%
  \BibitemOpen
  \bibfield  {author} {\bibinfo {author} {\bibfnamefont {F.}~\bibnamefont {Levkovich-Maslyuk}},\ }\bibfield  {title} {\bibinfo {title} {{The Bethe ansatz}},\ }\href {https://iopscience.iop.org/article/10.1088/1751-8113/49/32/323004/meta?casa_token=eOcRsU2cA_YAAAAA:4KrEmIo-RJ-PmdiHDt3Y6Lg9qhiZLfAxKP-NtqtV7h0Sv2Pt3o0AMDQKD6nJ2q1H6m0lx77krnh8705IhymN0CR1pqUUkg} {\bibfield  {journal} {\bibinfo  {journal} {J. Phys. A: Math. Theor.}\ }\textbf {\bibinfo {volume} {49}},\ \bibinfo {pages} {323004} (\bibinfo {year} {2016})}\BibitemShut {NoStop}%
\bibitem [{\citenamefont {Nepomechie}()}]{nepomechie2020bethe}%
  \BibitemOpen
  \bibfield  {author} {\bibinfo {author} {\bibfnamefont {R.~I.}\ \bibnamefont {Nepomechie}},\ }\href@noop {} {\bibinfo {title} {{Bethe ansatz on a quantum computer?}}},\ \Eprint {https://arxiv.org/abs/2010.01609} {arXiv:2010.01609} \BibitemShut {NoStop}%
\bibitem [{\citenamefont {Van~Dyke}\ \emph {et~al.}(2022)\citenamefont {Van~Dyke}, \citenamefont {Barnes}, \citenamefont {Economou},\ and\ \citenamefont {Nepomechie}}]{van2022preparing}%
  \BibitemOpen
  \bibfield  {author} {\bibinfo {author} {\bibfnamefont {J.~S.}\ \bibnamefont {Van~Dyke}}, \bibinfo {author} {\bibfnamefont {E.}~\bibnamefont {Barnes}}, \bibinfo {author} {\bibfnamefont {S.~E.}\ \bibnamefont {Economou}},\ and\ \bibinfo {author} {\bibfnamefont {R.~I.}\ \bibnamefont {Nepomechie}},\ }\bibfield  {title} {\bibinfo {title} {{Preparing exact eigenstates of the open XXZ chain on a quantum computer}},\ }\href {https://iopscience.iop.org/article/10.1088/1751-8121/ac4640} {\bibfield  {journal} {\bibinfo  {journal} {J. Phys. A: Math. Theor.}\ }\textbf {\bibinfo {volume} {55}},\ \bibinfo {pages} {055301} (\bibinfo {year} {2022})}\BibitemShut {NoStop}%
\bibitem [{\citenamefont {Raveh}\ and\ \citenamefont {Nepomechie}(2024{\natexlab{b}})}]{raveh2024estimating}%
  \BibitemOpen
  \bibfield  {author} {\bibinfo {author} {\bibfnamefont {D.}~\bibnamefont {Raveh}}\ and\ \bibinfo {author} {\bibfnamefont {R.~I.}\ \bibnamefont {Nepomechie}},\ }\bibfield  {title} {\bibinfo {title} {Estimating bethe roots with vqe},\ }\href {https://iopscience.iop.org/article/10.1088/1751-8121/ad6db2/meta} {\bibfield  {journal} {\bibinfo  {journal} {J. Phys. A: Math. Theor.}\ }\textbf {\bibinfo {volume} {57}},\ \bibinfo {pages} {355303} (\bibinfo {year} {2024}{\natexlab{b}})}\BibitemShut {NoStop}%
\bibitem [{\citenamefont {Lutz}\ \emph {et~al.}()\citenamefont {Lutz}, \citenamefont {Piroli}, \citenamefont {Styliaris},\ and\ \citenamefont {Cirac}}]{lutz2025adiabatic}%
  \BibitemOpen
  \bibfield  {author} {\bibinfo {author} {\bibfnamefont {M.}~\bibnamefont {Lutz}}, \bibinfo {author} {\bibfnamefont {L.}~\bibnamefont {Piroli}}, \bibinfo {author} {\bibfnamefont {G.}~\bibnamefont {Styliaris}},\ and\ \bibinfo {author} {\bibfnamefont {J.~I.}\ \bibnamefont {Cirac}},\ }\href@noop {} {\bibinfo {title} {{Adiabatic quantum state preparation in integrable models}}},\ \Eprint {https://arxiv.org/abs/2503.21741} {arXiv:2503.21741} \BibitemShut {NoStop}%
\bibitem [{\citenamefont {B{\"a}rtschi}\ and\ \citenamefont {Eidenbenz}(2019)}]{bartschi2019deterministic}%
  \BibitemOpen
  \bibfield  {author} {\bibinfo {author} {\bibfnamefont {A.}~\bibnamefont {B{\"a}rtschi}}\ and\ \bibinfo {author} {\bibfnamefont {S.}~\bibnamefont {Eidenbenz}},\ }\bibfield  {title} {\bibinfo {title} {{Deterministic preparation of Dicke states}},\ }in\ \href {https://link.springer.com/chapter/10.1007/978-3-030-25027-0_9} {\emph {\bibinfo {booktitle} {Fundamentals of Computation Theory}}}\ (\bibinfo {organization} {Springer},\ \bibinfo {year} {2019})\ pp.\ \bibinfo {pages} {126--139}\BibitemShut {NoStop}%
\bibitem [{\citenamefont {{IBM Quantum}}()}]{ibm_quantum_resources}%
  \BibitemOpen
  \bibfield  {author} {\bibinfo {author} {\bibnamefont {{IBM Quantum}}},\ }\href@noop {} {\bibinfo {title} {Ibm quantum resources}},\ \bibinfo {howpublished} {\url{https://quantum.ibm.com/services/resources}},\ \bibinfo {note} {accessed: April 15, 2025}\BibitemShut {NoStop}%
\bibitem [{\citenamefont {Kaldenbach}\ and\ \citenamefont {Heller}(2024)}]{kaldenbach2024mapping}%
  \BibitemOpen
  \bibfield  {author} {\bibinfo {author} {\bibfnamefont {T.~N.}\ \bibnamefont {Kaldenbach}}\ and\ \bibinfo {author} {\bibfnamefont {M.}~\bibnamefont {Heller}},\ }\bibfield  {title} {\bibinfo {title} {{Mapping quantum circuits to shallow-depth measurement patterns based on graph states}},\ }\href {https://iopscience.iop.org/article/10.1088/2058-9565/ad802b/meta} {\bibfield  {journal} {\bibinfo  {journal} {Quantum Sci. Technol.}\ }\textbf {\bibinfo {volume} {10}},\ \bibinfo {pages} {015010} (\bibinfo {year} {2024})}\BibitemShut {NoStop}%
\bibitem [{\citenamefont {Sun}\ \emph {et~al.}(2023)\citenamefont {Sun}, \citenamefont {Tian}, \citenamefont {Yang}, \citenamefont {Yuan},\ and\ \citenamefont {Zhang}}]{sun2023asymptotically}%
  \BibitemOpen
  \bibfield  {author} {\bibinfo {author} {\bibfnamefont {X.}~\bibnamefont {Sun}}, \bibinfo {author} {\bibfnamefont {G.}~\bibnamefont {Tian}}, \bibinfo {author} {\bibfnamefont {S.}~\bibnamefont {Yang}}, \bibinfo {author} {\bibfnamefont {P.}~\bibnamefont {Yuan}},\ and\ \bibinfo {author} {\bibfnamefont {S.}~\bibnamefont {Zhang}},\ }\bibfield  {title} {\bibinfo {title} {Asymptotically optimal circuit depth for quantum state preparation and general unitary synthesis},\ }\href {https://ieeexplore.ieee.org/abstract/document/10044235?casa_token=gKVavFxw_PIAAAAA:Q1ni_-tCdMZFzZhAnS3gpe7I_gFsXLMiu4sUPWKLc8gEplPWGk9NekXB7UEVDeTQ_-qbD1JN-g} {\bibfield  {journal} {\bibinfo  {journal} {IEEE Transactions on Computer-Aided Design of Integrated Circuits and Systems}\ }\textbf {\bibinfo {volume} {42}},\ \bibinfo {pages} {3301} (\bibinfo {year} {2023})}\BibitemShut {NoStop}%
\end{thebibliography}%

\end{document}